\documentclass[11pt]{article}

\usepackage[T1]{fontenc}
\usepackage[numbers]{natbib}
\usepackage[margin=1in]{geometry}
\usepackage{subcaption} 
\usepackage{amsmath}
\usepackage{amssymb}
\usepackage{xspace}
\usepackage{color}           
\usepackage{booktabs}      
\usepackage{threeparttable}  
\usepackage{tikz}            
\usepackage[ruled,vlined,linesnumbered]{algorithm2e}     
\usepackage{url}             
\usepackage{hyperref}
\usepackage{amsthm}
\usepackage[capitalize]{cleveref}
\crefname{appendix}{Appendix}{Appendices}
\usepackage{lineno}
\usepackage{textcomp}
\usepackage[utf8]{inputenc}
\usepackage{complexity}
\usepackage{dsfont}
\usepackage{authblk}
\allowdisplaybreaks[0]

\newtheorem{theorem}{Theorem}
\newtheorem{lemma}{Lemma}

\newtheorem{definition}{Definition}
\newtheorem{remark}{Remark}
\newtheorem{example}{Example}
\newcommand{\BibTeX}{B\kern-.05em{\sc i\kern-.025em b}\kern-.08em\TeX}

\newcommand{\sentence}{\Gamma}

\newcommand{\fotwoformula}{\psi}

\newcommand{\weight}{w}

\newcommand{\negweight}{\bar{w}}

\newcommand{\wfomc}{WFOMC}

\newcommand{\veck}{\mathbf{k}}
\newcommand{\vecK}{\mathbf{K}}
\newcommand{\vecU}{\mathbf{U}}

\newcommand{\vecc}{\mathbf{c}}

\newcommand{\vecu}{\mathbf{u}}
\newcommand{\vect}{\mathbf{t}}

\newcommand{\symwfomc}{\ensuremath{\mathsf{WFOMC}}}

\newcommand{\fotwo}{\ensuremath{\mathbf{FO}^2}\xspace}
\newcommand{\ctwo}{\ensuremath{\mathbf{C}^2}\xspace}
\newcommand{\sctwo}{\ensuremath{\mathbf{SC}^2}}

\newcommand{\indicator}[1]{\mathds{1}\left\{#1\right\}}

\newcommand{\domain}{\Delta}

\newcommand{\nat}{\mathbb{N}}

\newcommand{\fomodels}[2]{\ensuremath{\mathcal{M}_{#1, #2}}}
\newcommand{\promodels}[1]{\ensuremath{\mathcal{M}_{#1}}}

\newcommand{\ufotwo}{\ensuremath{\mathbf{UFO}^2}\xspace}
\newcommand{\fotwocc}{\ensuremath{\mathbf{FO}^2_{\textrm{CC}}}\xspace}

\newcommand{\typeweight}[1]{\mathcal{W} (#1)}

\newcommand{\incwfomc}{\textsc{IncrementalWFOMC}}
\newcommand{\ouralgo}{\textsc{IncrementalWFOMC3}}

\newcommand{\ctwomod}{\ensuremath{\mathbf{C}^2_{\textrm{mod}}}}

\newcommand{\vecsig}{\boldsymbol{\sigma}}
\newcommand{\unittensor}{\boldsymbol{\delta}}
\newcommand{\vecbeta}{\ensuremath{\boldsymbol{\beta}}}
\newcommand{\expctype}[2]{\ensuremath{\#{#2}(#1)}}
\newcommand{\dimensions}{\ensuremath{D}}

\newcommand{\realizedele}[2]{#1\langle #2\rangle}
\newcommand{\update}{\textsf{update}}
\newcommand{\Fast}{\textsc{Fast}}
\newcommand{\Recursive}{\textsc{Recursive}}

\newcommand{\Ganak}{\textsc{Ganak}\xspace}
\newcommand{\ApproxMC}{\textsc{ApproxMC}\xspace}
\newcommand{\Ours}{\textsc{Ours}}

\newcommand{\remaind}{\textsf{rem}}

\newcommand{\Mmod}{M_{\textrm{mod}}}

\title{A Fast Model Counting Algorithm for Two-Variable Logic with Counting and Modulo Counting Quantifiers}

\author[1]{Shixin Sun}
\author[2]{Astrid Klipfel}
\author[3]{Ond\v{r}ej Ku\v{z}elka}
\author[1]{Yuanhong Wang\thanks{Corresponding author.}}
\author[1,4]{Yi Chang}

\affil[1]{School of Artificial Intelligence, Jilin University, Changchun 130012, China}
\affil[2]{CRIL, Université d'Artois, Lens 62307, France}
\affil[3]{Czech Technical University, Prague 160 00, Czechia}
\affil[4]{Engineering Research Center of Knowledge-Driven Human-Machine Intelligence, MOE, Changchun 130012, China}

\date{}

\begin{document}



%

\maketitle

\begin{abstract}
    Weighted first-order model counting (WFOMC) is a central task in lifted probabilistic inference: It asks for the weighted sum of all models of a first-order sentence over a finite domain. A long line of work has identified domain-liftable fragments of first-order logic, that is, syntactic classes for which WFOMC can be solved in time polynomial in the domain size. Among them, the two-variable fragment with counting quantifiers, \ctwo{}, is one of the most expressive known liftable fragments. Existing algorithms for \ctwo{}, however, establish tractability through multi-stage reductions that eliminate counting quantifiers via cardinality constraints, which introduces substantial practical overhead as the domain size grows. In this paper, we introduce \ouralgo{}, a lifted algorithm for WFOMC on \ctwo{} and its modulo counting extension, \ctwomod{}. Instead of relying on reduction techniques, \ouralgo{} operates directly on a Scott normal form that retains counting quantifiers throughout inference. This direct treatment yields two main results. First, we derive a tighter data-complexity bound for WFOMC in \ctwo{}, reducing the degree of the polynomial from quadratic to linear in the counting parameters. Second, we prove that \ctwomod{} is domain-liftable, extending tractability from \ctwo{} to a richer fragment with native modulo counting support. Finally, our empirical evaluation shows that \ouralgo{} delivers orders-of-magnitude runtime improvements and better scalability than both existing \wfomc{} algorithms and state-of-the-art propositional model counters.
\end{abstract}

\section{Introduction}

\emph{Weighted first-order model counting} (\emph{WFOMC}) \citep{vandenbroeckLiftedProbabilisticInference2011,milch2008lifted} asks, given a first-order sentence, a finite domain, and a pair of weighting functions over predicates, to compute the weighted sum of all models of the sentence over that domain. Each model receives a weight equal to the product of predicate weights over all true and false ground atoms, and WFOMC sums these weights over all satisfying models. This formulation underlies probabilistic inference in statistical relational learning \citep{getoor2007introduction}: in frameworks such as Markov logic networks \citep{richardson2006markov} and probabilistic databases \citep{gribkoff_lifted_2014-1}, tasks such as partition-function computation, marginal inference, and expectation computation reduce naturally to WFOMC. For example, the sentence $\sentence = \forall x\colon \bigl(\mathit{adopt}(x) \leftrightarrow \exists^{\geq k} y\colon (\mathit{fr}(x,y) \land \mathit{adopt}(y))\bigr)$ states that a person adopts a behavior if and only if at least $k$ of their friends have adopted it; the WFOMC of $\sentence$ over a domain of $n$ individuals gives the weighted count of all adoption states consistent with this rule. Beyond probabilistic inference, WFOMC also provides a useful abstraction for enumerative combinatorics, where it can encode counting problems over structures such as permutations, trees, and directed acyclic graphs \citep{kuang_et_al:LIPIcs.CSL.2026.7}.

The straightforward strategy is to ground the first-order sentence into a propositional formula and count its weighted models. This approach does not scale as the domain grows, since computing the weighted model count of the resulting formula is \#\P-hard in general \citep{valiant_complexity_1979-2}, even though the grounding itself is only polynomial in the domain size. Moreover, it is unlikely that WFOMC can be computed in time polynomial in the \emph{sentence} size, as WFOMC subsumes \#SAT even for fixed relational schemas \citep{Beame_2015}. Consequently, most work focuses on the \emph{data complexity}~\citep{ceylan_open-world_2021}, measuring the cost of computing WFOMC as a function of the domain size $n$ alone with the sentence and weights treated as \emph{fixed}. This motivates the study of \emph{domain-liftable} fragments~\citep{vandenbroeckLiftedProbabilisticInference2011}---syntactic classes of first-order sentences for which WFOMC can be computed in time polynomial in $n$.

Several domain-liftable fragments have since been identified. Starting from the universally quantified two-variable fragment \ufotwo{} \citep{vandenbroeckLiftedProbabilisticInference2011}, liftability was extended to the two-variable logic \fotwo{} \citep{van2014skolemization}, then to \fotwocc{} \citep{kuzelka2021weighted} enriched with cardinality constraints, and finally to \ctwo{} \citep{kuzelka2021weighted}, which augments \fotwocc{} with counting quantifiers (i.e., quantifiers of the form $\exists^{\geq k}$, $\exists^{=k}$, and $\exists^{\leq k}$).
In particular, \ctwo{} is among the most expressive fragments currently known to admit polynomial-time WFOMC in data complexity. 
Existing algorithms for \ctwo{} rely on multi-stage reductions that reduce the problem to \ufotwo{}.
Briefly, a \ctwo{} sentence is transformed into a Scott normal form \citep{kuzelka2021weighted}, counting quantifiers are eliminated via cardinality constraints \citep{kuzelka2021weighted,KR2024-64}, and the result is then reduced further to a universally quantified fragment \citep{kuzelka2021weighted,van2014skolemization}.
We will review this reduction-based approach in more detail in \Cref{sec:reductions}.
Though these reductions preserve polynomial-time tractability in the domain size, they introduce substantial practical overhead, which limits scalability in practice.
For instance, the state-of-the-art reduction for \ctwo{} proposed recently by \citet{KR2024-64} yields a data-complexity exponent that grows quadratically with the counting parameters $k$ (note that it is still polynomial in the domain size for fixed $k$), which can lead to prohibitive runtimes even for moderate parameter values.

The limitations of the reduction-based approach become even more pronounced when \emph{modulo counting quantifiers} \citep{bednarczyk2017modulo,lodaya_two-variable_2017} are considered. These quantifiers extend standard counting quantifiers by constraining witness counts modulo a fixed integer; specifically, the sentence $\exists^{=r,k} x\colon \phi(x)$ asserts that the number of witnesses satisfying $\phi$ is congruent to $r$ modulo $k$. Over varying domain sizes, modulo counting quantifiers are strictly more expressive than standard counting quantifiers. For instance, periodic properties such as ``every vertex has a degree divisible by $k$'' can be captured by the sentence $\forall x\exists^{=0,k} y\colon E(x,y)$, while expressing this same property using only standard counting quantifiers requires a disjunction whose length grows with the domain size, which critically violates the requirements of domain-liftability.~\footnote{Simply transforming modulo counting quantifiers into a domain-dependent disjunction of standard counting quantifiers, e.g., transforming $\forall x\exists^{=0,2} y\colon \phi(y)$ into $\forall x\left(\exists^{=0} y\colon \phi(y) \lor \exists^{=2} y\colon \phi(y) \lor \dots \lor \exists^{=n} y\colon \phi(y)\right)$ for an even domain size $n$, does not yield a polynomial-time algorithm. Because the parameters of standard counting quantifiers appear in the exponent of the data-complexity polynomial in existing liftable algorithms, even a linear increase in these parameters with respect to the domain size leads to an exponential blow-up in runtime.} 
While historically the study of modulo counting quantifiers has been largely confined to theoretical explorations---specifically, in the realm of descriptive complexity and model checking, where existing literature primarily investigates the expressive boundaries and fixed-parameter tractability of first-order logic enriched with modulo counting quantifiers over restricted relational structures \citep{bednarczyk2017modulo,heimberg2016hanf}---their algorithmic integration into exact model counting remains critically underexplored.
This provides a natural motivation for studying domain-liftable fragments beyond the \ctwo{} fragment, which we address in this work.

We address this gap by introducing \ouralgo{}, a novel lifted algorithm for WFOMC on \ctwo{} and \ctwomod{}. 
\ouralgo{} builds upon the IncrementalWFOMC family \citep{toth2023lifted,DBLP:conf/ecai/ZouMZ0KWC25}, a series of algorithms that compute WFOMC by growing the domain one element at a time. 
The key insight behind these algorithms is that the weighted count of models over a domain of size $n$ can be computed from the counts over smaller domains by tracking a compact summary of the weighted count organized by \emph{cell configuration}---that is, tracking how many domain elements satisfy each combination of unary predicates in the sentence \citep{van2021faster}.
\ouralgo{} inherits this domain-recursive structure but maintains a richer state representation that also tracks, for each domain element, how many of its neighbors satisfy each counting quantified formula, which allows it to operate directly on the normal form of the input sentence without eliminating counting quantifiers via reductions.
This augmented state-tracking is reminiscent of lifted model sampling~\citep{wang2024lifted} and enumeration~\citep{meng2025model}, where type-level summaries are likewise enriched with counting information during model construction.

By working directly on the sentence and avoiding the computational overhead of modular reductions to \ufotwo{}, \ouralgo{} achieves significant theoretical improvements. For \ctwo{}, it yields a tighter data-complexity bound, improving the previous result by \citet{KR2024-64} by reducing the degree of the data-complexity polynomial from quadratic to linear in the counting parameters. More importantly, this direct state-tracking framework easily extends to modulo constraints, allowing us to prove for the first time that \ctwomod{} is domain-liftable.

The main contributions of this paper are summarized as follows:
\begin{itemize}
    \item We introduce \ouralgo{}, a more efficient algorithm for WFOMC on \ctwo{} and \ctwomod{}.
    \item We establish a tighter data-complexity bound for \ctwo{}, improving the previous bound by reducing the degree of the polynomial from quadratic to linear in the counting parameters.
    \item With \ouralgo{}, we prove for the first time that \ctwomod{} is domain-liftable, extending tractability to a richer fragment with native modulo counting support.
    \item We empirically evaluate \ouralgo{} on several benchmark families and observe substantial improvements in runtime and scalability over strong lifted and propositional baselines.
\end{itemize}

The remainder of the paper is organized as follows. 
\Cref{sec:related_work} reviews related work. 
\Cref{sec:background} introduces the necessary background on WFOMC, domain-liftability, and the reduction-based approach from \ctwo{} to \ufotwo{}.
\Cref{sec:incrementalWFOMC3} presents the main algorithm, \ouralgo{}.
\Cref{sec:modulo_counting} extends the algorithm to modulo counting quantifiers and proves the domain-liftability of \ctwomod{}. 
\Cref{sec:experiments} reports the experimental evaluation. 
Finally, \Cref{sec:conclusion} concludes the paper.

\section{Related Work}\label{sec:related_work}

We review the literature that is most relevant to our work.
For a broader review of the applications of WFOMC in probabilistic inference and combinatorics, we refer the reader to \citet{VandenBroeck_Kersting_Natarajan_Poole_2021} and \citet{kuang_et_al:LIPIcs.CSL.2026.7}, respectively.

\paragraph{Domain-liftable fragments.}
Domain-liftability was first established for the universally quantified two-variable fragment \ufotwo{}~\cite{vandenbroeckLiftedProbabilisticInference2011} using first-order knowledge compilation, providing the foundational algorithm on which subsequent reductions are built.
\citet{Beame_2015} later developed a database-theoretic perspective on symmetric WFOMC, establishing $\#\P_1$-hardness for more expressive fragments such as $\mathbf{FO}^3$.
Together, these results set the tractability frontier for two-variable logic and motivated the search for domain-liftable fragments within it, e.g., augmenting \fotwo{} with cardinality constraints and counting quantifiers~\cite{kuzelka2021weighted}, linear orders~\cite{toth2023lifted}, and general combinatorial constraints~\cite{kuang_et_al:LIPIcs.CSL.2026.7}.

\paragraph{\wfomc{} reductions.}
A key proof technique in this line of work is \emph{\wfomc{} reductions}: transformations that convert a WFOMC on a richer fragment into one or more WFOMCs on simpler fragments, while preserving polynomial-time tractability in the domain size.
Building on the liftability of \ufotwo{}, subsequent work extended domain-liftability to progressively richer two-variable languages via such reductions.
Specifically, Skolemization~\cite{van2014skolemization} reduces WFOMC on \fotwo{} to \ufotwo{}, and thereby proves the liftability of \fotwo{}.
The framework was further extended to \fotwocc{}~\cite{Kazemi_Kimmig_Van_Broeck_Poole_2016}, which additionally allows cardinality constraints, and ultimately to \ctwo{}~\cite{kuzelka2021weighted}, which augments \fotwocc{} with counting quantifiers~\cite{gradel1997two}.
The fragment \ctwo{} is among the most expressive two-variable fragments currently known to admit polynomial-time WFOMC in data complexity.
Recently, \citet{KR2024-64} proposed a more efficient reduction from \ctwo{} to \ufotwo{}, setting a data-complexity bound to beat that is polynomial in the domain size with a degree that grows quadratically in the counting parameters.

\paragraph{Domain recursion and Incremental WFOMCs.}

Another technique for \wfomc{} that is orthogonal to modular reductions is the \emph{domain recursion rule}~\cite{broeckCompletenessFirstOrderKnowledge2011,KazemiKBP17}.
It was originally proposed as a rule for first-order knowledge compilation \citep{broeckCompletenessFirstOrderKnowledge2011} and later analyzed in more depth by \citet{KazemiKBP17}, who showed that it can be used to establish domain-liftability for certain theories beyond \fotwo{}, which they called $\mathbf{S}^2\mathbf{FO}^2$ and $\mathbf{S}^2\mathbf{RU}$.
This rule also underlies the IncrementalWFOMC family~\cite{toth2023lifted,lifted_2025}, a series of \wfomc{} algorithms primarily for \ufotwo{} with a linear-order axiom that enforces some binary predicate to represent the linear ordering of the domain elements.
We refer the reader to \citep{toth2023lifted} and \citep{DBLP:conf/ecai/ZouMZ0KWC25} for details on the linear order axiom, as it is not the main focus of this paper.
We do, however, note that the linear order axiom is naturally supported by our new algorithm following the same incremental approach, and we empirically evaluate \ouralgo{} with and without the linear order axiom in our experiments in \Cref{sec:experiments}.
A closely related algorithm, RecursiveWFOMC~\cite{Meng2024}, utilizes a similar idea but a different state representation in the recursion, achieving the state-of-the-art in this linear-order setting.
For more specific comparisons on these algorithms in the linear-order setting, see \citet{Meng2024}.
Related ideas also appear in work on first-order model sampling~\cite{wang2024lifted} and enumeration~\cite{meng2025model}, where models are constructed incrementally by growing the domain one element at a time.
Our algorithm, \ouralgo{}, builds on the same domain-recursion foundation, extending it to \ctwo{} with counting and modulo counting quantifiers.

\paragraph{Counting and modulo counting in two-variable logic.}
Counting quantifiers in two-variable logic have been studied from both logical and combinatorial perspectives.
\citet{gradel1997two} established the decidability of \ctwo{} and analyzed its model-theoretic properties, and later \citet{kopczynski2015regular} analyzed the spectra of \ctwo{}, connecting the fragment to the theory of Presburger arithmetic and semilinear sets.
From a combinatorial perspective, \ctwo{} can be used to express a wide range of combinatorial properties, such as regularity in graphs whose counting and sampling problems have been studied in the literature~\citep{cooper2007sampling,GaoW17}.
\ctwo{} has also been connected to the expressiveness of graph neural networks (GNNs).
\citet{barcelo2020logical} showed that standard message-passing GNNs are logically captured within \ctwo{}, and the subsequent work established that \ctwo{} strictly subsumes their expressive power~\cite{Hauke_Walega_2026}.
Since counting quantifiers cannot express even such simple properties as parity, Wolfgang Thomas and his coauthors introduced modulo counting quantifiers of the form ``there exist $r$ mod $k$ elements $x$ such that \dots'' in the 1980s. 
A survey of results in first-order logic extended with modulo counting quantifiers can be found in \citep{flum2008logic}.
Recent work has further analyzed the complexity of satisfiability and model checking for two-variable logic with modulo counting quantifiers, establishing strict computational limits for this fragment~\citep{lodaya_two-variable_2017,bednarczyk2017modulo,heimberg2016hanf}.
However, none of this prior work addresses the data complexity of WFOMC in the presence of modulo counting quantifiers, leaving the domain-liftability of \ctwomod{} as an open question that we resolve in this paper.
Finally, we note the work by \citet{benedikt2024two} on the two-variable fragment with ultimately periodic counting quantifiers, which are a generalization of modulo counting quantifiers that allow for more complex periodicity patterns in the witness counts.
Our framework would also apply to this more general class of counting quantifiers, and we leave the question of domain-liftability for the resulting fragment as an interesting direction for future work.

\section{Background}\label{sec:background}

We introduce the necessary notations and concepts used in this paper.
We use boldface letters, e.g., $\mathbf{k}$, to denote tensor-like objects, such as vectors and matrices.
The entry in the $i_1, \ldots, i_d$ position of a $d$-dimensional tensor $\mathbf{k}$ is denoted as $\mathbf{k}[i_1, \ldots, i_d]$.
We also write $\veck[\mathbf{i}]$ for $\mathbf{k}[i_1, \ldots, i_d]$ when $\mathbf{i} = (i_1, \ldots, i_d)$ is a vector.
The 1-norm of a $d$-dimensional tensor $\mathbf{k}$ is defined as $|\veck| = \sum_{i_1, \ldots, i_d} |\veck[i_1, \ldots, i_d]|$.
The basic binary operators, such as $+$ and $-$, are extended to tensors in the natural way, applied \emph{element-wise}.
We write $\veck_1 \le \veck_2$ if $\veck_1[i_1, \ldots, i_d] \le \veck_2[i_1, \ldots, i_d]$ for all $i_1, \ldots, i_d$.
The other comparisons, such as $<$, $>$, $\ge$, and $=$ are defined similarly.
By $\mathbf{0}^{n_1\times n_2\times \ldots \times n_d}$ we denote the zero tensor of dimensions $n_1 \times n_2 \times \ldots \times n_d$.
The unit tensor of dimensions $n_1 \times n_2 \times \ldots \times n_d$ where only the entry in the $i_1, \ldots, i_d$ position is $1$ and all other entries are $0$ is denoted as $\unittensor_{i_1,\ldots ,i_d}^{n_1\times\ldots\times n_d}$.
We omit the dimensions when they are clear from the context.
We denote $\binom{n}{\veck} = \binom{n}{\veck[1], \ldots, \veck[d]}$ as the multinomial coefficient for some vector $\veck \in \nat^d$.

Let $n$ be a positive integer.
We denote the set $\{1, 2, \ldots, n\}$ by $[n]$.
We denote $\mathbb{Z}_n$ as the set of all remainders modulo $n$, that is $\mathbb{Z}_n = \{0, 1, \ldots, n-1\}$.
The remainder function modulo $n$ is written as $\remaind_n: \nat \to \mathbb{Z}_n$, i.e., $\remaind_n(x) = x \bmod n$.
We write the indicator function as $\indicator{\cdot}$, which is $1$ if the argument is true and $0$ otherwise.

\subsection{First-Order Logic and Extensions}

In this paper, we consider the \emph{function-free} \emph{finite-domain} first-order logic (FOL).
An \emph{atomic formula} (or simply \emph{atom}) is an expression of the form $P(t_1, \ldots, t_k)$, where $P/k$ is a predicate (also called a \emph{relation}) of arity $k$ and $t_1, \ldots, t_k$ are either \emph{constants} from a finite domain $\domain$, or \emph{variables} $\{x, y, \ldots\}$.
A \emph{literal} is an atomic formula or its negation.
\emph{Formulas} are defined inductively:
All literals are formulas;
if $\phi$ and $\psi$ are formulas, then $\neg \phi$, $\phi \land \psi$, $\phi \lor \psi$, $\phi\to \psi$, and $\phi \leftrightarrow \psi$ are also formulas;
a formula surrounded by $\forall x$ or $\exists x$, where $x$ is a logical variable, is also a formula, which is called a \emph{quantified formula}.
The set of predicates appearing in a formula $\sentence$ is denoted as $\mathcal{P}_{\sentence}$.
A variable is \emph{free} in a formula if it is not bound by a quantifier.
For instance, $coin(x) = (H(x) \lor T(x)) \land \neg (H(x) \land T(x))$ is a formula with one free variable $x$.
A formula without any quantifiers is called a \emph{quantifier-free formula}.
A formula with no free variables is called a \emph{sentence}. 
A \emph{ground} formula is a formula without variables.
A ground formula can be viewed as a propositional formula, where each ground atom $P(c_1, \ldots, c_k)$ is a propositional (Boolean) variable.
Given a domain $\domain$ and a sentence $\sentence$, one can obtain the \emph{grounding} of $\sentence$ by replacing each variable with a constant from $\domain$, and expanding the quantifiers $\forall x$ and $\exists x$ into conjunctions and disjunctions, respectively.


We make use of the \emph{Herbrand Base} (HB) to define the semantics.
The HB of a sentence $\sentence$ on a finite domain $\domain$ is the set of all ground atoms that can be formed by the predicates in $\sentence$ and the constants in $\domain$.
A \emph{Herbrand interpretation} (or simply an interpretation) of $\sentence$ over $\domain$ is a truth assignment to all atoms in the HB of $\sentence$ on $\domain$.
We often write an interpretation $\mathcal{I}$ as a set of ground literals such that for each ground atom $a$ in HB, either $a \in \mathcal{I}$ or $\neg a \in \mathcal{I}$, interpreted respectively as $a$ is true or $a$ is false.
For example, $\mathcal{I} = \{H(1),\neg T(1),\neg H(2), T(2)\}$ is an interpretation of $\forall x: coin(x)$ over the domain $\{1, 2\}$.
A \emph{model} of a sentence $\sentence$ is an interpretation $\mathcal{I}$ that satisfies $\sentence$ following the standard semantics of FOL, which is denoted by $\mathcal{I} \models \sentence$.
For instance, the aforementioned interpretation $\mathcal{I}$ is a model of $\forall x: coin(x)$.
We denote the set of all models of $\sentence$ over the domain $\{1, 2, \ldots, n\}$ as $\fomodels{\sentence}{n}$.
When the sentence $\sentence$ is ground, a model is a set of ground literals interpreting each ground atom in $\sentence$ that makes $\sentence$ true.
The set of all models of a ground sentence $\sentence$ is denoted as $\promodels{\sentence}$.

There are several extensions of FOL that augment the expressiveness of the language.
In this paper, we mainly focus on \emph{counting quantifiers}~\cite{gradel1997two} and their variant \emph{modulo counting quantifiers}~\cite{lodaya_two-variable_2017}.
\begin{itemize}
    \item Counting quantifiers $\exists^{=k}$, $\exists^{\le k}$ and $\exists^{\ge k}$ are generalizations of the traditional existential quantifier. Their semantics are defined as follows. An interpretation $\mathcal{I}$ satisfies $\exists^{=k} x:\phi(x)$ (resp. $\exists^{\le k} x:\phi(x)$, $\exists^{\ge k} x:\phi(x)$) iff. there are exactly (resp. at most, at least) $k$ constants $c\in\Delta$ such that $\mathcal{I}\models\phi(c)$. For example, the sentence $(\forall x: coin(x)) \land (\exists^{=2} x: H(x))$ requires that there are exactly two coins that are heads.
    \item The modulo counting quantifiers $\exists^{=r,k}$, $\exists^{\le r,k}$, and $\exists^{\ge r,k}$ constrain the count modulo $k$. An interpretation $\mathcal{I}$ satisfies $\exists^{=r,k} x:\varphi(x)$, $\exists^{\le r,k} x:\varphi(x)$, $\exists^{\ge r,k} x:\varphi(x)$ iff. $\remaind_{k}(m) = r$, $\remaind_{k}(m) \le r$, and $\remaind_{k}(m) \ge r$, respectively, where $m = |\{c\in\Delta: \mathcal{I}\models\varphi(c)\}|$.
\end{itemize}

\subsection{Weighted First-Order Model Counting}

The \textit{first-order model counting problem}~\citep{vandenbroeckLiftedProbabilisticInference2011} asks, given a domain $\domain$ and a sentence $\sentence$, how many models $\sentence$ has over $\domain$.
The \textit{weighted first-order model counting problem} (WFOMC) extends the input with a pair of weighting functions $(\weight,\negweight)$ that map predicates in $\mathcal{P}_\sentence$ to real weights: $\mathcal{P}_\sentence \to \mathbb{R}$.
Given a set $L$ of literals, the weight of $L$ is defined as
\begin{equation*}
  \typeweight{L} = \prod_{lit\in L^T} \weight(\mathsf{pred}(lit)) \cdot \prod_{lit\in L^F} \negweight(\mathsf{pred}(lit)),
\end{equation*}
where $L^T$ and $L^F$ are the sets of true and false literals in $L$, respectively, and $\mathsf{pred}(lit)$ maps a literal $lit$ to its predicate.

\begin{example}\label{ex:litweight}
    Let $\domain = \{1, 2, 3\}$, and consider the sentence $\sentence = \forall x: coin(x)$.
    For the weighting functions $\weight(H) = 2$, $\negweight(H) = \weight(T) = \negweight(T) = 1$, 
    the weight $\typeweight{\mathcal{I}}$ of the model $\mathcal{I} = \{H(1), \neg T(1), H(2), \neg T(2), \neg H(3), T(3)\}$ is $2^2 \cdot 1^4 = 4$.
\end{example}

\begin{definition}[Weighted first-order model counting]
    Let $(\weight, \negweight)$ be a weighting on a sentence $\sentence$. The WFOMC of $\sentence$ over a finite domain of size $n$ under $(\weight, \negweight)$ is
    \begin{equation*}
        \symwfomc(\sentence, n, \weight, \negweight) = \sum_{\mu\in\fomodels{\sentence}{n}}\typeweight{\mu}.
    \end{equation*}
\end{definition}

\begin{example}
    Continuing \Cref{ex:litweight}, we introduce a modulo counting quantifier to form a new sentence $\sentence' = \sentence\land \exists^{=1,2} x: H(x)$.
    We have $\symwfomc(\sentence, 3, \weight, \negweight) = (1 + 2)^3 = 27$ and $\symwfomc(\sentence', 3, \weight, \negweight) = \binom{3}{1}\cdot 2 + \binom{3}{3}\cdot 2^3 = 14$.
    The ratio $14 / 27$ can be interpreted as the probability of obtaining an odd number of heads when each coin is independently flipped with a probability of heads being $2/3$.
\end{example}

\subsection{Domain-Liftability and \wfomc{} Reductions} \label{sec:reductions}

We consider the complexity of computing \wfomc{} in terms of the domain size $n$, while the sentence and the weights are \emph{fixed}.
This perspective is often referred to as the \emph{data complexity} in the database literature~\citep{Beame_2015,ceylan_open-world_2021}.
We are particularly interested in the syntactic fragments of FOL that admit polynomial-time algorithms for \wfomc{}, which are called \emph{domain-liftable}~\cite{VandenBroeck_Kersting_Natarajan_Poole_2021}.

\begin{definition}[Domain-liftability]
    A fragment $\mathcal{F}$ of FOL is \emph{domain-liftable} (or simply \emph{liftable}) if there exists an algorithm that for any fixed sentence $\sentence\in \mathcal{F}$, any fixed weights $(\weight, \negweight)$, computes $\symwfomc(\sentence, n, \weight, \negweight)$ in time polynomial in the domain size $n$.
\end{definition}

The first fragment that was identified as liftable is the universally quantified two-variable fragment of FOL (\ufotwo{}) consisting of sentences of the form $\forall x \forall y\colon \psi(x,y)$, where $\psi(x,y)$ is a quantifier-free formula over at most two variables.
Its domain-liftability was originally shown by \citet{vandenbroeckLiftedProbabilisticInference2011} using knowledge compilation techniques, and later reformulated by \citet{Beame_2015}.
Recently, \citet{toth2023lifted} provided another algorithm based on the work of \citet{Beame_2015}, which forms the basis of our algorithms, and will be sketched in \Cref{sec:incwfomc}.
The other liftable languages include: 
\begin{itemize}
    \item the two-variable fragment (\fotwo{}) that restricts the number of distinct variables to two, but unlike \ufotwo{} also allows existential quantifiers~\cite{van2014skolemization},
    \item the two-variable fragment with cardinality constraints (\fotwocc{}) that extends \fotwo{} with cardinality constraints of the form $|P| = d$, where $P$ is a predicate and $d$ is a non-negative integer, restricting the number of true ground literals of $P$ in a model to $d$~\citep{kuzelka2021weighted}, and
    \item the two-variable fragment with counting (\ctwo{}) that extends \fotwocc{} with counting quantifiers~\citep{kuzelka2021weighted}.
\end{itemize}
All these fragments were proven to be domain-liftable following reductions to \ufotwo{}, as shown in \Cref{fig:reduction}.
Specifically, a \ctwo{} sentence can be transformed into a normal form, called \sctwo{} (\emph{Scott normal form with counting}\footnote{The name introduced by \citet{wang2024lifted} is inspired by the Scott normal form for \fotwo{}~\cite{gradel1997decision}, which is of the form $\forall x\forall y: \phi(x,y) \land \bigwedge_{i=1}^m \forall x\exists y: \psi_i(x,y)$, where $\phi$ and $\psi_i$ are quantifier-free formulas.}):

\begin{equation}\label{eq:ctwo_normal_form}
\Gamma = \Psi \wedge \bigwedge_{i \in [M]} \left( \forall x \exists^{=k_i} y : R_i(x, y) \right),
\end{equation}
where $\Psi$ is an $\mathbf{FO}^2_{\mathrm{CC}}$ sentence, each $R_i/2$ is a distinguished binary predicate, each $k_i$ is a non-negative integer, 
and $M$ is a non-negative integer~\citep{kuzelka2021weighted}. 
Then the counting quantifiers in \cref{eq:ctwo_normal_form} can be eliminated by introducing additional cardinality constraints~\citep{kuzelka2021weighted}, resulting in another $\mathbf{FO}^2_{\mathrm{CC}}$ sentence $\Gamma'$.

Finally, \wfomc{} of $\sentence'$ can be computed with polynomial calls to \wfomc{}s of \fotwo{} sentences using Lagrange interpolation~\cite{kuzelka2021weighted}, which is further reduced to \ufotwo{} by Skolemization eliminating the existential quantifiers~\cite{van2014skolemization,Beame_2015}.
All these reductions preserve domain-liftability, i.e., the target \wfomc{} problems can be solved in polynomial time in the domain size with access to an oracle for the source \wfomc{} problems.
The details of reductions as well as the proofs of the reductions preserving domain-liftability can be found in~\cite[Section 2.3]{KR2024-64}.
Finally, we note that the reductions eliminating cardinality constraints and existential quantifiers are \emph{modular} in the sense that they can be applied to any \wfomc{} problem even outside the fragment of two-variable logic, and thus they also apply to the setting with modulo counting quantifiers, which we will discuss in \Cref{sec:modulo_counting}.

\begin{figure}[tbp]
    \centering
    \includegraphics[width=\linewidth]{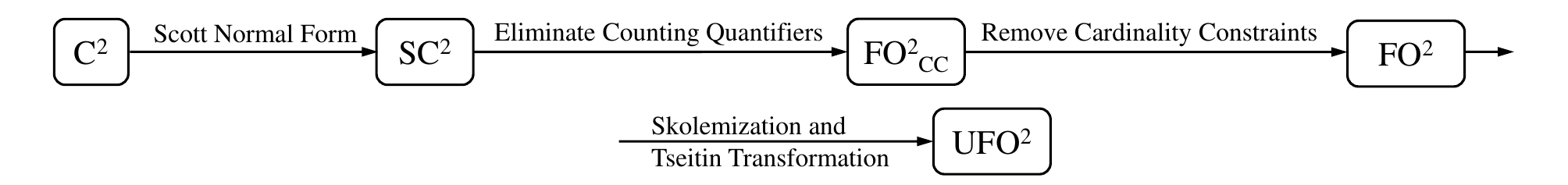}
    \caption{\wfomc{} reductions for \ctwo{}.}\label{fig:reduction}
\end{figure}


\subsection{1-Types and 2-Tables}\label{sec:1types}
Before introducing \incwfomc{} \citep{toth2023lifted}, we need the notions of \emph{1-types} and \emph{2-tables}, which are commonly used in the literature studying two-variable logic~\cite{lin_towards_2021-1,wang2024lifted}.

A \emph{unary literal} is a literal containing only one variable $x$.
A \emph{binary literal} is a literal containing both variables $x$ and $y$.
For example, $P(x)$ is a unary literal, and $R(x,y)$ is a binary literal, while $R(x,x)$ is a unary literal even though $R$ is a binary predicate.
We call a set of literals \emph{maximally consistent} if it is not possible to add any literal without introducing a contradiction.

\begin{definition}[1-types and 2-tables]
    A \emph{1-type} $\tau$ of a two-variable quantifier-free formula $\phi$ is a maximal consistent set of unary literals formed from the predicates in $\phi$.
    A \emph{2-table} $\pi$ of $\phi$ is a maximal consistent set of binary literals formed from the predicates in $\phi$.
\end{definition}

We also view a 1-type or a 2-table as a conjunction of literals.
In this sense, a 1-type $\tau$ and a 2-table $\pi$ are both quantifier-free formulas, and we write them as $\tau(x)$ and $\pi(x,y)$, respectively.
A 1-type $\tau$ of a formula $\psi(x,y)$ is \emph{valid} if and only if for any constant $a$, $\tau(a)\models \psi(a,a)$.
In the rest of the paper, we assume that all 1-types mentioned are valid, unless otherwise specified.
A 2-table $\pi$ is \emph{compatible} with $(\tau,\tau')$ for $\psi(x,y)$ if $\tau(a) \land \pi(a,b)\land \tau'(b) \models \psi(a,b)$, for any constants $a$ and $b$.
Given a tuple of 1-types $(\tau, \tau')$, we denote the set of all 2-tables compatible with $(\tau, \tau')$ as $\Pi_{\tau, \tau'}$.

\begin{example}\label{ex:1type}
    We consider the sentence encoding 2-colored graphs: 
    \begin{align*}
        \sentence_{CG} = & \forall x\forall y: E(x,y) \to E(y,x) \land \\
        &\forall x: (R(x) \lor B(x)) \land \neg(R(x) \land B(x)) \land\\
        & \forall x\forall y: E(x,y) \to \left(\neg(R(x) \land R(y)) \land \neg(B(x) \land B(y))\right).
    \end{align*}
    This sentence enforces that every vertex is either red $(R)$ or black $(B)$, 
    but not both, and an edge $(E)$ can only connect vertices of different colors.
    The valid 1-types are $\tau_1 = \{R(x), \neg B(x), \neg E(x,x)\}$ (representing a red vertex with no self-loop)
    and $\tau_2 = \{\neg R(x), B(x), \neg E(x,x)\}$ (representing a black vertex with no self-loop).
    Note that self-loops are prohibited due to the coloring constraint $\forall x: E(x,x) \to \neg(R(x) \land R(x)) \land \neg(B(x) \land B(x))$.
    The compatible 2-tables of $(\tau_1, \tau_2)$ (as well as $(\tau_2, \tau_1)$) contain $\pi_1 = \{E(x,y), E(y,x)\}$ and $\pi_2 = \{\neg E(x,y), \neg E(y,x)\}$, i.e., $\Pi_{\tau_1, \tau_2} = \Pi_{\tau_2, \tau_1} = \{\pi_1, \pi_2\}$.
    The compatible 2-table of $(\tau_1, \tau_1)$ (as well as $(\tau_2, \tau_2)$) contains only $\pi_2 = \{\neg E(x,y), \neg E(y,x)\}$.
\end{example}

Given an interpretation $\mathcal{I}$ over a domain $\domain$, a constant $a\in\domain$ is said to \emph{realize} a 1-type $\tau$ in $\mathcal{I}$ if $\tau(a)$ is satisfied in $\mathcal{I}$, i.e., $\tau(a) \subseteq \mathcal{I}$.
Each constant can realize only one 1-type in an interpretation.
The set of all constants realizing a 1-type $\tau$ in $\mathcal{I}$ is denoted as $\realizedele{\mathcal{I}}{\tau}$.
Let $\tau_1, \ldots, \tau_p$ be all the 1-types of $\sentence$.
The \emph{1-type configuration} of an interpretation $\mathcal{I}$ of $\sentence$ is a non-negative integer vector $\veck$ of size $p$, where $\veck[i] = |\realizedele{\mathcal{I}}{\tau_i}|$ for each $i\in[p]$.
The number of all possible 1-type configurations over a domain of size $n$ is $\binom{n+p-1}{p-1}$, which is the number of non-negative integer solutions to the equation $x_1 + \ldots + x_p = n$.


\begin{figure}[tbp]
    \centering
    \begin{subfigure}[b]{0.22\textwidth}
        \includegraphics[width=\linewidth]{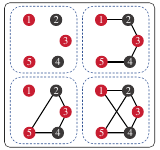}
        \caption{2-colored graphs}\label{fig:2-color-example}
    \end{subfigure}
    \begin{subfigure}[b]{0.3\textwidth}
        \includegraphics[width=\linewidth]{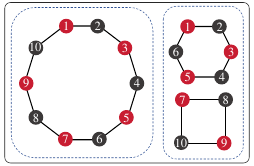}
        \caption{2-regular 2-colored graphs}\label{fig:2-reg-2-color-example}
    \end{subfigure}

    \caption{Illustrated models of colored graphs.}
\end{figure}

\begin{example}\label{ex:1_type_config}
    Example models of the 2-colored graph sentence $\sentence_{CG}$ over the domain size of 5 are illustrated in \Cref{fig:2-color-example}.
    In these models, constants 1, 3, and 5 realize the 1-type $\tau_1$ (red), and constants 2, 4 realize the 1-type $\tau_2$ (black). 
    The corresponding 1-type configuration is the vector $\veck = (3,2)$.
\end{example}

\subsection{\incwfomc{}}\label{sec:incwfomc}

We sketch the algorithm \incwfomc{} from \citep{toth2023lifted}.
The algorithm \incwfomc{} works on a fixed \ufotwo{} sentence $\sentence = \forall x\forall y: \fotwoformula(x,y)$ over a domain $[n]$ under fixed weights $(\weight, \negweight)$.
The approach also applies to \ctwo{}, \fotwocc{}, and \fotwo{} sentences via the reductions in \Cref{sec:reductions}.

Let $\tau_1, \ldots, \tau_p$ be all the 1-types of $\sentence$.
Since $\sentence$ is fixed, the number of 1-types $p$ is also fixed.
Define $T_h(\veck)$ for some $h\in[n]$ as the summation of $\typeweight{\mu}$ over all models $\mu\in\fomodels{\sentence}{h}$ with 1-type configuration $\veck$ over the domain $[h]$.
Then the \wfomc{} can be computed as
\begin{equation}\label{eq:incwfomc_final}
    \symwfomc(\sentence, n, \weight, \negweight) = \sum_{\veck\in\nat^p: |\veck| = n} T_n(\veck).
\end{equation}

\incwfomc{} computes $T_h(\veck)$ using dynamic programming.
In the base case $h=1$, there is a single element in the domain.
Accordingly, we initialize $T_1$ by setting $T_1(\unittensor_i) = \typeweight{\tau_i}$ for each $i\in[p]$, where $\unittensor_i$ is the unit tensor of size $p$ with only the $i$-th entry being $1$.

For $h>1$, \incwfomc{} computes $T_h$ from $T_{h-1}$.
Observe that any model $\mu'\in\fomodels{\sentence}{h-1}$ can be extended to a (possibly empty) set of models $\mu\in\fomodels{\sentence}{h}$:
\begin{equation}\label{eq:incwfomcmodel}
    \mu = \mu' \cup \tau_l(h) \cup \bigcup_{a\in[h-1]} \pi^{\mu'}_{al}(a,h),
\end{equation}
for some $l\in[p]$ and some $\pi^{\mu'}_{al}\in\Pi_{\zeta_{a,\mu'},\tau_l}$, where $\zeta_{a,\mu'}$ is the 1-type of $a$ in $\mu'$.
Denote by $M_{\mu',\tau_l}$ the set of all such models $\mu$ constructed from $\mu'$ and $\tau_l$. 
It is easy to check that all $M_{\mu',\tau_l}$ form a partition of $\fomodels{\sentence}{h}$, i.e., for any $\mu'_1, \mu'_2\in\fomodels{\sentence}{h-1}$ and $l_1, l_2\in[p]$ such that $\mu'_1 \neq \mu'_2$ or $l_1 \neq l_2$, we have $M_{\mu'_1, \tau_{l_1}} \cap M_{\mu'_2, \tau_{l_2}} = \emptyset$.
Moreover, if the 1-type configuration of $\mu'$ is $\veck$, then the 1-type configuration of $\mu$ in $M_{\mu',\tau_l}$ is $\veck + \unittensor_l$.

Let $W_{\mu',\tau_l} = \sum_{\mu\in M_{\mu',\tau_l}} \typeweight{\mu}$.
Then the value of $T_h(\veck)$ can be computed as the summation of $W_{\mu',\tau_l}$ over all $\mu'\in\fomodels{\sentence}{h-1}$ and $l\in[p]$ such that $\mu'$ has 1-type configuration $\veck-\unittensor_l$.
\Cref{fig:extension} illustrates examples of $M_{\mu',\tau_l}$ and $W_{\mu',\tau_l}$, where all weights of predicates are set to 1.

\begin{figure}
    \centering
    \includegraphics[width=\textwidth]{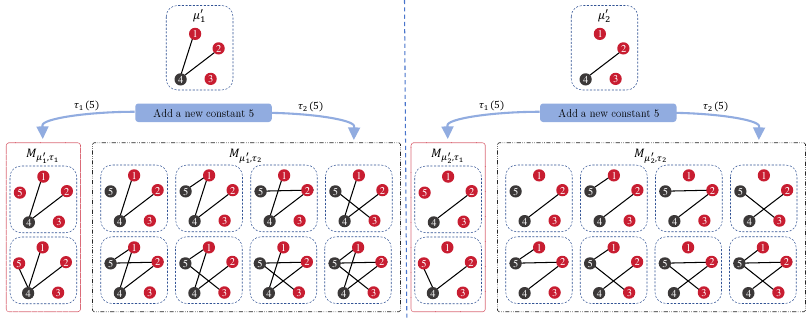}
    \caption{Model extensions of two colored graphs $\mu_1'$ and $\mu_2'$ with the same 1-type configuration $(3, 1)$ by adding a new vertex $5$.
    The sets of extended models $M_{\mu_1', \tau_1}$, $M_{\mu_1', \tau_2}$, $M_{\mu_2', \tau_1}$, and $M_{\mu_2', \tau_2}$ are mutually disjoint.
    As the weights of all predicates are set to 1, $W_{\mu, \tau_l}$ is simply the number of models in $M_{\mu, \tau_l}$.
    It is clear that $W_{\mu_1', \tau_1} = W_{\mu_2', \tau_1} = 2$ and $W_{\mu_1', \tau_2} = W_{\mu_2', \tau_2} = 8$, 
    since $\mu_1'$ and $\mu_2'$ have the same 1-type configuration and the added vertex $5$ has the same 1-type.}
    \label{fig:extension}
\end{figure}

Now we can compute $W_{\mu',\tau_l}$ for some $\mu'\in\fomodels{\sentence}{h-1}$ and $l\in[p]$.
Define
\begin{equation}\label{eq:weights}
    \begin{aligned}
        w_{\tau_i} = \typeweight{\tau_i} \qquad \text{ and } \qquad r_{\tau_i, \tau_j} = \sum_{\pi\in\Pi_{\tau_i, \tau_j}} \typeweight{\pi},
    \end{aligned}
\end{equation}
for each $i,j\in[p]$.
Since $\mu'$, $\tau_l(h)$, and $\pi^{\mu'}_{al}(a,h), a\in[h-1]$ are mutually disjoint, that is, any two of them do not share common ground literals, we can write the weight of $\mu$ as
\begin{equation}\label{eq:modelweight}
    \typeweight{\mu} = \typeweight{\mu'}\cdot w_{\tau_l} \cdot \prod_{a\in[h-1]}\typeweight{\pi^{\mu'}_{al}}.
\end{equation}
Moreover, the choices of $\pi^{\mu'}_{al}$ for different $a$ do not affect each other, since the 1-types $\zeta_{a,\mu'}$ and $\tau_l$ are fixed.
By the definition in \cref{eq:weights}, we have
\begin{equation}\label{eq:weightprod}
    W_{\mu',\tau_l} = \typeweight{\mu'}\cdot w_{\tau_l} \cdot \prod_{i\in[p]} r_{\tau_i, \tau_l}^{|\realizedele{\mu'}{\tau_i}|}.
\end{equation}
Note that $\left(|\realizedele{\mu'}{\tau_1}|, \ldots, |\realizedele{\mu'}{\tau_p}|\right)$ is exactly the 1-type configuration of $\mu'$.

Therefore, we can compute $T_h(\veck)$ as
\begin{equation}\label{eq:incwfomc}
    T_{h}(\veck) = \sum_{\substack{l\in[p]: \veck[l]\ge 1}} \left(T_{h-1}(\veck - \unittensor_l) \cdot w_{\tau_l} \cdot \prod_{i\in[p]} r_{\tau_i, \tau_l}^{(\veck - \unittensor_l)[i]}\right).
\end{equation}
Since the number of $\veck$ in $T_h$ for any $h\in[n]$ is $\binom{h+p-1}{p-1} = \mathcal{O}(n^{p-1})$, the time complexity of \incwfomc{} is $\mathcal{O}(n^p)$, which is polynomial in $n$ for fixed $p$ (due to the fixed sentence $\sentence$).



\section{\texorpdfstring{\ouralgo{} for \ctwo{}}{\ouralgo{} for C2}}\label{sec:incrementalWFOMC3}
This section introduces our main contribution, \ouralgo{}, an efficient algorithm for computing WFOMC on \ctwo{} formulas.
The algorithm is extended to the \ctwo{} fragment with modulo counting quantifiers (\ctwomod{}) in the next section.

\subsection{Normal Form}\label{sec:ctwonormalform}

Instead of relying on the \wfomc{} reductions from \ctwo{} to \ufotwo{} as discussed in \Cref{sec:reductions}, \ouralgo{} directly computes \wfomc{} for \ctwo{} sentences in the following normal form:
\begin{equation}\label{eq:sctwo}
    \sentence = \Psi \land \bigwedge_{i\in[M]} \left(\forall x\exists^{=k_i} y: R_i(x,y) \right)\land\bigwedge_{j\in[N]} \left(\exists^{=u_j} x: U_j(x)\right),
\end{equation}
where $\Psi=\forall x\forall y: \fotwoformula(x,y)$ is a \ufotwo{} sentence, $R_i/2$ and $U_j/1$ are distinguished binary and unary predicates, respectively, appearing in $\Psi$, $k_i$ and $u_j$ are non-negative integers, and $M$ and $N$ are positive integers.
We abuse the name \sctwo{} to refer to the fragment of \ctwo{} sentences in the form of \cref{eq:sctwo}, but note that this form is slightly different from \cref{eq:ctwo_normal_form}, where $\Psi$ is a \ufotwo{} sentence instead of an \fotwocc{} sentence, and the unary counting quantified formulas $\exists^{=u_j} x: U_j(x)$ are present instead of being transformed into cardinality constraints (as in \cref{eq:ctwo_normal_form}).
We provide the transformation from any \ctwo{} sentence to this normal form in \Cref{app:modk_reduction} for completeness.

\begin{remark}\label{remark:lessthan1}
    It will turn out that \ouralgo{} can also directly handle counting quantifiers of the form $\exists^{\le k}$ without expanding them into a disjunction of $\exists^{=1}, \ldots, \exists^{=k}$ as done in the prior work~\cite{kuzelka2021weighted}.
    While the expansion to a disjunction is liftable (due to the fixed $k$), it introduces $k$ new counting quantifiers, significantly degrading the efficiency of these algorithms.
    For the sake of better readability, we will continue to present our algorithm only on \sctwo{} sentences.
    We will, however, revisit this point once in a while, mainly in a series of remarks.
\end{remark}



\subsection{Counting Types and Configurations}\label{sec:countingtype}

Consider an \sctwo{} sentence $\sentence$ as in \cref{eq:sctwo}.
We introduce the notion of \emph{counting-1-types} that extends the 1-types to also represent the satisfaction of constants w.r.t. counting quantified formulas.
\begin{definition}[Counting type and counting-1-type]\label{def:c1type}
    A \emph{counting-type} (\emph{c-type}) is a vector $\vecc\in\nat^{M}$ such that $\vecc[i] \in \{0, 1, \ldots, k_i\}$ for each $i\in[M]$.
    A \emph{counting-1-type} (c1-type) is a tuple $(\tau_i, \vecc)$, where $\tau_i$ is a 1-type and $\vecc$ is a c-type.
\end{definition}
Similar to 1-types, a constant $a$ in the domain $\domain$ realizes a c-type $\vecc$ in an interpretation $\mathcal{I}$ over $\domain$ if and only if for each $i\in[M]$, there are exactly $\vecc[i]$ constants $b\in\domain$ such that $R_i(a, b) \in \mathcal{I}$.
A constant $a$ realizes a c1-type $(\tau, \vecc)$ in $\mathcal{I}$ if it realizes the 1-type $\tau$, as well as the c-type $\vecc$, in $\mathcal{I}$.
The set of constants realizing the c1-type $(\tau, \vecc)$ in $\mathcal{I}$ is denoted as $\realizedele{\mathcal{I}}{\tau, \vecc}$.


Let $\tau_1, \ldots, \tau_p$ be all the 1-types of $\Psi$ in \cref{eq:sctwo}.
Denote by $\mathcal{C}$ and $\mathcal{T}$ the set of all c-types and c1-types of $\sentence$, respectively, i.e., $\mathcal{C} = \{\vecc\in\nat^M: \vecc[i] \in \{0, 1, \ldots, k_i\} \text{ for each } i\in[M]\}$ and $\mathcal{T} = \{(\tau_i, \vecc): i\in[p], \vecc\in\mathcal{C}\}$.
We will alternatively denote a c1-type $(\tau_i, \vecc)$ as $(i, \vecc)$, and view it as a vector of size $1 + M$, where the first entry is the index of the 1-type, i.e., $(i, \vecc) = (i, \vecc[1], \ldots, \vecc[M])$.
Let $\dimensions = p\times (k_1 + 1) \times \ldots \times (k_M + 1)$.
The counting-type configuration and counting-1-type configuration are defined similarly to the 1-type configuration.

\begin{definition}[Counting-1-type configuration]\label{def:c1typeconfig}
    The \emph{counting-1-type configuration} (\emph{c1-type configuration}) of an interpretation $\mathcal{I}$ of $\sentence$ is a non-negative integer tensor $\vecK\in \nat^{\dimensions}$, where $\vecK[(i, \vecc)] = |\realizedele{\mathcal{I}}{\tau_i, \vecc}|$ for each $i\in[p], \vecc\in\mathcal{C}$.
\end{definition}

\begin{example}\label{ex:c1type}
    Consider $\sentence = \sentence_{CG} \land \forall x\exists^{=2} y: E(x,y)$, which encodes a 2-regular 2-colored graph.
    The 1-types are $\tau_1$ and $\tau_2$ representing red and black vertices, respectively, as in \Cref{ex:1type}.
    The c-types are $(0, ), (1, )$, and $(2, )$.
    \Cref{fig:2-reg-2-color-example} illustrates two example models of $\sentence$ over the domain size of $10$.
    The c1-type configuration of these models is $\vecK\in\nat^{2\times 3}$ such that $\vecK[1, 2] = 5$, $\vecK[2, 2] = 5$, and $0$ for all other entries.
\end{example}

Given a c1-type configuration $\vecK\in\nat^{\dimensions}$, we write $\vecK[(i, \cdot)]$ for the summation of $\vecK[(i, \vecc)]$ over all c-types $\vecc\in\mathcal{C}$, i.e., $\vecK[(i, \cdot)] = \sum_{\vecc\in\mathcal{C}} \vecK[(i, \vecc)]$.
Then the 1-type configuration corresponding to a c1-type configuration $\vecK$ can be written as $(\vecK[(1, \cdot)], \ldots, \vecK[(p, \cdot)])$.
For example, for the c1-type configuration $\vecK$ in \Cref{ex:c1type}, the corresponding 1-type configuration is $(\vecK[(1, \cdot)], \vecK[(2, \cdot)]) = (5, 5)$.

\begin{remark}\label{remark:c1types}
    The c-types and c1-types are essentially the block types and cell types, respectively, defined in~\cite{wang2024lifted} for sampling models of \sctwo{} sentences.
    However, in~\cite{wang2024lifted}, the block types and cell types are also used for handling the existential quantifiers, which are dealt with by Skolemization in our case.\footnote{We could also handle the existential quantifiers directly by adding special dimensions to the c-types with the values being either $0$ or $1$. However, we found no significant performance improvement compared to Skolemization in our initial experiments.}
\end{remark}

\begin{remark}\label{remark:lessthan2}
    When dealing with counting quantifiers of the form $\forall\exists^{\le k}$, we use the same definitions of c-types with $k$ being the upper bound.
    We do not care about the case when a constant has more than $k$ neighbors, as such cases are not allowed in any model of the sentence.
\end{remark}

\subsection{Algorithm Overview}

Now we give an overview of \ouralgo{}.
Fix an \sctwo{} sentence $\sentence$ as in \cref{eq:sctwo}, a weight function $(\weight, \negweight)$, and let $[n]$ be the domain.

Define $T_h(\vecK)$ as the summation of $\typeweight{\mu}$ over all models $\mu\in\fomodels{\Psi}{h}$ with c1-type configuration $\vecK$.
Clearly, the $\symwfomc(\sentence, n, \weight, \negweight)$ can be computed by summing $T_n(\vecK)$ over all c1-type configurations $\vecK$ 
that correspond to models satisfying the condition $\bigwedge_{i\in[M]} \forall x\exists^{=k_i} y: R_i(x,y)\land\bigwedge_{j\in[N]} \exists^{=u_j} x: U_j(x)$.

Note that $\exists^{=u_j} x: U_j(x)$ holds in a model $\mu$ if and only if there are exactly $u_j$ ground atoms of $U_j$ in $\mu$.
Moreover, the number of ground atoms of $U_j$ in $\mu$ is determined by the 1-type configuration of $\mu$.
Then the condition $\exists^{=u_j} x: U_j(x)$ holds if and only if the following holds:
\begin{equation}\label{eq:countingcondition1}
    \sum_{i\in[p]} \indicator{U_j(x)\in\tau_i} \cdot \vecK[(i,\cdot)] = u_j,
\end{equation}
where $\indicator{\cdot}$ is the indicator function.
The condition of $\bigwedge_{i\in[M]}\forall x\exists^{=k_i} y: R_i(x,y)$ is more straightforward from the definition of c1-type configuration: A model with c1-type configuration $\vecK$ satisfies $\bigwedge_{i\in[M]}\forall x\exists^{=k_i} y: R_i(x,y)$ if and only if 
\begin{equation}\label{eq:countingcondition2}
    \sum_{i\in[p]} \vecK[i, k_1, \ldots, k_M] = n,
\end{equation}
that is, all constants in the domain have the c-type $(k_1, \ldots, k_M)$.

\begin{remark}\label{remark:lessthan3}
    For $\exists^{\le u_j} x: U_j(x)$, the condition \cref{eq:countingcondition1} is changed to $\sum_{i\in[p]} \indicator{U_j(x)\in\tau_i} \cdot \vecK[(i, \cdot)] \le u_j$.
    For $\forall x\exists^{\le k_j}y:R_j(x,y)$, \cref{eq:countingcondition2} is replaced by $\sum_{i\in[p]} \sum_{0\le c_j\le k_j} \mathbf K[(i, k_1, \dots, c_j, \dots, k_M)] = n$, and multiple such conditions for different $j$ are combined together.
\end{remark}

Let $\mathcal{K}$ be the set of all c1-type configurations $\vecK$ such that \cref{eq:countingcondition1} and \cref{eq:countingcondition2} hold.
Then we have
\begin{equation}\label{eq:finalwfomc}
    \symwfomc(\sentence, n, \weight, \negweight) = \sum_{\substack{\vecK\in\nat^{\dimensions}: |\vecK| = n, \vecK\in\mathcal{K}}} T_n(\vecK).
\end{equation}
\ouralgo{} computes $T_h(\vecK)$ incrementally, similar to \incwfomc{}, in a dynamic programming fashion, as presented in \Cref{alg:ouralgo}.

\begin{algorithm}[tbp]
    \DontPrintSemicolon
    \caption{\ouralgo($\sentence, n, \weight, \negweight$)}\label{alg:ouralgo}
    \KwIn{A \sctwo{} sentence $\sentence$, a positive integer $n$, and weights $(\weight, \negweight)$}
    \KwOut{$\symwfomc(\sentence, n, \weight, \negweight)$}
    Initialize $T$ as an empty dictionary with default value 0\;
    \tcp{Base case for $h=1$}
    \ForEach{$l\in[p]$}{
       $T[\unittensor_{(l, \expctype{\tau_l}{x})}] \gets w_{\tau_l}$\;
    }
    \tcp{General case for $h>1$}
    \ForEach{$h = 2, \ldots, n$\label{line:forh}}{
        Initialize $T_{new}$ as an empty dictionary with default value 0\;
        \ForEach{$(\vecK_{old}, W_{old}) \in T$}{
            \ForEach{$l\in[p]$}{
                \ForEach{$(\vecK, W)\in F(\vecK_{old}, l)$ such that $|\vecK| = h$}{
                    $T_{new}[\vecK] \gets T_{new}[\vecK] + T[\vecK_{old}]\cdot W \cdot w_{\tau_l} $ \;
                }
            }
        }
        $T \gets T_{new}$\;
    }
    \Return \cref{eq:finalwfomc}\;
\end{algorithm}

\subsubsection{Base Case}

When $h=1$, there is a single constant $h$ in the domain.
Any model in $\fomodels{\Psi}{1}$ can be written as $\tau_l(h)$ for some $l\in[p]$, where the constant $h$ realizes the 1-type $\tau_l$.
For a set of literals $L$ over logical variables $x$ and $y$, define
\begin{equation}\label{eq:expctype} 
    \expctype{L}{x} = \Big(\indicator{R_1(x,\cdot)\in L}, \ldots, \indicator{R_M(x,\cdot)\in L}\Big),
\end{equation}
as a vector of size $M$ where $R_i(x,\cdot)\in L$ means that $L$ contains the literal $R_i(x,x)$ or $R_i(x,y)$.
Then the c1-type configuration of $\tau_l(h)$ is $\unittensor^\dimensions_{(l, \expctype{\tau_l}{x})}$, where $\unittensor^\dimensions_{(l, \vecc)}$ is the unit tensor of size $\dimensions$ with only the entry $(l, \vecc)$ being $1$.
The weight of the model $\tau_l(h)$ is $w_{\tau_l}$ defined in \cref{eq:weights}.
Thus, $T_1$ is initialized as $T_1(\unittensor^\dimensions_{(l, \expctype{\tau_l}{x})}) = w_{\tau_l}$ for each $l\in[p]$.

\subsubsection{General Case}

Consider the case $h>1$.
Following the computation in \incwfomc{}, for any $l\in[p]$, a model $\mu'\in\fomodels{\Psi}{h-1}$ can be always extended to a (possibly empty) set of models $\mu\in\fomodels{\Psi}{h}$ by adding the 1-type $\tau_l(h)$ and the 2-tables $\pi_{al}^{\mu'}$ between the constants $a\in[h-1]$ and $h$ as shown in \cref{eq:incwfomcmodel}.
However, now the c1-type configuration $\vecK$ of $\mu$ depends not only on $\mu'$ and $l$, but also on \emph{the choices of $\pi^{\mu'}_{al}$} for each $a\in[h-1]$.
First, the choices of $\pi^{\mu'}_{al}$ determine the c-type of $h$ in $\mu$.
Moreover, the c-types of other constants $a$ in $\mu$ might also depend on the choices of $\pi^{\mu'}_{al}$. For example, the possible ground literal $R_j(a,h)$ in $\pi^{\mu'}_{al}(a,h)$ will affect the $j$-th entry of the c-type of $a$ in $\mu$, and thus affect the c1-type configuration of $\mu$.

To this end, we need a fine-grained incremental computation of $T_h(\vecK)$.
Recall that $M_{\mu',\tau_l}$ is the set of all models $\mu\in\fomodels{\Psi}{h}$ constructed from $\mu'\in\fomodels{\Psi}{h-1}$ and $\tau_l(h)$.
We define $W_{\mu',\tau_l}(\vecK)$ as the weight of all models $\mu\in M_{\mu',\tau_l}$ with c1-type configuration $\vecK$.
In the next subsection, we will show that $W_{\mu',\tau_l}(\vecK)$ can be written as
\begin{equation}\label{eq:rewriteW}
    W_{\mu',\tau_l}(\vecK) = \typeweight{\mu'} \cdot w_{\tau_l} \cdot F_{\vecK', \tau_l}(\vecK),
\end{equation}
where $\vecK'$ is the c1-type configuration of $\mu'$, and $F_{\vecK', \tau_l}(\vecK)$ is a function conditioned on $\vecK'$ and $\tau_l$ that maps $\vecK$ to a real number.
Then we can compute $T_h(\vecK)$ as
\begin{equation}\label{eq:incwfomc2}
    T_h(\vecK) = \sum_{l\in[p]}\ \sum_{\substack{\vecK'\in\nat^\dimensions: |\vecK'| = h-1}} T_{h-1}(\vecK')\cdot w_{\tau_l} \cdot F_{\vecK', \tau_l}(\vecK).
\end{equation}
\Cref{fig:running_example} provides an illustration of this incremental computation from $T_5(\vecK')$ to $T_6(\vecK)$ for the 2-regular 2-colored graph example.

\begin{figure}[tbp]
    \centering
    \includegraphics[width=\textwidth]{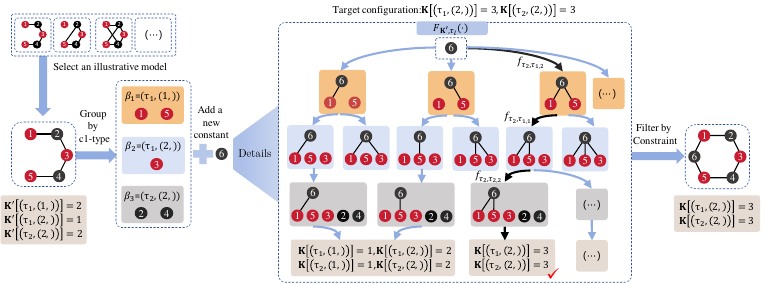}
    \caption{An illustration of the incremental computation of $T_6[\vecK]$ from $T_5[\vecK']$ for the 2-regular 2-colored graph example.
    Each model from $T_5[\vecK']$ (left) can be extended by adding a new vertex 6 with a specific 1-type (black, $\tau_2$ in this case) to form new models in $T_6[\vecK]$ (right).
    For clarity of presentation, only one example model and its extension are shown, as other models with the same c1-type configuration $\vecK'$ will be treated similarly.
    The extension process involves adding edges between the new vertex and existing vertices, whose contributed weights are computed by the function $F_{\vecK', l}(\vecK)$ (center).
    The computation of $F_{\vecK', l}(\vecK)$ is further performed incrementally through c-type updates as described in \Cref{sec:computingF}.
    Intuitively, we view the elements in the same c1-type as a group, and use the c-type update to capture the overall effect of adding edges from this group to the new vertex.
    }
    \label{fig:running_example}
\end{figure}

\subsection{Computing \texorpdfstring{$W_{\mu',\tau_l}(\vecK)$}{W underscore mu prime tau l vecK}}
We show how to compute $W_{\mu',\tau_l}(\vecK)$, which is the main part of \ouralgo{}.
We first need the notion of \emph{c-type update}.

\subsubsection{C-type Update}\label{sec:c_type_update}

Recall that by \cref{eq:expctype}, $\expctype{\pi}{x}$ and $\expctype{\pi}{y}$ are the vectors of the numbers of $R_i(x,y)$ and $R_i(y,x)$, respectively, in the 2-table $\pi$.

\begin{definition}[C-type update]\label{def:c-type-update}
    Let $\tau_i$ and $\tau_j$ be two 1-types of $\Psi$.
    A \emph{$d$-order c-type update} (or simply a \emph{$d$-order c-update}) of $\tau_i$ on $\tau_j$ is a tuple $(\vecu, \vecU)$, where
    \begin{equation*}
        \vecu = \sum_{\pi\in\mathcal{B}} \expctype{\pi}{x} \qquad \text{ and } \qquad \forall \vecsig\in\{0,1\}^M, \vecU[\vecsig] = |\{\pi\in\mathcal{B}: \expctype{\pi}{y} = \vecsig\}|,
    \end{equation*}
    for a d-tuple of 2-tables $\mathcal{B}\in\Pi_{\tau_i, \tau_j}^d$ (here $\Pi_{\tau_i, \tau_j}^d$ is the Cartesian product of $\Pi_{\tau_i, \tau_j}$ with $d$ copies).
    The tuple of 2-tables $\mathcal{B}$ is called a $d$-order $(\vecu, \vecU)$-\emph{c-type updater} (or simply $(\vecu, \vecU)$-updater) of $\tau_i$ on $\tau_j$.
\end{definition}

Suppose that there is a constant $a$ with c1-type $(\tau_i, \vecc_a)$ and $d$ constants $b_1, \ldots, b_d$ with the same c1-type $(\tau_j, \vecc_b)$.
Intuitively, a $d$-order c-type update of $\tau_i$ on $\tau_j$ describes how the c-type of $a$ and the c-type configuration of $b_1, \dots, b_d$ are updated if we choose the 2-tables between $a$ and $b_1, \ldots, b_d$ to be some $\mathcal{B}\in\Pi_{\tau_i, \tau_j}^d$.
Formally, given a $d$-order c-type update $(\vecu, \vecU)$ of $\tau_i$ on $\tau_j$, the c-type of $a$ updated by $(\vecu, \vecU)$ is $\vecc_a + \vecu$, and the updated c-type configuration of $b_1, \ldots, b_d$ is given by $\update(\vecc_b, \vecU)\in\nat^{(k_1 + 1) \times \ldots \times (k_M + 1)}$, where
\begin{equation}\label{eq:updateconf}
    \update(\vecc_b, \vecU)[\vecc] = \sum_{\vecsig\in\{0, 1\}^M: \vecc_b + \vecsig = \vecc} \vecU[\vecsig],
\end{equation}
for each $\vecc\in\mathcal{C}$.

\subsubsection{\texorpdfstring{Rewriting $W_{\mu',\tau_l}(\vecK)$}{Rewriting W mu prime U vecK}}

Now we are ready to show how to rewrite $W_{\mu',\tau_l}(\vecK)$ into the form of \cref{eq:rewriteW}.
Recall that $\realizedele{\mu'}{\tau,\vecc}$ is the set of constants realizing the c1-type $(\tau, \vecc)$ in $\mu'$, and $\mathcal{C}$ and $\mathcal{T}$ are the sets of all c-types and c1-types of $\sentence$, respectively.
We can rewrite \cref{eq:incwfomcmodel} as
\begin{equation}\label{eq:incwfomcmodel2}
    \mu = \mu' \cup \tau_l(h) \cup \bigcup_{(\tau,\vecc)\in\mathcal{T}} \left(\bigcup_{a\in \realizedele{\mu'}{\tau, \vecc}} \pi_{\tau a}(a,h)\right),
\end{equation}
where $\pi_{\tau a}\in \Pi_{\tau \tau_l}$ for each $a\in \realizedele{\mu'}{\tau, \vecc}$.
Let $(\vecu_{\tau, \vecc}, \vecU_{\tau, \vecc})$ be a $|\realizedele{\mu'}{\tau, \vecc}|$-order c-type update of $\tau_l$ on $\tau$ for each $(\tau, \vecc)\in\mathcal{T}$.
Then the c1-type configuration $\vecK$ of $\mu$ is determined by the 1-type $\tau_l$ and the c-type updates $(\vecu_{\tau, \vecc}, \vecU_{\tau, \vecc})$, i.e., for each $(\tau, \vecc)\in\mathcal{T}$,
\begin{equation}\label{eq:computeK}
    \vecK[(\tau, \vecc)] = \begin{cases}
        \sum_{\vecc' \in\mathcal{C}} \update(\vecc', \vecU_{\tau, \vecc'})[\vecc] & \text{if } \tau \neq \tau_l,\\
        \sum_{\vecc' \in\mathcal{C}} \update(\vecc', \vecU_{\tau, \vecc'})[\vecc] + \indicator{\vecc = \expctype{\tau_l}{x} + \sum_{(\tau', \vecc')\in\mathcal{T}} \vecu_{\tau', \vecc'}} & \text{if } \tau = \tau_l.
    \end{cases}
\end{equation}
The term $\update(\vecc', \vecU_{\tau, \vecc'})[\vecc]$ accounts for the number of constants with the 1-type $\tau$, whose c-type is updated from $\vecc'$ in $\mu'$ to $\vecc$ in $\mu$ by the c-type updates $(\vecu_{\tau, \vecc'}, \vecU_{\tau, \vecc'})$.
Moreover, the c-type of the new constant $h$ is $\expctype{\tau_l}{x} + \sum_{(\tau', \vecc')\in\mathcal{T}} \vecu_{\tau', \vecc'}$, contributing $1$ to the entry $\vecK[(\tau_l, \expctype{\tau_l}{x} + \sum_{(\tau', \vecc')\in\mathcal{T}} \vecu_{\tau', \vecc'})]$.
Note that $|\realizedele{\mu'}{\tau, \vecc}| = \vecK'[(\tau, \vecc)]$.
Thus, the value of $\vecK$ does not depend on the actual model $\mu'$ but on its c1-type configuration $\vecK'$. 

\begin{example}\label{ex:f-computation-detail}
    Consider the 2-regular 2-colored graph example from \Cref{ex:c1type}.
    The update of $W_{\mu', \tau_2}(\vecK)$ from an example model $\mu'$ with c1-type configuration $\vecK'$ is illustrated in \Cref{fig:running_example}.
    Rather than processing the connections between the new vertex $6$ and all vertices simultaneously, the algorithm groups them by c1-type as shown at the left of the figure.
    The c1-type configuration is updated incrementally through c-type updates from each group, with the overall weight computed as the product of weights from these updates (shown by the black arrows in the center).
\end{example}

Let $\mathfrak{B}_{\tau_i, \tau_j, d}(\vecu, \vecU)$ denote the set of all $d$-order $(\vecu, \vecU)$-updaters of $\tau_i$ on $\tau_j$.
We define $f_{\tau_i, \tau_j, d}(\vecu, \vecU)$ as the sum of the weights of $\mathcal{B}$ over $\mathfrak{B}_{\tau_i, \tau_j, d}(\vecu, \vecU)$, that is,
\begin{equation}\label{eq:f_definition}
    f_{\tau_i, \tau_j, d}(\vecu, \vecU) = \sum_{\mathcal{B} \in \mathfrak{B}_{\tau_i, \tau_j, d}(\vecu, \vecU)}\ \prod_{\pi\in\mathcal{B}} \typeweight{\pi}.
\end{equation}
Note that a c-type update only depends on the number of each 2-table in $\mathcal{B}$, rather than the actual $\mathcal{B}$.
This will help the efficient computation of $f_{\tau_i, \tau_j, d}(\vecu, \vecU)$, which we leave to \Cref{sec:computef}.

Suppose that $\vecbeta_1, \ldots, \vecbeta_D$ are all the c1-types in $\mathcal{T}$, and denote by $\vecbeta_i^1$ and $\vecbeta_i^2$ the first and second entries, i.e., the 1-type and the c-type, of $\vecbeta_i$ respectively.
Let $\mathfrak{U}_{\vecK}$ be the set of c-type updates $((\vecu_{\vecbeta_1}, \vecU_{\vecbeta_1}), \ldots, (\vecu_{\vecbeta_D}, \vecU_{\vecbeta_D}))$ such that the resulting c1-type configuration as in \cref{eq:computeK} is $\vecK$.
Then by the same argument (\cref{eq:modelweight,eq:weightprod}) as \incwfomc{}, we can rewrite $W_{\mu',\tau_l}(\vecK)$ into
\begin{equation}\label{eq:rewriteW2}
    W_{\mu',\tau_l}(\vecK) = \typeweight{\mu'} \cdot w_{\tau_l} \cdot F_{\vecK', \tau_l}(\vecK),
\end{equation}
where $F_{\vecK', \tau_l}(\vecK)$ is defined as
\begin{equation}\label{eq:FKl}
    F_{\vecK', \tau_l}(\vecK) = \sum_{((\vecu_{\vecbeta_1}, \vecU_{\vecbeta_1}), \ldots, (\vecu_{\vecbeta_D}, \vecU_{\vecbeta_D}))\in\mathfrak{U}_{\vecK}}\ \prod_{i\in[D]} f_{\tau_l, \vecbeta_i^1, \vecK'[\vecbeta_i]}(\vecu_{\vecbeta_i}, \vecU_{\vecbeta_i}).
\end{equation}

\subsection{\texorpdfstring{Fast Computation of $F_{\vecK', \tau_l}(\vecK)$}{Fast Computation of F vecK prime tau l vecK}}
\label{sec:computingF}

Note that the dimensions of $\vecu$ and $\vecU$ in any c-type update, as well as the number $D$, are all constants determined by the sentence $\sentence$ (recall that $D$ is the number of c1-types).
It is easy to check that the size of $\mathfrak{U}_\vecK$ is polynomial in $n$.
Thus one can directly compute $F_{\vecK', \tau_l}(\vecK)$ by \cref{eq:FKl}, which was the approach used in~\cite{wang2024lifted}.
Though theoretically tractable, directly computing $F_{\vecK', \tau_l}(\vecK)$ involves a large exponent on $n$ in the complexity, prohibiting the practical use in real applications.
\ouralgo{} rather uses a more efficient way again in an incremental manner.
We omit the subscripts $\vecK'$ and $\tau_l$ in $F_{\vecK', \tau_l}(\vecK)$ in this subsection as there is no confusion.

For any $t\in[D]$, let $G_t$ be the set of tuples $((\vecu^*, \vecK^*), W)$ such that there is a list of c-type updates $((\vecu_{\vecbeta_1}, \vecU_{\vecbeta_1}), \ldots, (\vecu_{\vecbeta_t}, \vecU_{\vecbeta_t}))$ such that $\vecu^* = \sum_{i\in[t]} \vecu_{\vecbeta_i}$ and for each $(\tau, \vecc)\in\mathcal{T}$,
\begin{equation*}
    \vecK^*[(\tau, \vecc)] = \sum_{j\in[t]: \vecbeta_j^1 = \tau} \update(\vecbeta_j^2, \vecU_{\vecbeta_j})[\vecc],
\end{equation*}
and $W$ is the summation of $\prod_{i\in[t]} f_{\tau_l, \vecbeta_i^1, \vecK'[\vecbeta_i]}(\vecu_{\vecbeta_i}, \vecU_{\vecbeta_i})$ over all such lists of c-type updates.
Intuitively, the sets $G_t$ contain the information of how the c1-type configuration $\vecK$ is updated by the partial c-type updates $((\vecu_{\vecbeta_1}, \vecU_{\vecbeta_1}), \ldots, (\vecu_{\vecbeta_t}, \vecU_{\vecbeta_t}))$, and the corresponding weight that the c-type updates contribute to the final value of $F(\vecK)$.
Aligned with $T_h$, we also write $G_t((\vecu^*, \vecK^*)) = W$.
Then the value of $F(\vecK)$ can be computed from $G_D$ as follows:
\begin{equation}\label{eq:compute_F}
    F(\vecK) = \sum_{\vecu^*\in\mathcal{C}, \vecK^*\in\nat^\dimensions:\\ \vecK^* + \unittensor^\dimensions_{(l, \expctype{\tau_l}{x}+\vecu^*)} = \vecK} G_D((\vecu^*, \vecK^*)),
\end{equation}
where $(l, \expctype{\tau_l}{x} + \vecu^*)$ is the c1-type of the new constant $h$ in $\mu$.

\begin{algorithm}[tbp]
    \DontPrintSemicolon
    \caption{\texorpdfstring{$F(\vecK', \tau_l)$}{ConfigWeight(K prime, tau l)}}\label{alg:ConfigWeight}
    Initialize $G$ as an empty dictionary with default value 0\;
    $G[(\boldsymbol{0}^M, \boldsymbol{0}^D)] \gets 1$\;
    \ForEach{$t \in [D]$\label{line:fort}}{
        Initialize $G_{new}$ as an empty dictionary with default value 0\;
        \ForEach{$((\vecu^*_{old}, \vecK^*_{old}), W_{old}) \in G$}{
            \ForEach{$((\vecu, \vecU), W) \in f(\tau_l, \vecbeta_t^1, \vecK'[\vecbeta_t])$}{
                \If{$\vecu^*_{old} + \vecu\in \mathcal{C}$}{
                    $\vecu^*_{new} \gets \vecu^*_{old} + \vecu$; $\vecK^*_{new} \gets$ \cref{eq:updateK}\;
                    $G_{new}[(\vecu^*_{new}, \vecK^*_{new})] \gets G_{new}[(\vecu^*_{new}, \vecK^*_{new})] + W_{old} \cdot W$\;
                }
                }
        }
        $G \gets G_{new}$\;
    }
    \Return \cref{eq:compute_F} with $G$\;
\end{algorithm}

The computation of $G_t$ is done incrementally as shown in \Cref{alg:ConfigWeight}, similar to computing $T_h$ in \Cref{sec:incwfomc}.
For $t = 0$, $G_0$ is initialized as $G_0 = \{((\boldsymbol{0}^M, \boldsymbol{0}^D), 1)\}$, where $\boldsymbol{0}^M$ and $\boldsymbol{0}^D$ are the zero vectors of appropriate dimensions.
For $t \in [D]$, we compute $G_t$ from $G_{t-1}$ by enumerating all possible $\vecK'[\vecbeta_t]$-order c-type updates $(\vecu, \vecU)$ of $\tau_l$ on $\vecbeta_t^1$.
According to the definition of $G_t$, for any $((\vecu^*_{old}, \vecK^*_{old}), W_{old}) \in G_{t-1}$, we can obtain a new tuple $((\vecu^*, \vecK^*), W)$ in $G_t$ by applying the c-type update $(\vecu, \vecU)$, where $\vecu^* = \vecu^*_{old} + \vecu$ and for each $(\tau, \vecc)\in\mathcal{T}$,
\begin{equation}
    \label{eq:updateK}
    \begin{aligned}
        \vecK^*[(\tau, \vecc)] = \begin{cases}
            \vecK^*_{old}[(\tau, \vecc)] + \update(\vecbeta_t^2, \vecU)[\vecc] & \text{if } \tau = \vecbeta_t^1,\\
            \vecK^*_{old}[(\tau, \vecc)] & \text{otherwise,}
        \end{cases}
    \end{aligned}
\end{equation}
and $W = W_{old} \cdot f_{\tau_l, \vecbeta_t^1, \vecK'[\vecbeta_t]}(\vecu, \vecU)$.
One example of this procedure for the 2-regular 2-colored graph example is illustrated in the center box of \Cref{fig:running_example}.

\subsection{\texorpdfstring{Computing $f_{\tau_i, \tau_j, d}(\vecu, \vecU)$}{Computing f tau i tau j d vecu vecU}}\label{sec:computef}

The remaining part is to compute $f_{\tau_i, \tau_j, d}(\vecu, \vecU)$ for any $d$-order c-type update $(\vecu, \vecU)$ of $\tau_i$ on $\tau_j$.
From \Cref{def:c-type-update}, we know that a $d$-order c-type update of $\tau_i$ on $\tau_j$ is induced by $\expctype{\pi}{x}$ and $\expctype{\pi}{y}$ for a d-tuple of 2-tables $\pi$.
Thus, we group the 2-tables in $\Pi_{\tau_i, \tau_j}$ by their c-type increments on both sides, denoted by
\begin{equation*}
    \Pi_{\tau_i,\tau_j,\vect,\vect'} = \{\pi\in\Pi_{\tau_i,\tau_j}: \expctype{\pi}{x} = \vect, \expctype{\pi}{y} = \vect'\},
\end{equation*}
and define $r_{i,j,\vect,\vect'} = \sum_{\pi\in\Pi_{\tau_i,\tau_j,\vect,\vect'}} \typeweight{\pi}$.
Then the value of $f_{\tau_i, \tau_j, d}(\vecu, \vecU)$ can be computed by aggregating weights over combinations of c-type increments $(\vect, \vect')$ that yield the given c-type update $(\vecu, \vecU)$.
Further observe that $f_{\tau_i, \tau_j, d}(\vecu, \vecU)$ can be obtained from $f_{\tau_i, \tau_j, d-1}(\cdot, \cdot)$ by adding one more 2-table.
This suggests a dynamic programming approach to compute $f_{\tau_i, \tau_j, d}(\vecu, \vecU)$ as we did for $F_{\vecK', \tau_l}(\vecK)$.

\begin{algorithm}[tbp]
\caption{$f(\tau_i, \tau_j, d)$}
\DontPrintSemicolon\label{alg:UpdateWeight}
    Initialize $H$ as an empty dictionary with default value 0\;
    $H[(\boldsymbol{0}^M, \boldsymbol{0}^{2^M})] \gets 1$\;
    \ForEach{$s \in [d]$}{
        Initialize $H_{new}$ as an empty dictionary with default value 0\;
        \ForEach{$((\vecu_{old}, \vecU_{old}), W_{old}) \in H$}{
            \ForEach{$\vect, \vect' \in \{0, 1\}^M$}{
                \If{$\vecu_{old} + \vect \in \mathcal{C}$}{
                    $\vecu_{new} \gets \vecu_{old} + \vect$; $\vecU_{new} \gets \vecU_{old} + \unittensor^{2^M}_{\vect'}$\;
                    $H_{new}[(\vecu_{new}, \vecU_{new})] \gets H_{new}[(\vecu_{new}, \vecU_{new})] + W_{old} \cdot r_{i,j}^{\vect, \vect'}$\;
                }
            }
        }
        $H \gets H_{new}$\;
    }
    \Return $H$\;
\end{algorithm}

The detailed procedure is shown in \Cref{alg:UpdateWeight}, where we denote the dimension $\overbrace{2\times 2\times\dots \times 2}^M$ by $2^M$ for simplicity.
Let $H_s$ be the set of tuples $((\vecu, \vecU), W)$ such that there is a $s$-order c-type update $(\vecu, \vecU)$ of $\tau_i$ on $\tau_j$ with weight $W = f_{\tau_i, \tau_j, s}(\vecu, \vecU)$.
For the base case $s = 0$, we have $H_0[(\boldsymbol{0}^M, \boldsymbol{0}^{2^M})] = 1$.
When $s > 0$, we compute $H_s$ from the tuples $((\vecu_{old}, \vecU_{old}), W_{old})$ in $H_{s-1}$ by enumerating all possible $\vect, \vect' \in \{0, 1\}^M$, and updating the value of $H_s[(\vecu, \vecU)]$ such that $\vecu_{old} + \vect = \vecu$ and $\vecU_{old} + \unittensor^{2^M}_{\vect'} = \vecU$ by adding $W_{old} \cdot r_{i,j}^{\vect, \vect'}$.
Finally, we return $H_d$, which contains the values of $f_{\tau_i, \tau_j, d}(\vecu, \vecU)$ for all possible c-type updates $(\vecu, \vecU)$.



\subsection{Complexity Analysis and Practical Optimizations}\label{sec:complexity_optimization}

Since the sentence $\sentence$ is fixed, 
the numbers of 1-types, c-types, and c1-type dimensions are all independent of the domain size $n$. 
As a result, the dynamic programs in \ouralgo{} only need to maintain polynomially many configurations in $n$.

\subsubsection{Complexity Analysis}\label{sec:complexity_analysis}

\begin{theorem}\label{thm:complexity}
    Given a sentence $\sentence$ in the normal form as in \cref{eq:sctwo}, \ouralgo{} computes $\symwfomc(\sentence, n, \weight, \negweight)$ in time $O(n^{2\cdot D + 2^M - 2})$, where $D = p \cdot \prod_{i=1}^{M}(k_i + 1)$.
\end{theorem}
\begin{proof}
    We first show that the time complexity of computing $f_{\tau_i, \tau_j, d}(\vecu, \vecU)$ is $O(n^{2^M - 1})$: The number of all possible $(\vecu, \vecU)$ is at most $D\cdot n^{2^M - 1}$, and each for-loop at \Cref{alg:UpdateWeight} iterates over $O(n^{2^M - 1})$ entries in $H$ and $O(1)$ entries in $r$.

    Then we consider the time complexity of computing $F_{\vecK', \tau_l}(\vecK)$.
    The number of all possible $(\vecu^*, \vecK^*)$ in $G_t$ is at most $D\cdot n^{D-1}$ for each $t\in[D]$.
    In each for-loop at \Cref{line:fort} in \Cref{alg:ConfigWeight}, we enumerate all possible $D\cdot n^{D-1}$ tuples in $G_{t-1}$ and $O(n^{2^M - 1})$ entries from $f$, resulting in totally $O(n^{D + 2^M - 2})$ iterations for the whole for-loop.

    Finally, there are totally $O(n^{D-1})$ c1-type configurations in $T_h$ for each $h\in[n]$, and each for-loop at \Cref{line:forh} in \Cref{alg:ouralgo} iterates over $O(n^{D-1})$ entries in $T$, $p$ 1-types and $O(n^{D + 2^M - 2})$ entries from $F$, resulting in totally $O(n^{2\cdot D + 2^M - 2})$ iterations for the whole algorithm.
\end{proof}

\begin{remark}
    Compared to the previous best algorithm~\cite{KR2024-64} with time complexity $O(n^{1 + 3M + p\prod_{i=1}^{M}(k_i^2 + 2k_i + 1)})$\footnote{The original paper analyzed the complexity for general \ctwo{} sentences, where the authors also used the same technique to reduce the input sentence to the normal form as in~\cref{eq:sctwo}. Therefore, the complexity presented for the normal form is easily derived from the original one.}, \ouralgo{} improves the exponent on $n$ from $O(p\cdot \prod_{i=1}^{M}(k_i^2 + 2k_i + 1))$ to $O(2\cdot p\cdot \prod_{i=1}^{M}(k_i + 1) + 2^M - 2)$.
    This is a significant improvement even for small $k_i$'s.
\end{remark}

\begin{remark}\label{remark:lessthan4}
    Involving $\exists^{\leq k}$ quantifiers in the input sentence of \ouralgo{} does not increase the time complexity, as the computation of $T_h$ remains the same.
\end{remark}

\subsubsection{Practical Optimizations}
\label{sec:optimization}
When implementing \ouralgo{}, we further utilize several optimizations to speed up the computation in practice.

\paragraph{Top-Down Compute $T_n$ with Memoization}
Instead of computing $T_n$ in a bottom-up manner as in \Cref{alg:ouralgo}, we can compute $T_n(\vecK)$ in a top-down manner with memoization.
By \cref{eq:finalwfomc}, only $\vecK\in\mathcal{K}$ contributes $\symwfomc(\sentence, n, \weight, \negweight)$.
Thus, we can only compute $T_n(\vecK)$ for each $\vecK\in\mathcal{K}$ to avoid unnecessary computations.
This is particularly useful when $\mathcal{K}$ is much smaller than $n^{D-1}$, which is often the case in practice.

\paragraph{Filter Invalid C-type Updates}
Incrementally computing $F_{\vecK', \tau_l}(\vecK)$ involves applying many c-type updates.
However, not all c-type updates are valid for a given $\vecK'$.
From \cref{eq:updateconf}, only $\sum_{\vecc\in\mathcal{C}} \update(\vecc_b, \vecU_b)[\vecc] = b$ (i.e., $\vecU[\vecsig] > 0$ only if $\vecc_b + \vecsig\in\mathcal{C}$) can result in a valid c1-type configuration.
Thus, when enumerating all possible c-type updates in \Cref{alg:ConfigWeight}, we can filter out those invalid ones to reduce the number of iterations.

\paragraph{Precompute $r_{i,j,\vect,\vect'}$ and Cache $f_{\tau_i, \tau_j, d}$ Results}
When computing $f_{\tau_i, \tau_j, d}(\vecu, \vecU)$ in \Cref{alg:UpdateWeight},
we need to enumerate all possible $\vect, \vect' \in \{0, 1\}^M$ and access the corresponding $r_{i,j, \vect, \vect'}$.
However, most pairs of $(\vect, \vect')$ do not correspond to any 2-table in $\Pi_{\tau_i, \tau_j}$, leading to $r_{i,j, \vect, \vect'} = 0$.
To avoid unnecessary computations, we precompute all non-zero $r_{i,j, \vect, \vect'}$ and store them in a dictionary.
Furthermore, since the result of $f_{\tau_i, \tau_j, d}(\vecu, \vecU)$ only depends on $\tau_i$, $\tau_j$, and $d$, we cache the results of $f_{\tau_i, \tau_j, d}$ calls in a global dictionary.
When $f(\tau_i, \tau_j, d)$ is invoked, we first find the largest $d' \leq d$ such that $f_{\tau_i, \tau_j, d'}$ is already cached, and then compute $f_{\tau_i, \tau_j, d}$ incrementally from $f_{\tau_i, \tau_j, d'}$.

\section{Extension to Modulo Counting Quantifiers} \label{sec:modulo_counting} 
In this section, we extend \ouralgo{} to the fragment of \ctwomod{}, which supports the modulo counting quantifiers $\exists^{=r, k}, \exists^{\geq r, k}$, and $\exists^{\leq r, k}$.
The normal form we work with is
\begin{equation}\label{eq:sctwo_mod}
    \sentence = \Psi \land \bigwedge_{i\in[M']} \left(\forall x\exists^{=r'_i, k'_i} y: R'_i(x,y) \right)\land\bigwedge_{j\in[N']} \left(\exists^{=v'_j, u'_j} x: U'_j(x)\right),
\end{equation}
where $\Psi$ is a sentence in \sctwo{} (i.e., \cref{eq:sctwo}), $R'_i$ are binary predicates, $U'_j$ are unary predicates, $k'_i,u'_j\ge 1$, and
$r'_i\in\{0,\ldots,k'_i-1\}$, $v'_j\in\{0,\ldots,u'_j-1\}$.
Following a similar reduction as in \Cref{sec:reductions}, any \ctwomod{} sentence can be transformed into the normal form as in \cref{eq:sctwo_mod}, while preserving the data complexity of the WFOMC problem.
The details are provided in \Cref{app:modk_reduction}.

To extend the c-types defined in \cref{sec:countingtype}, let $M_{\mathrm{mod}}=M+M'$. 
We denote the set of all \emph{modulo c-types} of $\sentence$ by $\mathcal{C}_{\mathrm{mod}}$:
\begin{align}\label{eq:modulo_ctypes}
\mathcal{C}_{\mathrm{mod}} = 
    \left\{
        \vecc\in \mathbb{N}^{M_{\mathrm{mod}}}:
        \begin{array}{l} 
            \vecc[i]\in\{0,1,\ldots,k_i\} \text{ for each } i\in[M],\\
            \vecc[M+i]\in\{0,1,\ldots,k'_i - 1\} \text{ for each } i\in[M']
        \end{array}
    \right\},
\end{align}
and the set of all \emph{modulo c1-types} is $\mathcal{T}_{\mathrm{mod}}=\left\{(\tau_i,\vecc):i\in[p],\ \vecc\in\mathcal{C}_{\mathrm{mod}}\right\}$.
The first $M$ dimensions of the c-types correspond to the counting quantified formulas $\forall x \exists^{=k_i}y:R_i(x,y)$ in $\Psi$, while the remaining $M'$ dimensions correspond to the modulo counting formulas $\forall x\exists^{=r'_i, k'_i} y: R'_i(x,y)$ in \cref{eq:sctwo_mod}.

We define modulo c1-type configurations analogously to the c1-type configurations in \Cref{def:c1typeconfig}.
For an interpretation $\mathcal{I}$, its modulo \emph{c1-type configuration} is a non-negative integer tensor $\vecK_{\textrm{mod}}\in\mathbb{N}^{D_{\mathrm{mod}}}$ given by $\vecK_{\textrm{mod}}[(i,\mathbf{c})] = |\mathcal{I}\langle \tau_i,\mathbf{c}\rangle|$ for each $i\in[p]$ and $\mathbf{c}\in\mathcal{C}_{\mathrm{mod}}$, where $D_{\mathrm{mod}} = p\cdot \prod_{i=1}^{M}(k_i+1) \cdot \prod_{i=1}^{M'}k_i'$.

Now, we extend \Cref{alg:ouralgo} to compute $\symwfomc(\sentence, n, \weight, \negweight)$ for a sentence $\sentence$ in the normal form as in \cref{eq:sctwo_mod}.
Define $T_h(\vecK_{\textrm{mod}})$ as the total weight of all interpretations $\mathcal{I}$ over the domain $[h]$ such that $\mathcal{I}\langle \tau_i, \mathbf{c}\rangle = \vecK_{\textrm{mod}}[(i,\mathbf{c})]$ for each $i\in[p]$ and $\mathbf{c}\in\mathcal{C}_{\mathrm{mod}}$.
Then we can compute
\begin{equation}
    \symwfomc(\sentence, n, \weight, \negweight) = \sum_{\substack{\vecK_{\textrm{mod}}\in\nat^{\dimensions_{\textrm{mod}}}: |\vecK_{\textrm{mod}}| = n, \vecK_{\textrm{mod}}\in\mathcal{K}_{\textrm{mod}}}} T_n(\vecK_{\textrm{mod}}),
\end{equation}
where $\mathcal{K}_{\textrm{mod}}$ is the set of all c1-type configurations that satisfy the counting conditions for the modulo counting quantifiers that is defined similar to \cref{eq:countingcondition1} and \cref{eq:countingcondition2} as described below.
For $\exists^{=v'_j, u'_j} x: U'_j(x)$, the condition in~\cref{eq:countingcondition1} is extended to
\begin{equation}\label{eq:countingcondition1_mod}
    \sum_{i\in[p]} \indicator{U'_j(x)\in\tau_i}\cdot \vecK_{\textrm{mod}}[(i, \cdot)] = v'_j \pmod{u'_j}.
\end{equation}
Similarly, for the modulo counting quantified formulas $\forall x\exists^{=r'_i, k'_i} y: R'_i(x,y)$, the condition in \cref{eq:countingcondition2} is extended to
\begin{equation}\label{eq:countingcondition2_mod}
    \sum_{i\in[p]} \vecK_{\textrm{mod}}[i,k_1,\ldots,k_M,r'_1,\ldots,r'_{M'}] = n.
\end{equation}

It remains to adapt the c-type update used in \cref{sec:c_type_update} such that it can handle the modulo counting quantifiers.
Given $\vecc \in \mathcal{C}_{\mathrm{mod}}$ and $\vecsig \in\{0,1\}^{\Mmod}$, we defined $\vecc+_{\bmod}\vecsig$ as
\begin{equation}\label{eq:modulo_update}
\begin{aligned}
(\vecc+_{\mathrm{mod}}\boldsymbol{\sigma})[i] &=
\vecc[i]+\boldsymbol{\sigma}[i], && \forall i\in[M],\\
(\vecc+_{\mathrm{mod}}\boldsymbol{\sigma})[M+i] &= 
\big(\vecc[M+i]+\boldsymbol{\sigma}[M+i]\big) \pmod{k'_i}, && \forall i\in[M'].
\end{aligned}
\end{equation}
Then, the definition of $\update$ in \cref{eq:updateconf} is modified to
\begin{equation}\label{eq:updateconf_mod}
    \update(\vecc_b, \vecU)[\vecc] = \sum_{\vecsig\in \{0, 1\}^{M_{\bmod}}: \vecc_b +_{\textrm{mod}} \vecsig = \vecc} \vecU[\vecsig], \quad \forall \vecc\in\mathcal{C}_{\bmod}.
\end{equation}
The rest of the algorithm remains unchanged. Clearly, the modifications above do not change the overall complexity of \ouralgo{}, which thus proves the domain-liftability of \ctwomod{}.

\begin{theorem}\label{thm:modulo_complexity}
    The fragment \ctwomod{} is domain-liftable using \ouralgo{}.
    Specifically, for a \ctwomod{} sentence $\sentence$ in the normal form of \cref{eq:sctwo_mod}, the modified \ouralgo{} computes $\symwfomc(\sentence, n, \weight, \negweight)$ in time $O(n^{2 \cdot D_{\textrm{mod}} + 2^{\Mmod} -2})$, where $D_{\textrm{mod}} = p \cdot \prod_{i=1}^{M}(k_i+1) \cdot \prod_{j=1}^{M'} k'_j$.
\end{theorem}
\begin{proof}
    The proof directly follows from the discussion above (with the normal form transformation presented in \Cref{app:modk_reduction}) and the complexity analysis in \Cref{sec:complexity_analysis}.
\end{proof}
\begin{remark}
    The extension of \ouralgo{} also applies to the normal form of \ctwomod{} sentences that include $\exists^{\leq r, k}$ quantifiers without changing the complexity.
    The required modifications are similar to those for the ordinary $\exists^{\leq r}$ quantifiers as mentioned in \Cref{remark:lessthan1,remark:lessthan2,remark:lessthan3}.
\end{remark}

\section{Experiments}\label{sec:experiments}

We empirically evaluate \ouralgo{} against lifted 
and propositional baselines on a diverse set of benchmark families spanning both \ctwo{} and \ctwomod{} sentences.

\subsection{Experimental Setup}\label{sec:setup}

\paragraph{Baselines.}
We compare \ouralgo{} against different state-of-the-art baselines depending on the fragment being evaluated:
\begin{itemize}
    \item \textbf{\ctwo{} Baselines:} 
    \begin{itemize}
        \item \Fast{}: The current state-of-the-art reduction-based approach implemented using the encoding framework of~\cite{KR2024-64} on top of the lifted counting procedure of~\cite{kuzelka2021weighted}. We re-implemented \Fast{} by ourselves based on the description in the original paper, as the public implementation by~\cite{KR2024-64} is written in Julia, and thus not directly comparable with our Python implementation of \ouralgo{}.
        \item \Recursive{}: A domain-recursion-based approach from~\cite{Meng2024}. We include this baseline because it is the best-known domain-recursion-based method for \ctwo{} with a linear-order axiom. We use their public implementation.\footnote{\url{https://github.com/yuanhong-wang/WFOMC}}
    \end{itemize}

    \item \textbf{\ctwomod{} Baselines:} Since neither \Fast{} nor \Recursive{} supports modulo counting, we compare against two propositional model counters applied after grounding:
    \begin{itemize}
        \item \Ganak{}~\citep{sharma2019ganak}: A state-of-the-art exact propositional model counter.\footnote{\url{https://github.com/meelgroup/ganak}}
        \item \ApproxMC{}~\citep{ChakrabortyMV13, PoteM025}: A state-of-the-art approximate propositional model counter.\footnote{\url{https://github.com/meelgroup/approxmc}}
    \end{itemize}
\end{itemize}

\paragraph{Benchmarks.}

We select benchmarks that vary in the size of the c1-type space, since this is a key factor influencing the complexity of the algorithms.
The characteristics of the selected benchmarks are summarized in \Cref{tab:benchmarks}.

\begin{table}[htbp]
    \centering
    \caption{Summary of Benchmark Characteristics.}
    \label{tab:benchmarks}
    \begin{tabular}{@{}llcl@{}}
    \toprule
    Benchmark & Fragment & No. of c1-types & Description \\ \midrule
    $k$-regular graphs & \ctwo{} & $k+1$ & Undirected graphs where each vertex has degree $k$ \\
    $k$-regular $l$-colored & \ctwo{} & $l\cdot (k+1)$ & $k$-regular graphs with an $l$-vertex-coloring \\
    $k$-regular digraphs & \ctwo{} & $(k+1)^2$ & Directed graphs where each vertex has out-degree and in-degree $k$ \\
    BA graphs & \ctwo{} & $8$ & Barab\'asi--Albert models\\
    $r$-mod-$k$ regular & \ctwomod{} & $k$ & Undirected graphs where each vertex degree is $r \pmod{k}$ \\
    $m$-odd-degree & \ctwomod{} & $4$ & Undirected graphs with exactly $m$-odd-degree vertices \\ \bottomrule
    \end{tabular}
\end{table}

\paragraph{Evaluation Metrics.}
We report the wall-clock runtime as the primary performance metric. We also evaluate correctness (see \Cref{app:correctness}) and memory usage (see \Cref{app:memory_results}).

\paragraph{Evaluation Protocol.}

For each benchmark, we measure the scaling behavior of the methods by varying the domain size $n$.
Unless stated otherwise, each benchmark instance is executed once.
For $\Fast{}$ and $\Recursive{}$, there is no additional parameter tuning, and we use the default solver invocation in our evaluation pipeline.
For $\Ganak{}$, we also use the default solver invocation, again without additional parameter tuning.
For $\ApproxMC{}$, we use the default configuration with $\epsilon = 0.05$ and $\delta = 0.1$, using a fixed random seed of $42$.
A timeout of $10^4$ seconds is imposed for each run.
Timeouts and failed runs are excluded from the plotted runtime curves.
No explicit memory limit is imposed during evaluation.
Since the reported runtimes are based on single runs, small differences between methods should be interpreted with caution, whereas substantial and consistent gaps across benchmark families remain indicative of the overall scalability trend.

\paragraph{Implementation and Hardware.}
Our implementation of \ouralgo{} is written in Python~3.9 and 
incorporates the algorithm-specific optimizations described in \Cref{sec:optimization}.
We do not include a separate ablation study; the reported results, therefore, reflect the end-to-end performance of the complete implementation.
All experiments were conducted on an Ubuntu server equipped with $4$ Intel Xeon Gold 5218 CPUs and 512~GB of RAM.

\subsection{\ctwo{} Benchmarks}\label{sec:c2_benchmarks}
We evaluate \ouralgo{} on four benchmark families of graph-counting tasks expressible in $\ctwo{}$ (possibly with a linear order axiom), following the prior work~\cite{KR2024-64}.

\subsubsection{$K$-Regular Graphs and Variants}

We first consider the problem of counting $k$-regular graphs and their variants.
\begin{itemize}
    \item \textbf{$k$-regular graphs:}
    This is the basic benchmark in $\ctwo{}$, encoding the problem of counting undirected graphs where each vertex has degree $k$.
\begin{align*}
    \sentence_{\text{$k$-regular}} = \left(\forall x: \neg E(x,x) \right) \land  \left(\forall x \forall y: E(x,y) \to E(y,x) \right)\land \left(\forall x \exists^{=k} y: E(x,y) \right).
\end{align*}

\item \textbf{$k$-regular $l$-colored graphs:}
This benchmark extends the $k$-regular graphs task with assignments of $l$ colors to vertices, such that adjacent vertices must have different colors. 
This extension increases the number of c1-types by a factor of $l$.
The formula below illustrates the special case of $k$-regular $2$-colored graphs:
\begin{align*}
    \sentence_{\text{$k$-regular $2$-colored}} = &\sentence_{\text{$k$-regular}} \land \left(\forall x: R(x) \lor B(x)\right) \land \left(\forall x: \neg R(x) \lor \neg B(x) \right) \land \\
    &\forall x \forall y: E(x,y) \to \left(\neg (R(x) \land R(y)) \land \neg (B(x) \land B(y))\right).
\end{align*}

\item \textbf{$k$-regular digraphs:}
This benchmark removes the undirectedness constraint from the $k$-regular graphs, and instead requires that each vertex has out-degree $k$ and in-degree $k$. 
Note the number of c1-types increases to $(k+1)^2$ since we need to track both out-degree and in-degree counts for each vertex.
\begin{align*}
    \sentence_{\text{digraph}} = \left(\forall x: \neg E(x,x) \right)\land \left(\forall x \exists^{=k} y: E(x,y) \right) \land \left(\forall x \exists^{=k} y: E(y,x) \right).
\end{align*}
\end{itemize}

The runtime results for these benchmarks are shown in \Cref{fig:time_kregular,fig:k_regular_l_colored,fig:k_regular_digraphs}, and are discussed in detail below.
Correctness results and peak-memory measurements are reported separately in the \Cref{app:correctness} and \Cref{app:memory_results}.

\paragraph{$k$-regular graphs. } 
The results are shown in \Cref{fig:time_kregular}. 
Across all tested settings ($k=3,4,5$), \ouralgo{} is consistently faster than both \Fast{} and \Recursive{}, 
and the gap grows as either $n$ or $k$ increases. 
In the $5$-regular setting (\Cref{fig:time_5_regular}), 
\Recursive{} does not scale beyond domain size $35$ and \Fast{} beyond $44$, 
whereas \ouralgo{} continues to substantially larger domains.

\begin{figure}[htbp]
    \centering
    \begin{subfigure}[b]{0.32\textwidth}
        \includegraphics[width=\linewidth]{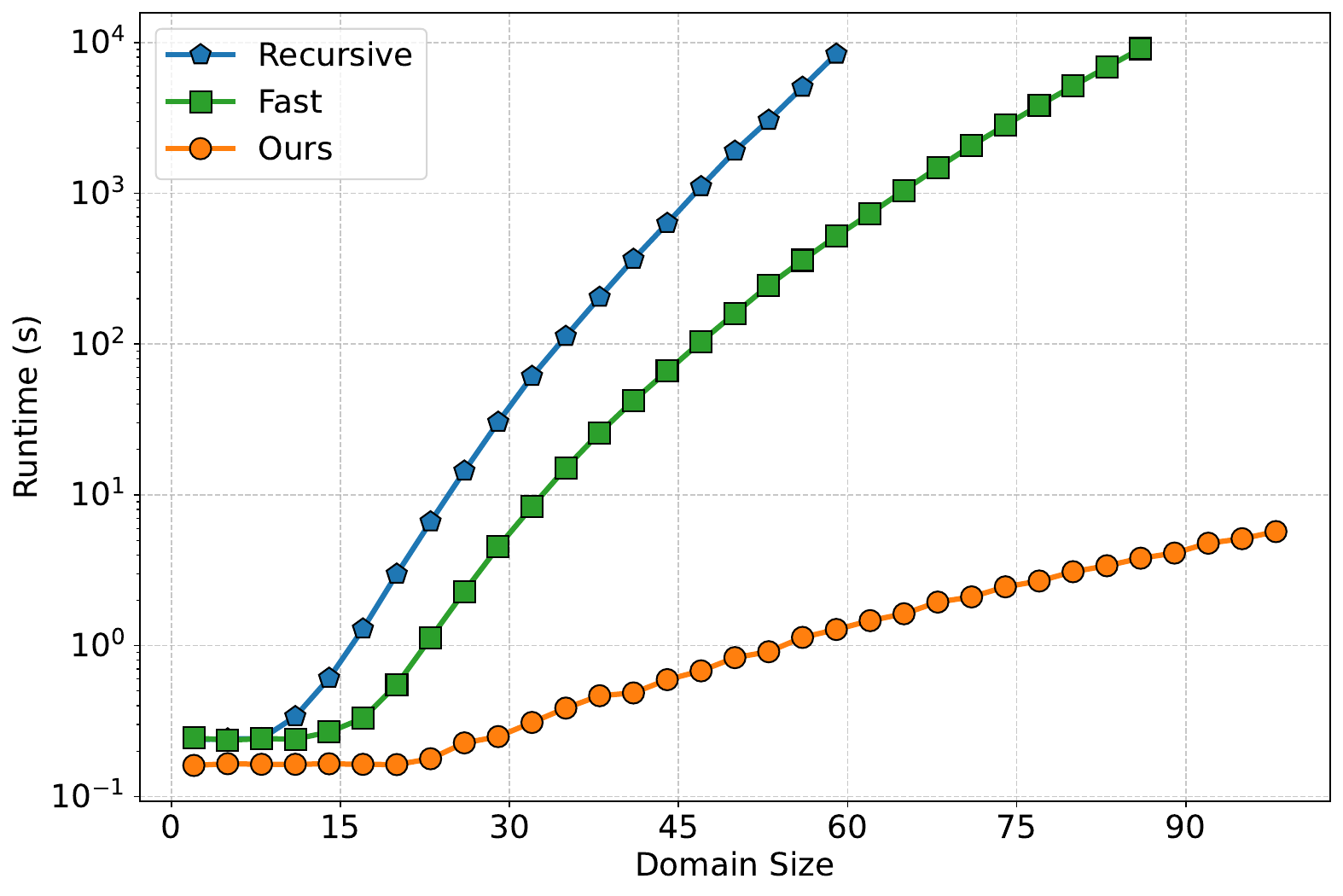}
        \caption{$3$-regular graph}
    \end{subfigure}
    \begin{subfigure}[b]{0.32\textwidth}
        \includegraphics[width=\linewidth]{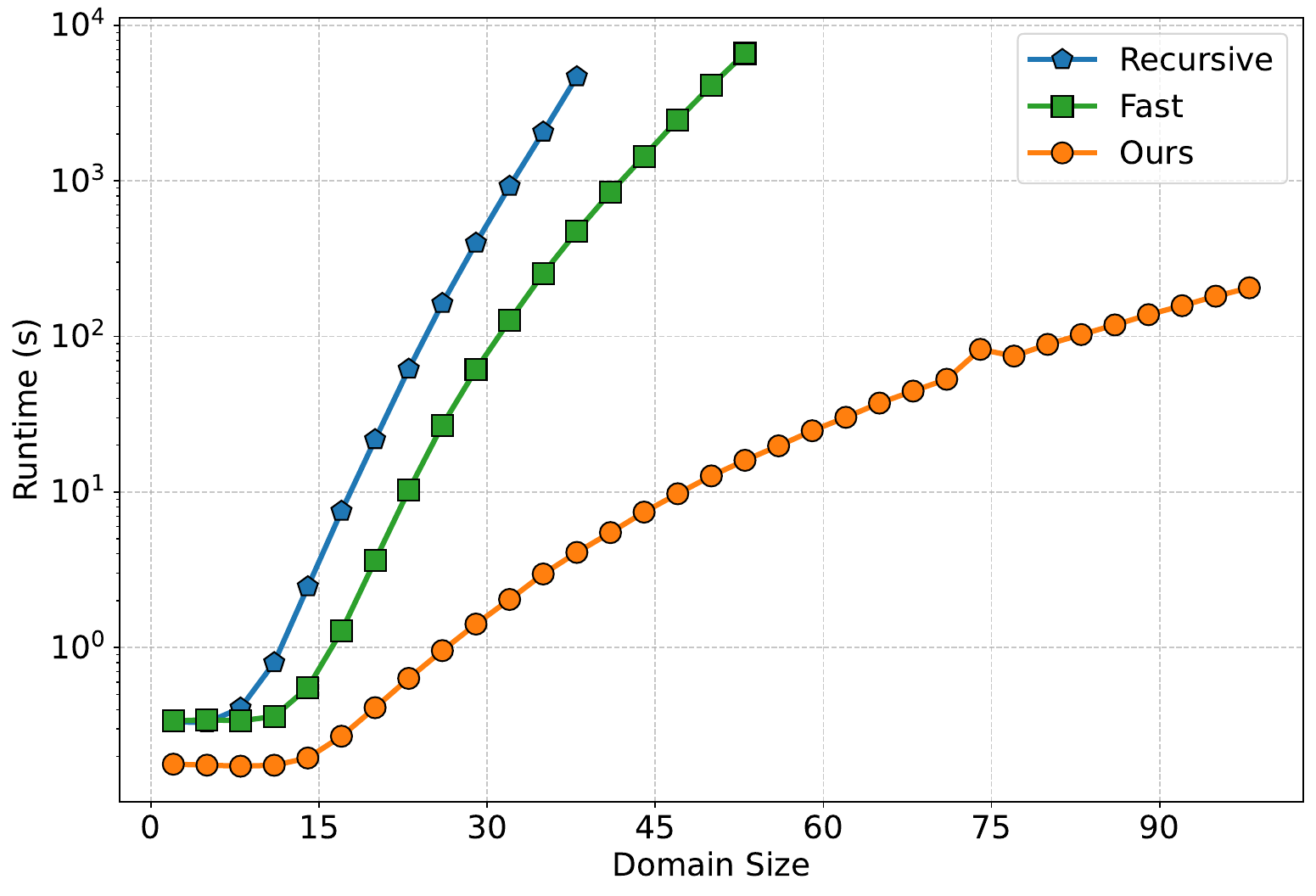}
        \caption{$4$-regular graph}
    \end{subfigure}
    \begin{subfigure}[b]{0.32\textwidth}
        \includegraphics[width=\linewidth]{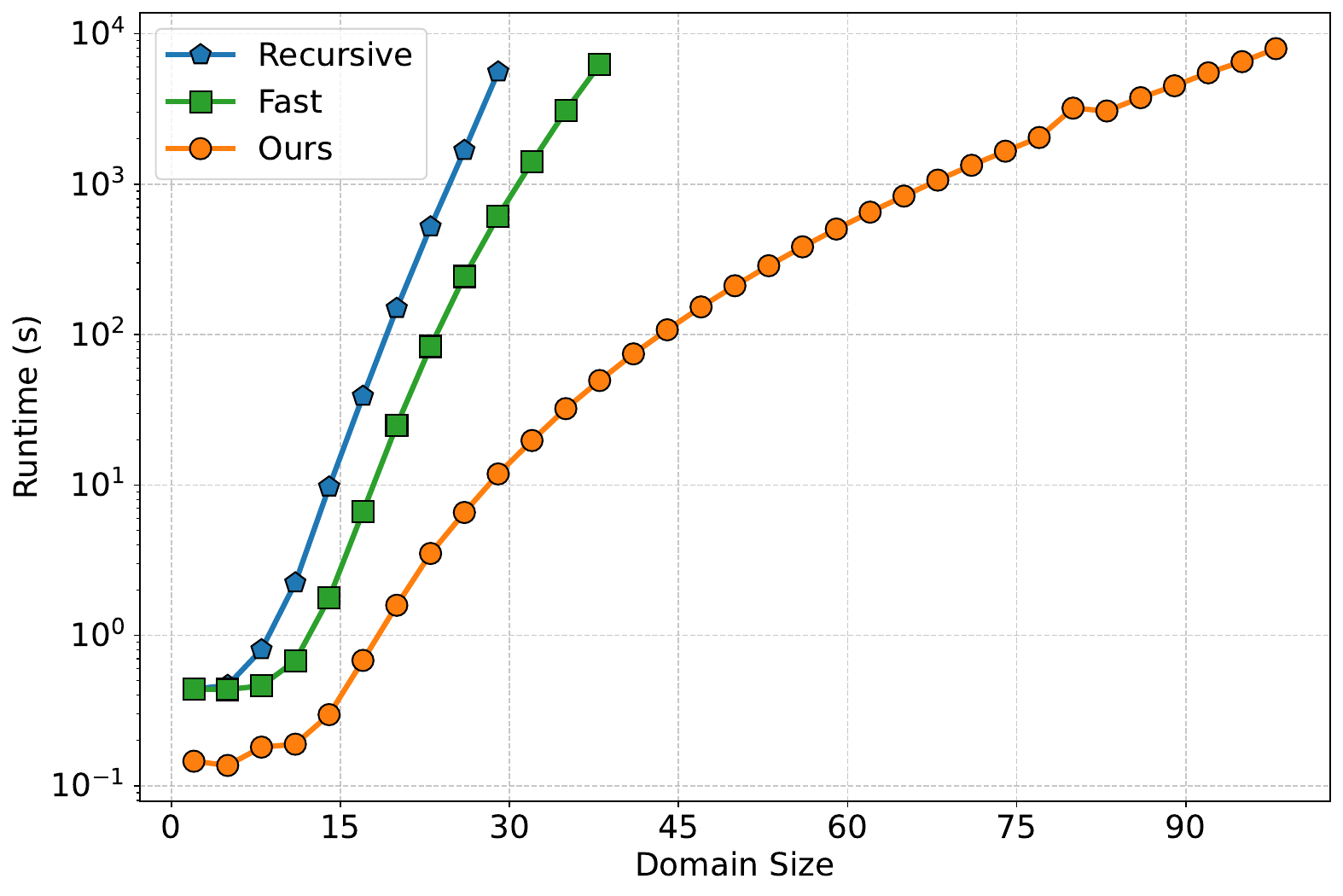}
        \caption{$5$-regular graph}
        \label{fig:time_5_regular}
    \end{subfigure}
    \caption{Runtime (log scale) comparison for counting $k$-regular graphs.}\label{fig:time_kregular}
\end{figure}

\paragraph{$k$-regular $l$-colored graphs.} 
The results in \Cref{fig:k_regular_l_colored} show a similar overall pattern to the $k$-regular case. 
\ouralgo{} remains the fastest method across the tested settings. 
The advantage is especially clear in the $2$-colored cases, while increasing the number of colors reduces the maximum domain sizes reachable by all methods. 
Overall, the results indicate that \ouralgo{} remains effective even as the induced 1-type space becomes larger.

\begin{figure}[htbp]
    \centering
    \begin{subfigure}[b]{0.32\textwidth}
        \includegraphics[width=\linewidth]{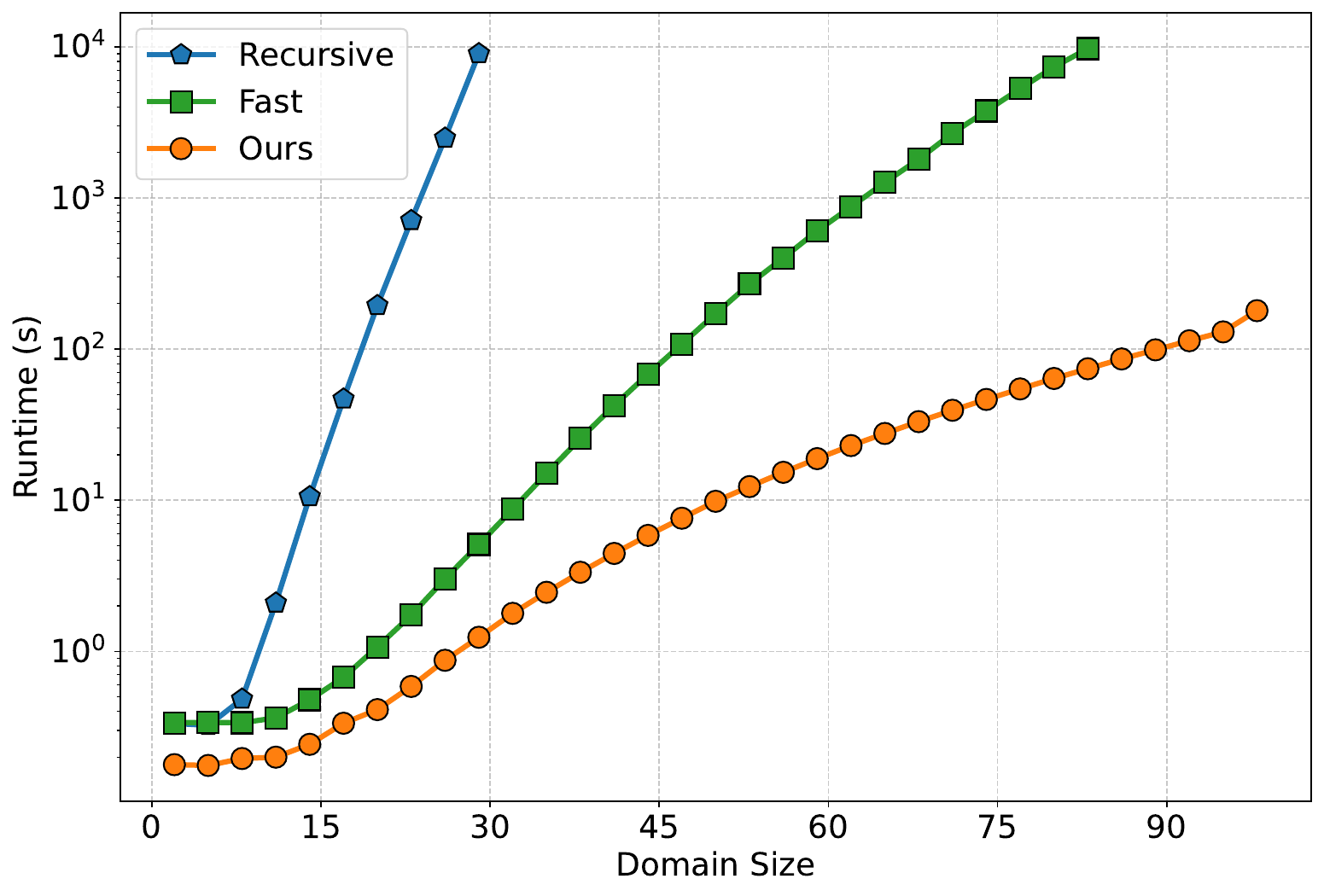}
        \caption{$3$-regular $2$-colored graph}
    \end{subfigure}
    \begin{subfigure}[b]{0.32\textwidth}
        \includegraphics[width=\linewidth]{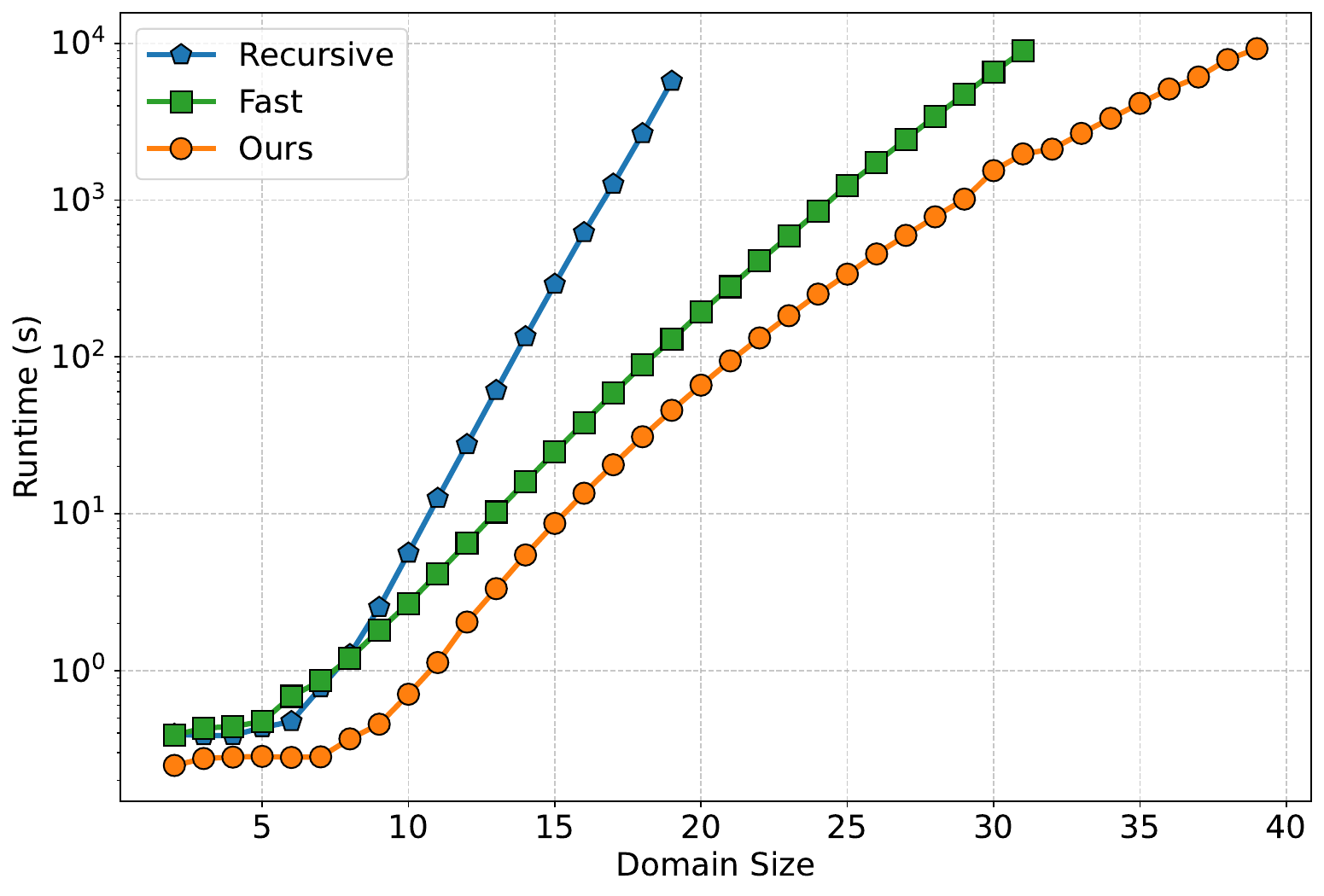}
        \caption{$3$-regular $3$-colored graph}
    \end{subfigure}
    \begin{subfigure}[b]{0.32\textwidth}
        \includegraphics[width=\linewidth]{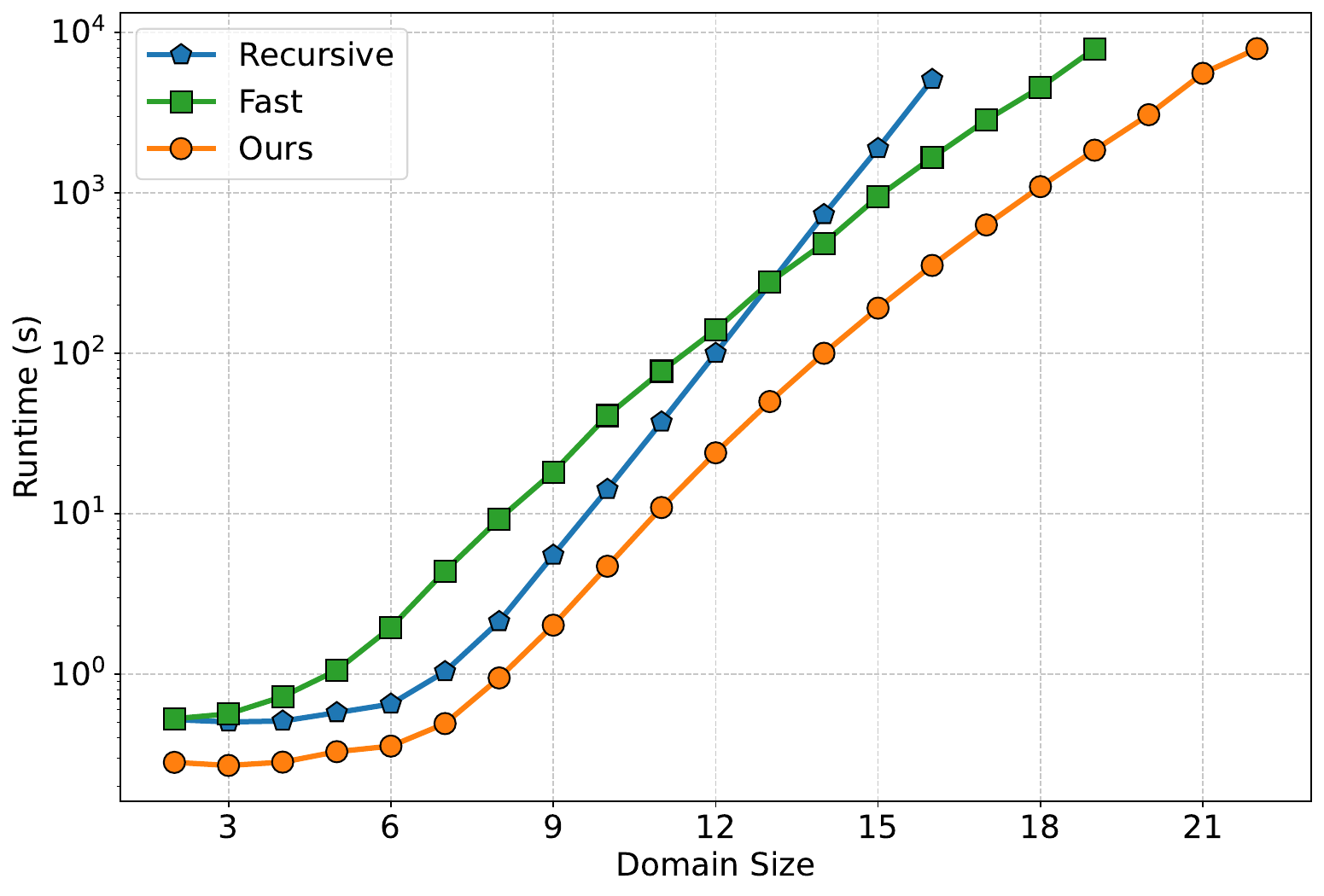}
        \caption{$3$-regular $4$-colored graph}
    \end{subfigure}
    \\[1ex] 
    \begin{subfigure}[b]{0.32\textwidth}
        \includegraphics[width=\linewidth]{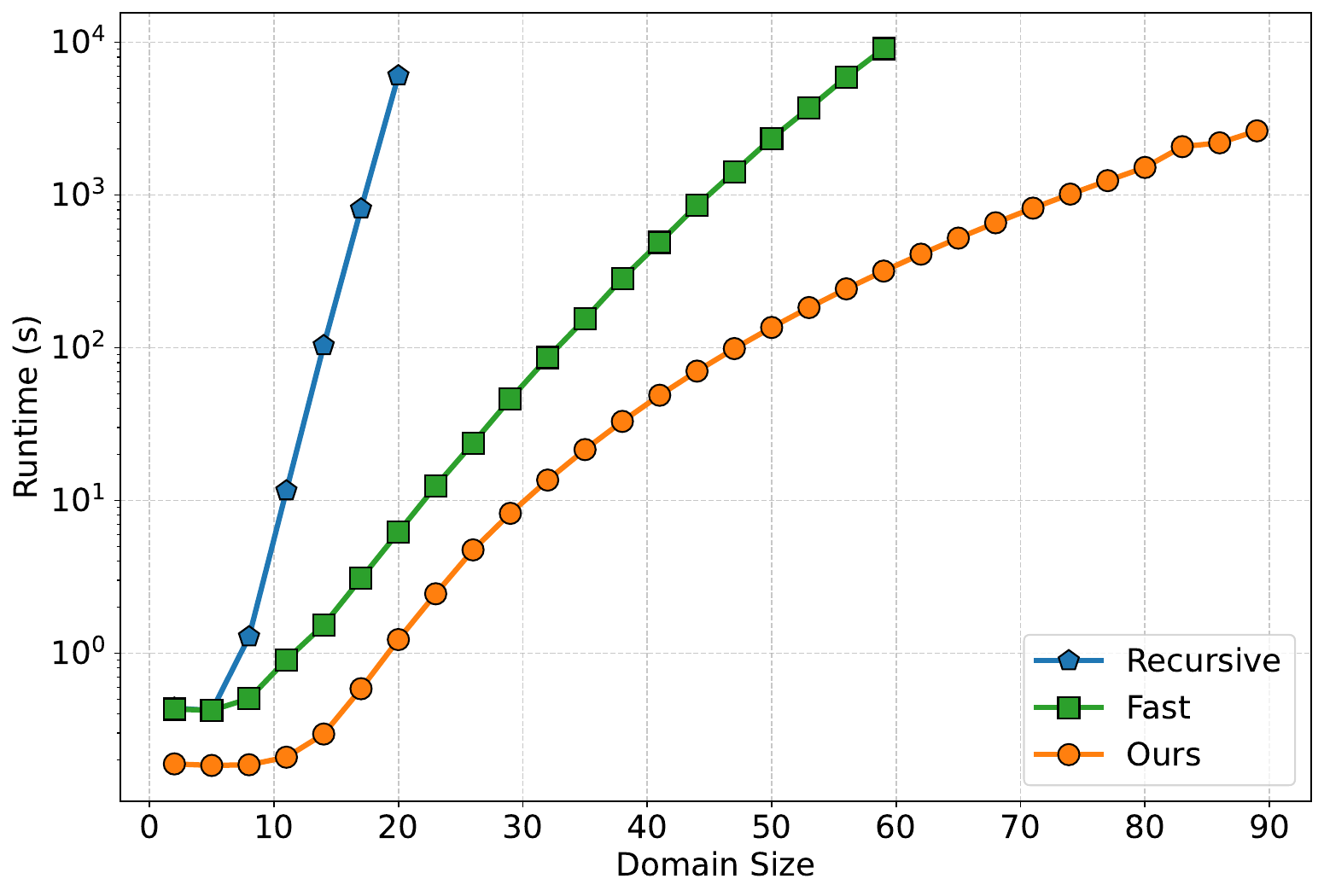}
        \caption{$4$-regular $2$-colored graph}
    \end{subfigure}
    \begin{subfigure}[b]{0.32\textwidth}
        \includegraphics[width=\linewidth]{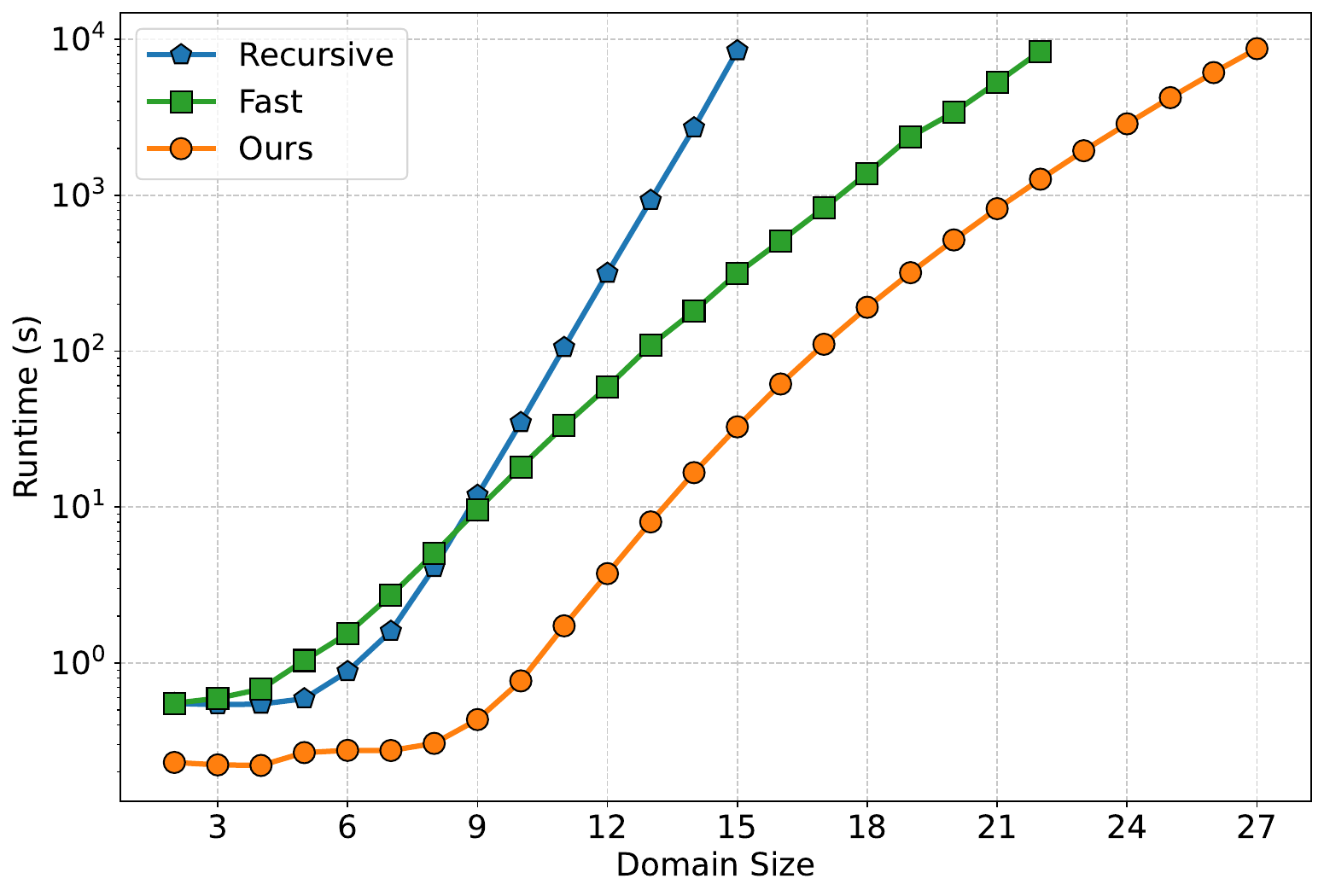}
        \caption{$4$-regular $3$-colored graph}
    \end{subfigure}
    \begin{subfigure}[b]{0.32\textwidth}
        \includegraphics[width=\linewidth]{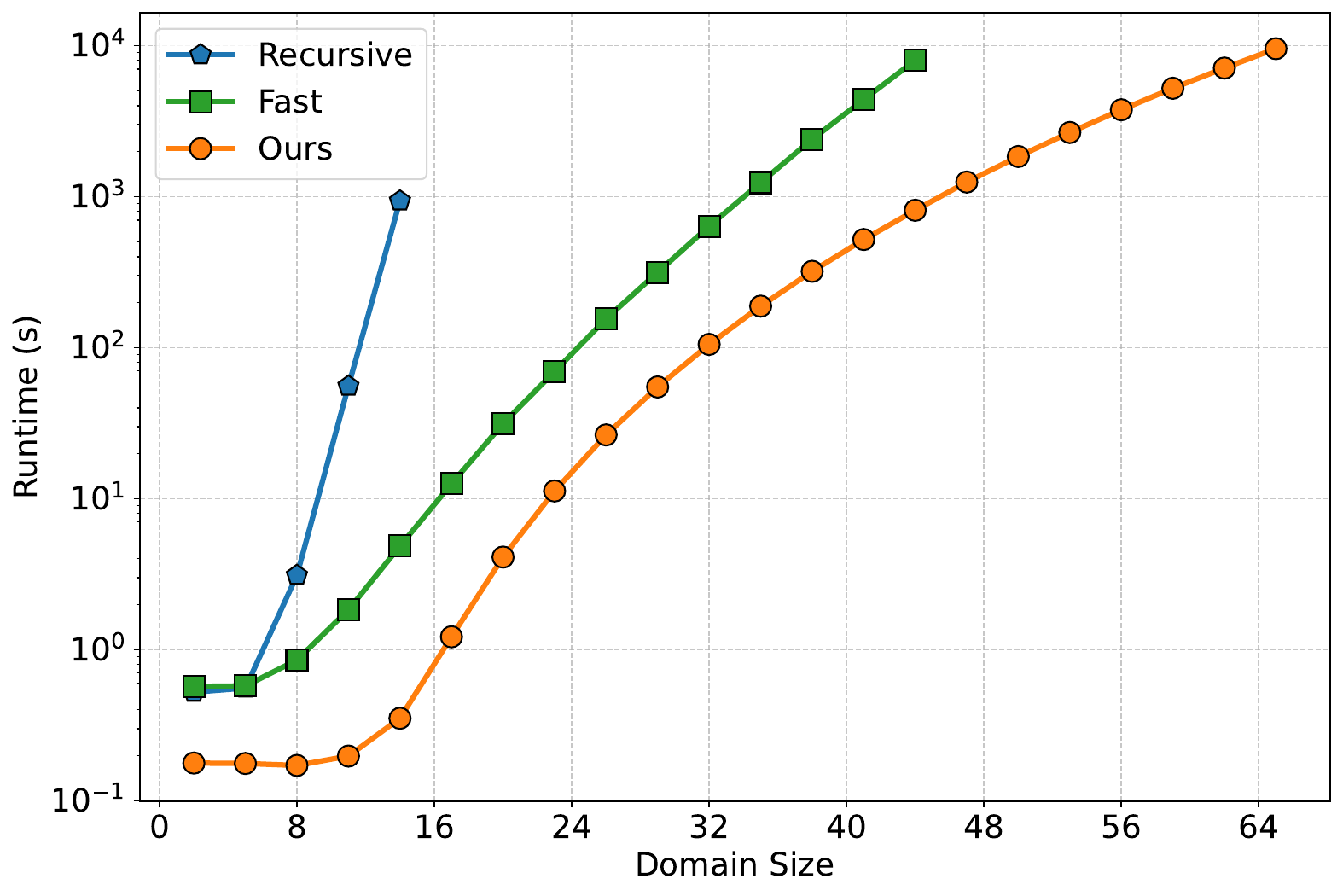}
        \caption{$5$-regular $2$-colored graph}\label{fig:time_5_regular_2_colored}
    \end{subfigure}
    \caption{Runtime (log scale) comparison for counting $k$-regular $l$-colored graphs,
    across varying degree $k \in \{3, 4, 5\}$ and number of colors $l \in \{2, 3, 4\}$.}
    \label{fig:k_regular_l_colored}
\end{figure}

\paragraph{$k$-regular digraphs.} 
\Cref{fig:k_regular_digraphs} shows that \ouralgo{} remains consistently faster than both \Fast{} and \Recursive{} on directed graphs. 
The gap is already visible in the $2$-regular case and becomes clearer as the domain size increases. 
In the $3$-regular setting, the three methods are closer on small domains, but the baselines deteriorate more quickly as $n$ grows.

\begin{figure}[htbp]
    \centering
    \begin{subfigure}[b]{0.32\textwidth}
        \includegraphics[width=\linewidth]{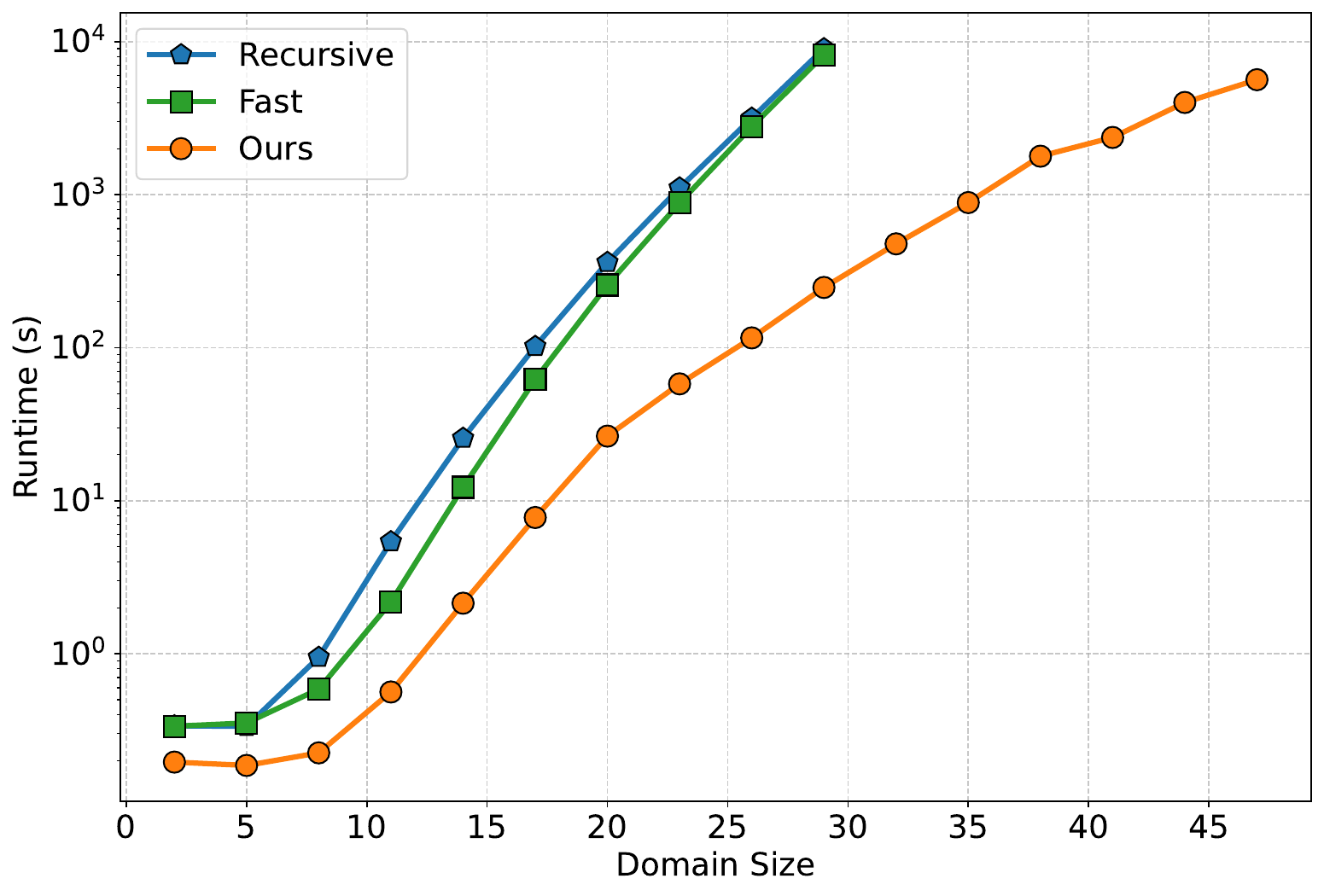}
        \caption{$2$-regular digraph}\label{fig:time_2_regular_digraph}
    \end{subfigure}
    \begin{subfigure}[b]{0.32\textwidth}
        \includegraphics[width=\linewidth]{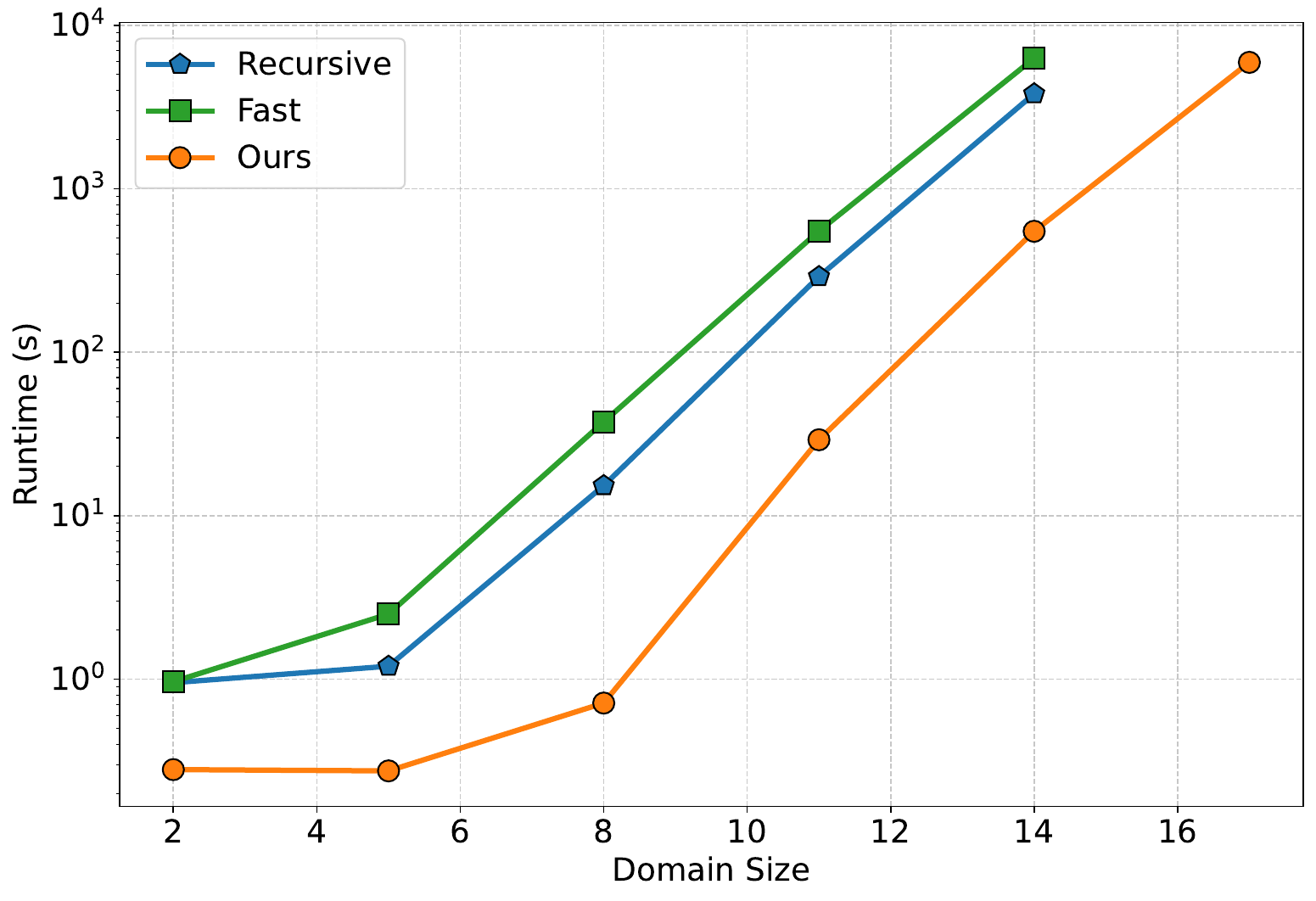}
        \caption{$3$-regular digraph}\label{fig:time_3_regular_digraph}
    \end{subfigure}
    \caption{
        Runtime comparison for counting (a) $2$-regular and (b) $3$-regular digraphs.
    }\label{fig:k_regular_digraphs}
\end{figure}

\subsubsection{BA graphs.}

Barabási--Albert (BA) graphs are widely used to study the emergence of complex networks, 
including social networks and the internet~\cite{albert2002statistical}. 
Following prior lifted inference work, we consider a simplified version of the Barabási--Albert (BA) model, whose encoding in $\ctwo{}$ with a linear order axiom is:
\begin{equation}\label{eq:ba_encoding}
    \begin{aligned}
    \sentence_{\text{BA}} = {} & \left(\forall x: Eq(x, x)\land \neg R(x,x)\right) \land (|Eq| = n) \land \left(\exists^{=k+1}x: K(x)\right) \land \\
    & \forall x \forall y: K(x) \land K(y) \land \neg Eq(x, y) \Rightarrow R(x, y) \land \\
    & \forall x \exists^{=k}y: R(x, y) \land \\
    & \forall x \forall y: R(x, y) \land \neg (K(x) \land K(y)) \Rightarrow y \leq x \land \\
    & \forall x \forall y: K(x) \land \neg K(y) \Rightarrow x \leq y,
    \end{aligned}
\end{equation}
where $\le$ is a special binary predicate that is interpreted as a linear order over the domain.

Conceptually, the $n$-vertex graph encoded by $\sentence_{\text{BA}}$ is generated through a sequential construction process. 
First, the vertices are ordered using the linear order axiom, and a  $(k+1)$-vertex complete graph is formed on the initial $k+1$ vertices. The structure is then expanded by iteratively appending the remaining vertices $i \in \{k+2, k+3, \dots, n\}$. 
As each new vertex $i$ is introduced, it forms $k$ outgoing edges that connect strictly to preceding vertices; that is, all new edges are of the form $(i, j)$ where $j \in \{1, 2, \dots, i-1\}$. 
When enumerating these graphs, the final count is divided by $n!$ to account for the symmetries introduced by vertex labeling, ensuring that isomorphic graphs are not overcounted.
In our experiments, we focus on the case of $k=3$ for BA graphs aligned with \citep{KR2024-64}, which already yields a non-trivial c1-type space of size $8$.

In the above encoding \cref{eq:ba_encoding}, the equality predicate $Eq$ is enforced through $\forall x: Eq(x,x)$ and the cardinality constraint $|Eq| = n$.
We also consider a variant $\sentence_{\mathrm{BA-NOCC}}$, in which $|Eq| = n$ is replaced by $\forall x \forall y : Eq(x,y) \leftrightarrow \left((x \le y) \wedge (y \le x)\right)$.
This variant allows us to exclude the effect of cardinality constraints, since the reduction to eliminate cardinality constraints based on Lagrange interpolation~\cite{kuzelka2021weighted} usually leads to a degradation in performance, which might obscure the scalability advantage of \ouralgo{} over the lifted baselines.

For the $\sentence_{\text{BA}}$ and $\sentence_{\mathrm{BA-NOCC}}$ benchmarks, the theory includes a linear order axiom. 
Since \Fast{} does not support this setting, we only compare \ouralgo{} against \Recursive{}.
The results are shown in \Cref{fig:compare_ba}. 
For the BA benchmark encoded in $\sentence_{\mathrm{BA-NOCC}}$, the baseline method becomes infeasible at relatively small domain sizes, whereas \ouralgo{} remains feasible over the tested range. 
For the $\sentence_{\mathrm{BA}}$ variant, all methods incur higher runtimes, but \ouralgo{} still scales to larger domains than the baseline. 
These results suggest that the proposed method remains effective even when the theory combines counting quantifiers, 
linear order, and a global cardinality constraint.

\begin{figure}[htbp]
    \centering
    \begin{subfigure}[b]{0.32\textwidth}
        \includegraphics[width=\linewidth]{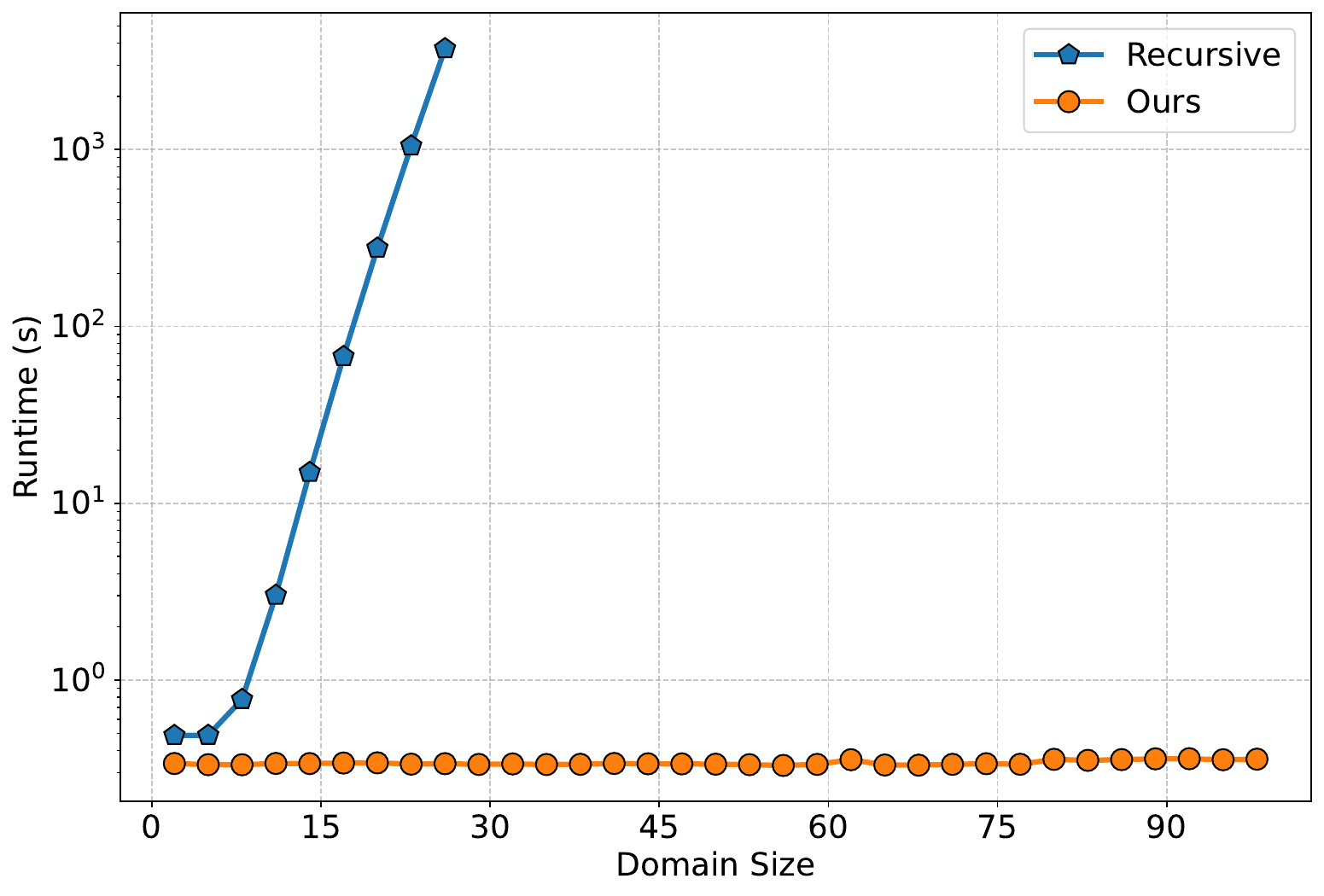}
        \caption{$\sentence_{\mathrm{BA-NOCC}}$}\label{fig:ba}
    \end{subfigure}
    \begin{subfigure}[b]{0.32\textwidth}
        \includegraphics[width=\linewidth]{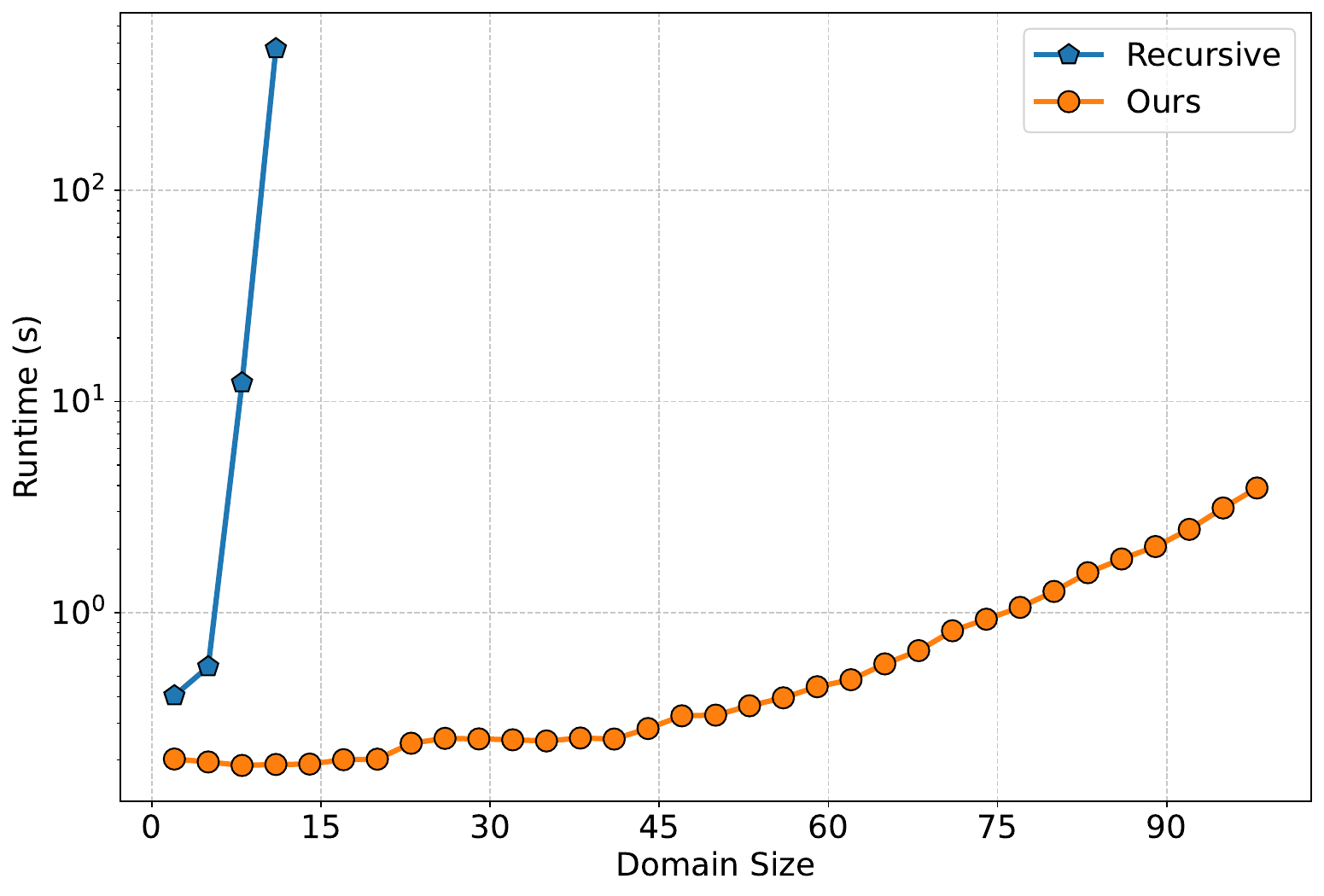}
        \caption{$\sentence_{\mathrm{BA}}$}\label{fig:ba_cc}
    \end{subfigure}
    \caption{Runtime (log scale) comparison for counting Barabási-Albert (BA) graphs under two encodings with and without the cardinality constraint on the equality predicate.
    Since \Fast{} is not applicable to benchmarks with linear order, the comparison here is against \Recursive{}.
    }\label{fig:compare_ba}
\end{figure}

\subsection{\ctwomod{} Benchmarks}\label{sec:ctwomod_performance}

We evaluate the modulo counting extension on two benchmark families:

\begin{itemize}
\item \textbf{$r$-mod-$k$-regular graphs:}
The task of counting $r$-mod-$k$-regular undirected graphs can be expressed by the following \ctwomod{} sentence:
\begin{align*}
    \sentence_{\text{$r$-mod-$k$-regular}} = &\left(\forall x: \neg E(x,x)\right) \land \left(\forall x \forall y: E(x,y) \to E(y,x)\right) \land \\
    &\forall x \exists^{=r, k} y: E(x,y).
\end{align*}

\item \textbf{$m$-odd-degree graphs:}
The task of counting graphs with exactly $m$-odd-degree vertices can be expressed by the following sentence, 
where $k$ denotes the number of undirected edges $E$. 
Since $E$ is represented as a symmetric binary relation, each undirected edge contributes two true ground atoms, 
and thus the corresponding cardinality constraint is $|E|=2k$. For readability, we first present the direct formulation:
\begin{align*}
    \sentence_{\text{$m$-odd-degree}} = &\left(\forall x: \neg E(x,x)\right) \land \left(\forall x \forall y: E(x,y) \to E(y,x)\right) \land\\
    & \forall x: Odd(x) \leftrightarrow \exists^{=1,2} y: E(x,y)\land \\
    &\exists^{=m} x: Odd(x) \land \\
    &|E| = 2k.
\end{align*}
The normalized sentence used by \ouralgo{} is obtained via the reductions in \Cref{app:modk_reduction} and
is given in \Cref{app:reduction_m_odd_degree}.
\end{itemize}

\paragraph{$r$-mod-$k$-regular graphs.} 
The runtime results are shown in \Cref{fig:time_modkregular}. 
Across all three modulo counting tasks,
namely $0$-mod-$2$-regular, $1$-mod-$2$-regular, and $2$-mod-$4$-regular graphs,
\ouralgo{} remains feasible 
over a substantially larger range of domain sizes than both \Ganak{} and \ApproxMC{}. 
The observed gap mainly reflects the difference between handling modulo counting directly at the first-order level and solving the grounded propositional instances. 
The results therefore indicate that lifted inference with direct support for modulo counting quantifiers can provide a significant practical advantage, not only in theory.

\begin{figure}[htbp]
    \centering
    \begin{subfigure}[b]{0.32\textwidth}
        \includegraphics[width=\linewidth]{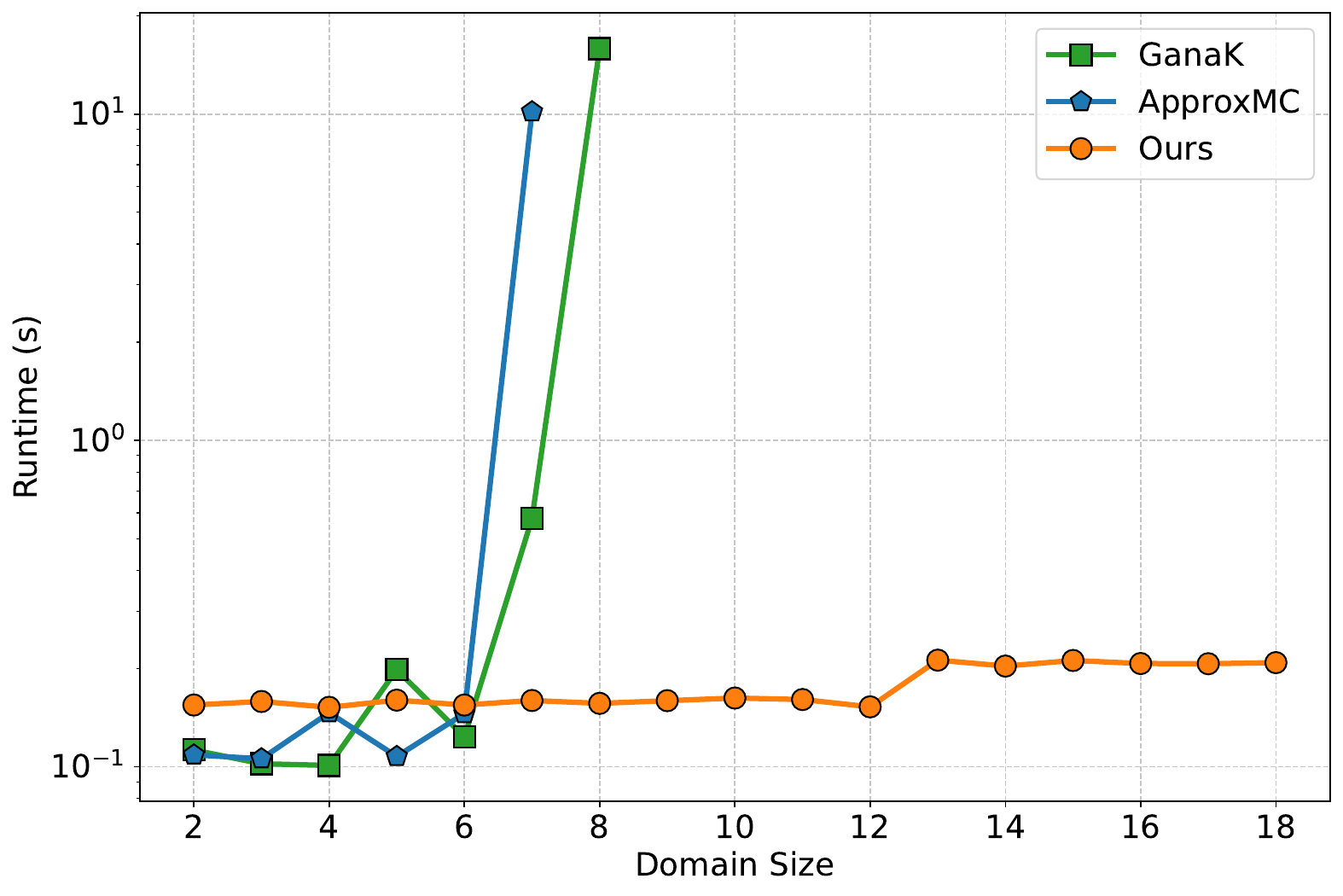}
        \caption{$0$-mod-$2$-regular graph}
    \end{subfigure}
    \begin{subfigure}[b]{0.32\textwidth}
        \includegraphics[width=\linewidth]{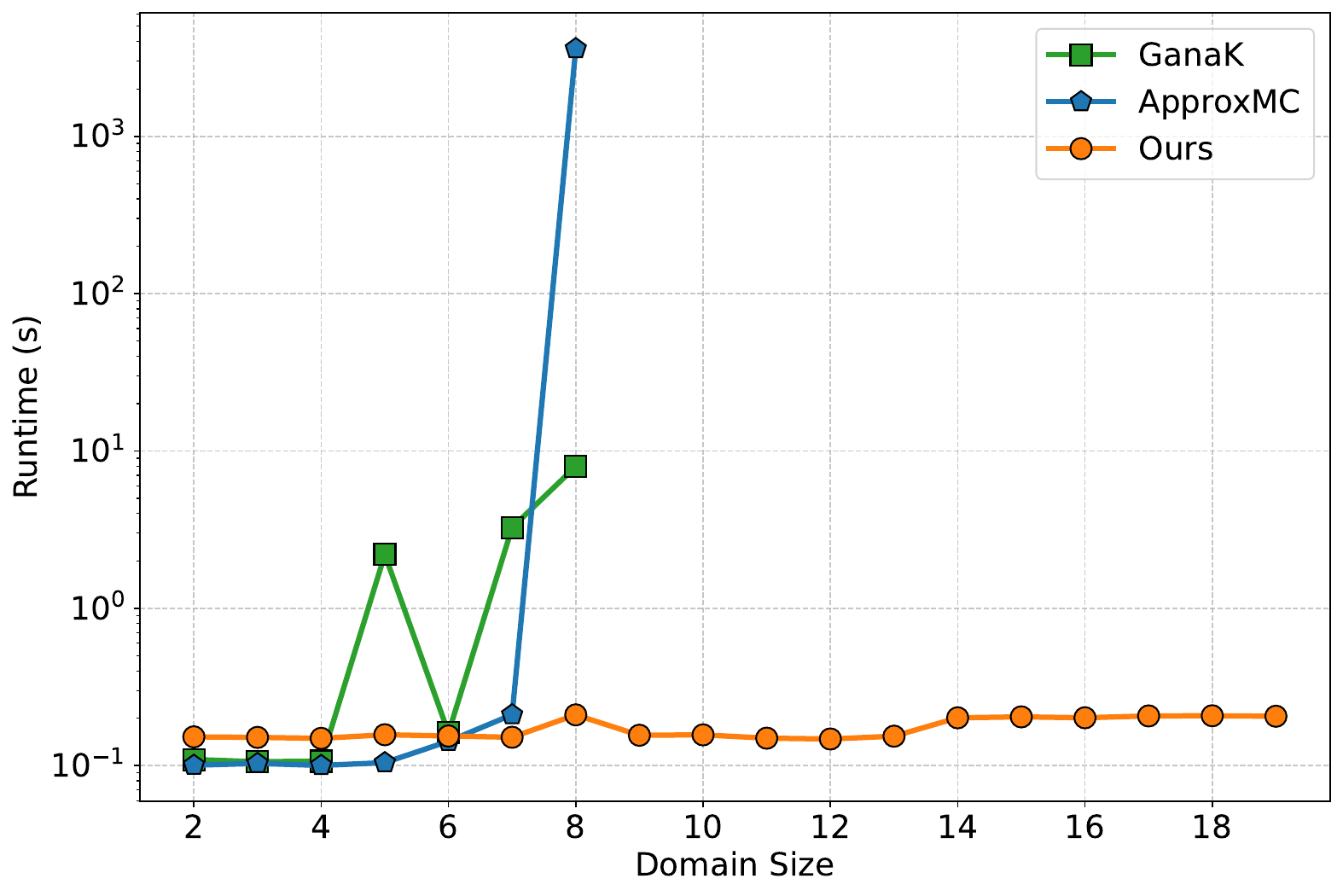}
        \caption{$1$-mod-$2$-regular graph}
    \end{subfigure}
    \begin{subfigure}[b]{0.32\textwidth}
        \includegraphics[width=\linewidth]{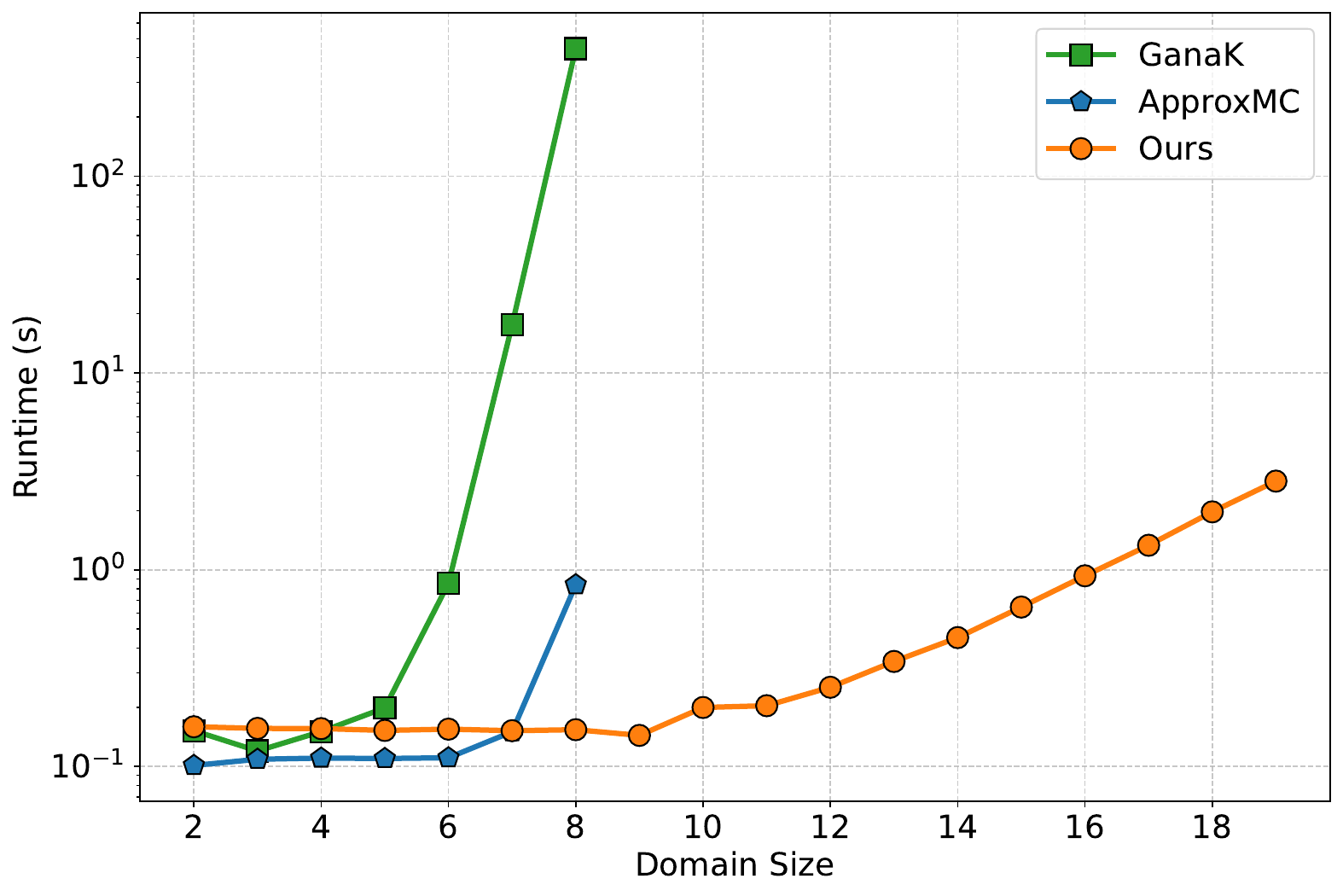}
        \caption{$2$-mod-$4$-regular graph}
    \end{subfigure}
    \caption{Runtime comparison for three modulo counting graph families.}\label{fig:time_modkregular}
\end{figure}

\paragraph{$m$-odd-degree graphs.}

The $m$-odd-degree task depends on three parameters: 
 the domain size $n$, 
 the number of odd-degree vertices $m$, 
 and the number of undirected edges $k$. 
 To keep the presentation focused, 
 we fix $k=2n$ and vary $m$, which we found to provide a representative subset of the parameter space in preliminary experiments.
 Since the number of odd-degree vertices in an undirected graph must be even, 
 we only consider even values of $m$; the cases $m=2,4,6$ already provide simple nontrivial slices of the parameter space.
The runtime comparison is shown in \Cref{fig:odd_degree_time}. 
Across all three settings, 
\ouralgo{} again scales to substantially larger domains than the propositional baselines. 
The runtimes of \Ganak{} and \ApproxMC{} 
increase rapidly and become infeasible on relatively small domains, 
whereas the growth of \ouralgo{} is markedly smoother. 
The same qualitative pattern appears for $m=2$, $m=4$, and $m=6$.

\begin{figure}[htbp]
    \centering
    \begin{subfigure}[b]{0.32\textwidth}
        \includegraphics[width=\linewidth]{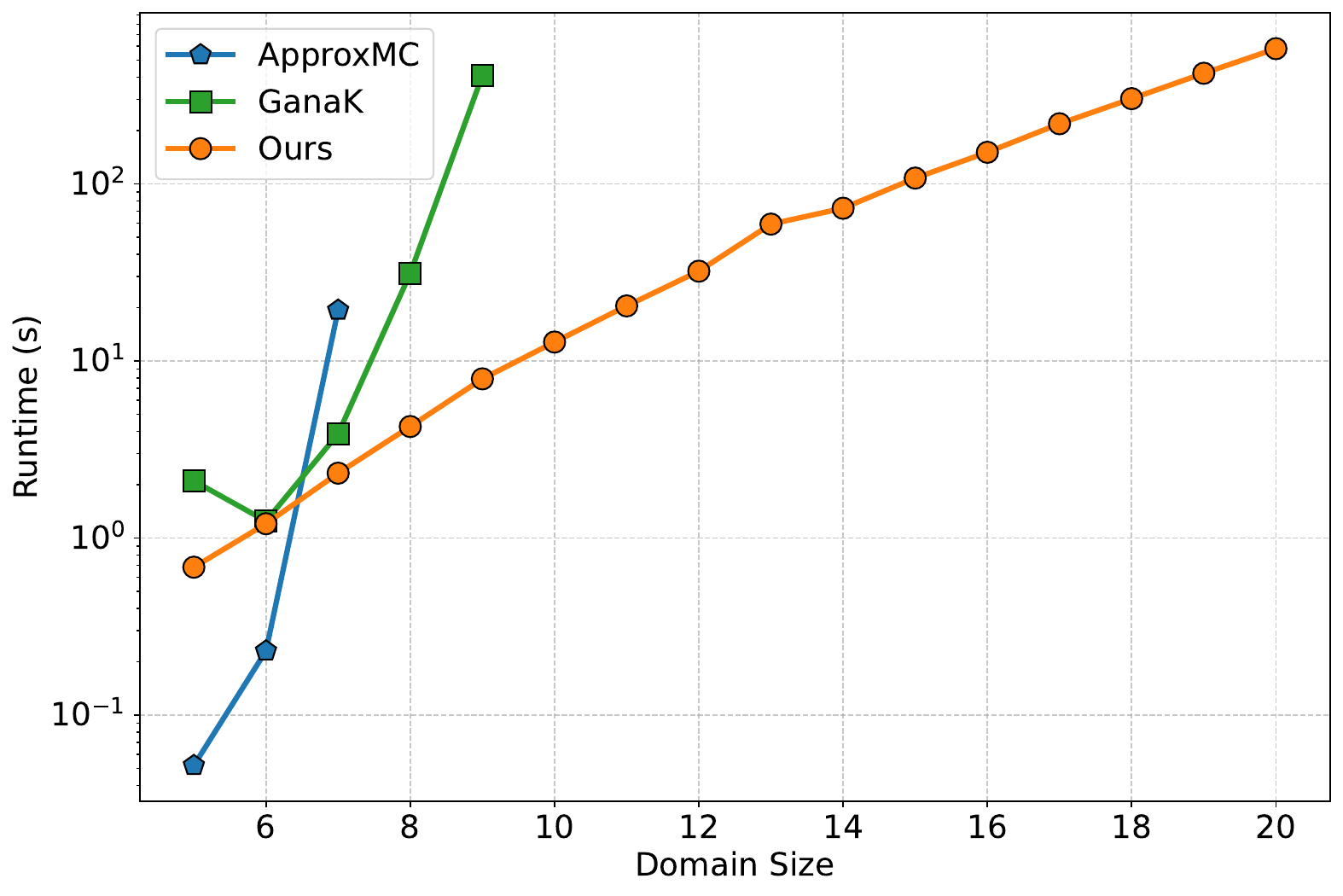}
        \caption{$2$-odd-degree graph}
    \end{subfigure}
    \begin{subfigure}[b]{0.32\textwidth}
        \includegraphics[width=\linewidth]{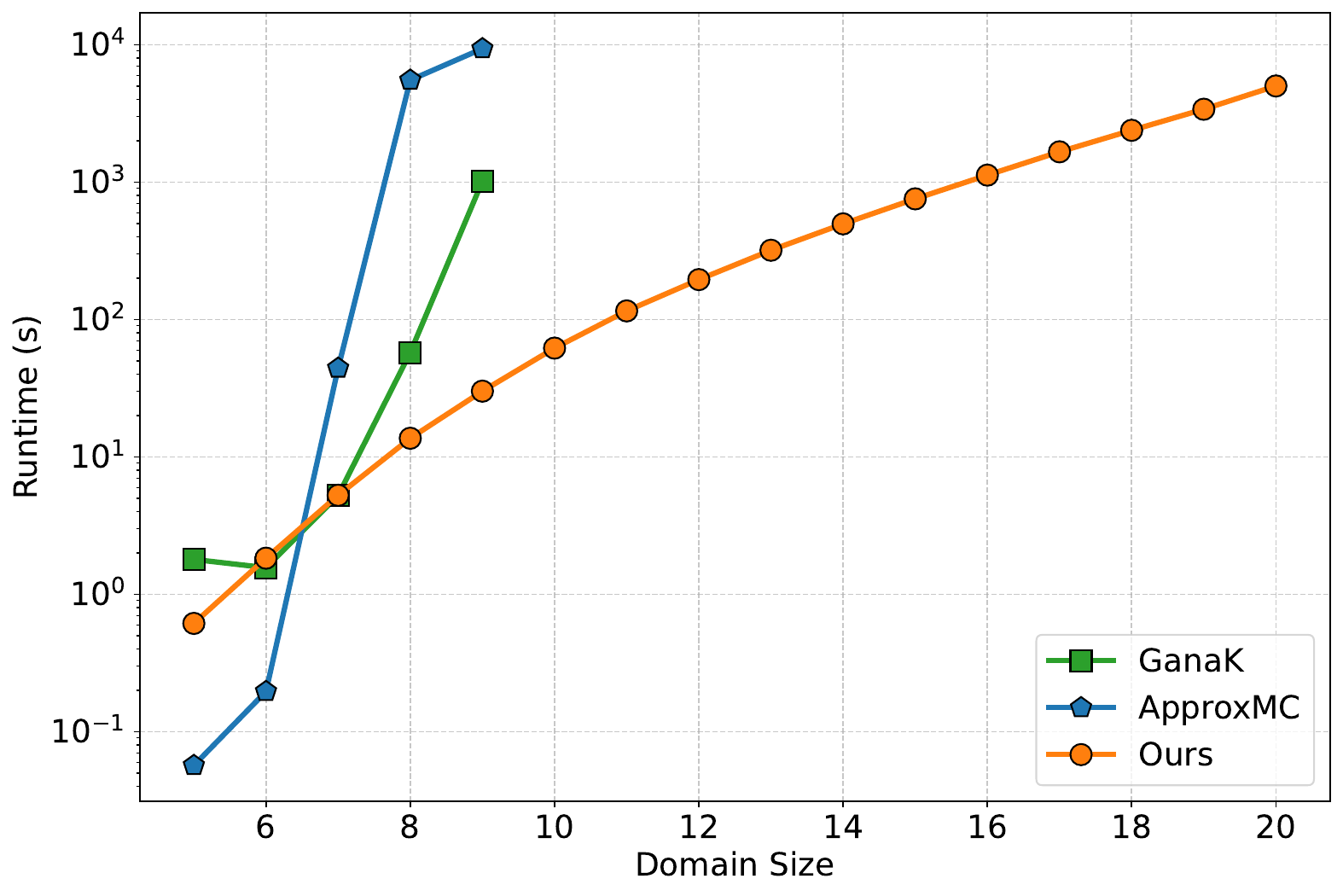}
        \caption{$4$-odd-degree graph}
    \end{subfigure}
    \begin{subfigure}[b]{0.32\textwidth}
        \includegraphics[width=\linewidth]{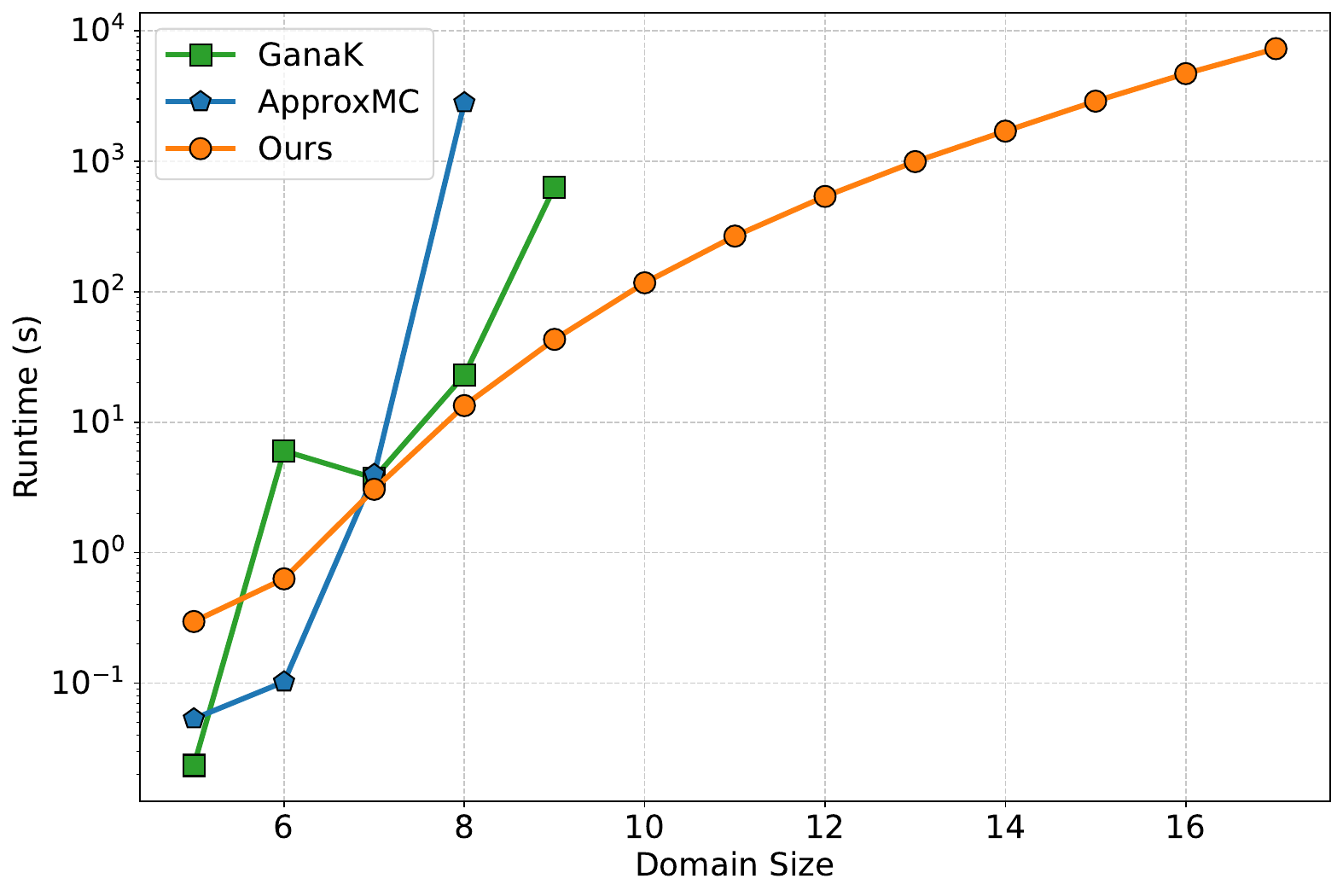}
        \caption{$6$-odd-degree graph}
    \end{subfigure}
    \caption{Runtime comparison for counting $m$-odd-degree graphs with undirected edges $k=2n$ and
    (a) $m=2$, (b) $m=4$, and (c) $m=6$.}\label{fig:odd_degree_time}
\end{figure}

Beyond runtime,
we also examined the model counts of $\sentence_{\text{$m$-odd-degree}}$, denoted by $T(n,m,k)$.
For some parameter settings, 
these counts coincide with known integer sequences in the \href{https://oeis.org/}{OEIS}.
For instance, $T(n,0,k)$ counts the number of labeled even-degree graphs with $n$ vertices and $k$ edges
and corresponds to OEIS sequence \href{https://oeis.org/A058878}{A058878},
while $T(n,2k,k)$ counts matchings of size $k$ in the complete graph $K_n$ and
corresponds to OEIS sequence \href{https://oeis.org/A100861}{A100861}.
We also provide a full table of $T(n,m,k)$ for $n \leq 8$, $k \leq 13$, and $m \leq n$ in \Cref{app:sequences}, which may be of independent combinatorial interest.

\section{Conclusion}\label{sec:conclusion}

In this paper,
we introduced \ouralgo{}, an efficient algorithm for weighted first-order model counting (\wfomc{}) on \ctwo{} and its modulo counting extension \ctwomod{}.
With \ouralgo{}, we obtain two main theoretical consequences:
For \ctwo{}, it improves the best known data-complexity bound by 
reducing the dependence on the counting parameters from quadratic to linear in the exponent,
and for \ctwomod{}, it establishes domain-liftability for the first time.
Across benchmark families ranging from standard counting tasks to linear-order and modulo counting settings, \ouralgo{} consistently scales to larger domains than existing lifted or propositional baselines, often by substantial margins. 

Several promising avenues for future work remain open.
On the theoretical side, it would be interesting to determine whether the dependence on the 
counting parameters can be reduced further. On the algorithmic side, 
a natural next step is to extend the same direct state-tracking viewpoint to 
richer liftable fragments that combine counting or modulo counting with additional axioms, 
and to investigate broader applications of the resulting methods in lifted probabilistic inference and combinatorial counting.

\section*{Acknowledgements}
Shixin Sun and Yuanhong Wang were supported by the National Natural Science Foundation of China (Grant No. 62506141). 
Astrid Klipfel and Ond\v{r}ej Ku\v{z}elka were supported by the Czech Science Foundation under Project 24-11820S (Automatic Combinatorialist).
Yi Chang was supported by the New Cornerstone Science Foundation through the XPLORER PRIZE.


\bibliographystyle{plainnat}
\bibliography{cas-refs}

\clearpage

\appendix

\section{\wfomc{} Reductions from \ctwomod{} to the Normal Form}
\label{app:modk_reduction}
In this section, we show how to transform a \ctwomod{} sentence into the normal form in \cref{eq:sctwo_mod}.
To simplify the presentation, 
throughout this section, 
whenever fresh auxiliary predicates are introduced 
and their weights are not specified explicitly, 
we extend the weighting functions by assigning them unit weights, 
i.e., for every such fresh predicate $P$, we set $w(P)=\bar{w}(P)=1$.
The techniques used in this section mainly follow the same principles as the reduction from \ctwo{} to \sctwo{} in \citep{KR2024-64}.
Throughout this section, we assume without loss of generality that all counting and modulo counting parameters are smaller than domain size $n$, as the remaining cases concern only finitely many domain sizes and do not affect the data complexity.


\subsection{Transforming $\exists^{\le k}$, $\exists^{\ge k}$, $\exists^{\le r', k'}$, and $\exists^{\ge r', k'}$}

First, we transform counting quantifiers of the form $\exists^{\le k}$ and $\exists^{\ge k}$ into $\exists^{=k}$, and transform modulo counting quantifiers of the form $\exists^{\le r',k'}$ and $\exists^{\ge r',k'}$ into $\exists^{=r',k'}$.\footnote{Since our new algorithm can directly handle both $\exists^{=k}$ and $\exists^{\leq k}$ quantifiers, we can keep $\exists^{\leq k}$ and $\exists^{\le r',k'}$ quantifiers as they are. Obviously, this significantly simplifies the transformation.}
For $\exists^{\le k}$ and $\exists^{\le r',k'}$, it is easy: We rewrite $\exists^{\le k}$ as $\bigvee_{i=0}^k \exists^{=i}$, and rewrite $\exists^{\le r',k'}$ as $\bigvee_{i=0}^{r'} \exists^{=i,k'}$.
For $\exists^{\ge k}$ and $\exists^{\ge r',k'}$ , we can first rewrite $\exists^{\ge k}$ as $\neg \left(\exists^{\le k-1}\right)$ and $\exists^{\ge r',k'}$ as $\neg \left(\exists^{\le r'-1, k'}\right)$, and then apply the above transformation to $\exists^{\le k-1}$ and $\exists^{\le r'-1, k'}$.
Note that these transformations introduce only a constant number of subformulas for each counting or modulo counting quantifier, and the resulting formula is of size linear in the original formula (and independent of the domain size).

\subsection{Subformula Axiomatization}\label{subsec:subformula_isolation}

Next, we axiomatize counting and modulo counting subformulas, transforming the original \ctwomod{} sentence into:
\begin{align}
    \Psi \land &\bigwedge_{i=1}^M \forall x: \left(P_i(x) \leftrightarrow \exists^{= k_i} y: \psi_i(x,y)\right) \land \label{eq:counting_isolation}
    \bigwedge_{j=1}^N \left(Q_j() \leftrightarrow  \exists^{= u_j} x: \phi_j(x)\right) \land \\
    &\bigwedge_{i=1}^{M'} \forall x: \left(P'_i(x) \leftrightarrow \exists^{= r'_i, k'_i} y: \psi'_i(x,y)\right) \land \label{eq:mod_counting_isolation}
    \bigwedge_{j=1}^{N'} \left( Q'_j() \leftrightarrow \exists^{=v'_j,u'_j} x: \phi'_j(x) \right), 
\end{align}
where $\Psi$ is a \fotwo{} sentence, $P_i$ and $P'_i$ are fresh unary predicates, $Q_j$ and $Q'_j$ are fresh nullary predicates, and $\psi_i$, $\psi'_i$, $\phi_j$, and $\phi'_j$ are formulas that do not contain counting or modulo counting quantifiers.

The process is as follows: Starting from the innermost counting or modulo counting subformulas and working outwards, replace any subformula $\exists^{=k} y: \psi(x,y)$ and $\exists^{=r', k'} y: \psi'(x,y)$, where $\psi(x,y)$ and $\psi'(x,y)$ do not contain counting or modulo counting quantifiers, with fresh unary predicates $P(x)$ and $P'(x)$, respectively. We then add corresponding definitions $\forall x: \left(P(x) \leftrightarrow \exists^{=k} y: \psi(x,y)\right)$ and $\forall x: \left(P'(x) \leftrightarrow \exists^{=r', k'} y: \psi'(x,y)\right)$ to the resulting formula.
Similarly, replace any subformula $\exists^{=u} x: \phi(x)$ and $\exists^{=v', u'} x: \phi'(x)$, where $\phi(x)$ and $\phi'(x)$ do not contain counting or modulo counting quantifiers, with fresh nullary predicates $Q()$ and $Q'()$, respectively, and add corresponding definitions $Q() \leftrightarrow \exists^{=u} x: \phi(x)$ and $Q'() \leftrightarrow \exists^{=v', u'} x: \phi'(x)$ to the resulting formula.
It is easy to check that the above transformation preserves WFOMC, since each replacement introduces only a fresh predicate together with a defining equivalence, and all fresh predicates are assigned unit weights.

Before proceeding to the next step, we further replace the equivalences in \cref{eq:counting_isolation,eq:mod_counting_isolation} with $\land$, $\lor$ and $\neg$. 
For instance, $\forall x: \left(P_i(x) \leftrightarrow \exists^{= k_i} y: \psi_i(x,y)\right)$ is replaced by
\begin{align}
    &\forall x: \left(\neg P_i(x) \lor \left(\exists^{= k_i} y: \psi_i(x,y)\right)\right) \land \label{eq:implication1}\\
    &\forall x: \left(P_i(x) \lor \neg \left(\exists^{= k_i} y: \psi_i(x,y)\right)\right) \label{eq:implication2},
\end{align}
and $\left(Q_j() \leftrightarrow  \exists^{= u_j} x: \phi_j(x)\right)$ is replaced by
\begin{align}
    &\left(\neg Q_j() \lor \left(\exists^{= u_j} x: \phi_j(x)\right)\right) \land \label{eq:implication3}\\
    &\left(Q_j() \lor \neg \left(\exists^{= u_j} x: \phi_j(x)\right)\right) \label{eq:implication4}.
\end{align}

\subsection{Negation Elimination}

For the formulas of the form of \cref{eq:implication2,eq:implication4} (as well as their modulo counting counterparts), we can apply the following lemma from~\cite[Appendix A.2]{Beame_2015} such that the negation of counting and modulo counting subformulas is eliminated.
\begin{lemma}\label{lemma:neg_counting}
    Let $\neg \phi(x_1, \ldots, x_l)$ be a subformula of a first-order sentence 
    $\Phi$ with $l \ge 0$ free variables. 
    Let $A$ and $B$ be two new predicates of arity $l$.
    Denote by $\Phi'$ the sentence obtained from $\Phi$ by replacing $\neg \phi(x_1, \ldots, x_l)$ with $A(x_1, \ldots, x_l)$.
    Let
    \begin{align*}
        \Upsilon &= \forall x_1\forall x_2\dots\forall x_l: \Big(\left(A(x_1, \ldots, x_l) \lor \phi(x_1, \ldots, x_l)\right) \land \left(B(x_1, \ldots, x_l) \lor \phi(x_1, \ldots, x_l)\right)\land \\
        &\left(A(x_1, \ldots, x_l) \lor B(x_1, \ldots, x_l)\right)\Big),
    \end{align*}
    and extend the weight function $\weight, \negweight$ to $A$ and $B$ by setting $\weight(A) = \weight(B) = \negweight(A) = 1$ and $\negweight(B) = - 1$.
    Then for every domain size $n$, it holds that
    \begin{equation*}
        \symwfomc(\Phi, n, \weight, \negweight) = \symwfomc(\Phi' \land \Upsilon, n, \weight, \negweight).
    \end{equation*}
\end{lemma}

Applying the above lemma to the formulas of the form \cref{eq:implication2} and its modulo counting counterpart, the resulting formulas are of the forms, 
\begin{equation}\label{eq:negation_elimination1}
\begin{aligned}
    &\forall x: \left(\tilde{P}_i(x) \lor \left(\exists^{= k_i} y: \psi_i(x,y)\right)\right), \\
    &\forall x: \left(\tilde{P}'_i(x) \lor \left(\exists^{= r'_i, k'_i} y: \psi'_i(x,y)\right)\right).
\end{aligned}
\end{equation}
Similarly, applying the lemma to the formulas of the form \cref{eq:implication4} and its modulo counting counterpart, the resulting formulas are of the forms,
\begin{equation}\label{eq:negation_elimination2}
\begin{aligned}
    &\left(\tilde{Q}_j() \lor \left(\exists^{= u_j} x: \phi_j(x)\right)\right), \\
    &\left(\tilde{Q}'_j() \lor \left(\exists^{= v'_j, u'_j} x: \phi'_j(x)\right)\right).
\end{aligned}
\end{equation}
Note that we can further rewrite \cref{eq:implication1,eq:implication3} into the same form as \cref{eq:negation_elimination1,eq:negation_elimination2}, respectively, by replacing $\neg P_i(x)$ and $\neg Q_j()$ with fresh predicates $\tilde{P}_i(x)$ and $\tilde{Q}_j()$.
Now, the original \ctwomod{} sentence is transformed into a sentence, where the counting and modulo counting subformulas only appear in the form of \cref{eq:negation_elimination1} and \cref{eq:negation_elimination2}, and the rest of the sentence is a \fotwo{} sentence.

\subsection{Shannon Expansion of Nullary Predicates}

Let us first handle the formulas of the form \cref{eq:negation_elimination2}.
Since $\tilde{Q}_j$ and $\tilde{Q}'_j$ are nullary predicates, one can apply Shannon expansion to eliminate them, which decomposes the WFOMC computation into a constant number of subproblems.
That is, consider all possible truth assignments to the nullary predicates $\tilde{Q}_j$ and $\tilde{Q}'_j$, and for each assignment, replace the corresponding nullary predicates with their assigned truth values in the sentence.
One can easily check that the resulting sentence can only contain unary counting and modulo counting subformulas of the form
\begin{equation*}
    \bigwedge_{j=1}^N \left(\exists^{= u_j} x: \phi_j(x)\right) \land \bigwedge_{j=1}^{N'} \left(\exists^{= v'_j, u'_j} x: \phi'_j(x)\right),
\end{equation*}
which aligns with the normal form in \cref{eq:sctwo_mod} (after axiomatizing the subformulas $\phi_j(x)$ and $\phi'_j(x)$ by fresh predicates, as described in \Cref{subsec:subformula_isolation}).

\subsection{Transforming Disjunctions over Counting and modulo counting Subformulas}

Now, the remaining formulas that are not yet in the normal form are of the forms \cref{eq:negation_elimination1}.
Before we perform the transformation, we axiomatize the subformulas $\psi_i(x,y)$ and $\psi'_i(x,y)$ by adding $\forall x \forall y: \left(R_i(x,y) \leftrightarrow \psi_i(x,y)\right)$ and $\forall x \forall y: \left(R'_i(x,y) \leftrightarrow \psi'_i(x,y)\right)$, where $R_i$ and $R'_i$ are fresh binary predicates, and then replace $\psi_i(x,y)$ and $\psi'_i(x,y)$ with $R_i(x,y)$ and $R'_i(x,y)$, respectively.
This yields the formulas $\forall x: \left(\tilde{P}_i(x) \lor \exists^{= k_i} y: R_i(x,y)\right)$ and $\forall x: \left(\tilde{P}'_i(x) \lor \exists^{= r'_i, k'_i} y: R'_i(x,y)\right)$.

For the formulas of the form $\forall x: \left(\tilde{P}_i(x) \lor \exists^{= k_i} y: R_i(x,y)\right)$, we can apply the following lemma from~\cite[Lemma~4]{kuzelka2021weighted} to transform them into the normal form:

\begin{lemma}[{~\cite[Lemma 4]{kuzelka2021weighted}}]
    \label{lemma:counting_reduction}
    Let $\Phi$ be a first-order sentence, $A$ a unary predicate, $R$ a binary predicate, and $U$ and $B$ fresh unary and binary predicates, respectively, that do not occur in $\Phi$, and let $k$ be a non-negative integer.
    Define
    \begin{align*}
        \Upsilon_1 &= \forall x\exists^{= k} y: B(x,y), \\
        \Upsilon_2 &= (|U| = k), \\
        \Upsilon_3 &= \forall x\forall y: \left(A(x) \land B(x,y) \to U(y)\right), \\
        \Upsilon_4 &= \forall x\forall y: \neg A(x) \rightarrow \left(B(x,y) \leftrightarrow R(x,y)\right).
    \end{align*}
    Then for every domain size $n$, it holds that
    \begin{equation*}
        \symwfomc(\Phi \land \forall x: \left(A(x) \lor \exists^{= k} y: R(x,y)\right), n, \weight, \negweight) = \frac{1}{\binom{n}{k}} \symwfomc(\Phi \land \Upsilon_1 \land \Upsilon_2 \land \Upsilon_3\land \Upsilon_4, n, \weight, \negweight).
    \end{equation*}
\end{lemma}

The formulas of the form $\forall x: \left(\tilde{P}'_i(x) \lor \exists^{= r'_i, k'_i} y: R'_i(x,y)\right)$ can be handled analogously, via the following lemma, whose proof essentially follows the same line of reasoning as that of \Cref{lemma:counting_reduction}.

\begin{lemma}
    \label{lemma:counting_mod_reduction}
    Let $\Phi$ be a first-order sentence, $A$ a unary predicate, $R$ a binary predicate, and $U$ and $B$ fresh unary and binary predicates, respectively, that do not occur in $\Phi$, and let $r$ and $k$ be non-negative integers with $0 \le r < k$.
    Define
    \begin{align*}
        \Upsilon_1 &= \forall x \exists^{=r, k} y : B(x, y), \\
        \Upsilon_2 &= (|U| = r), \\
        \Upsilon_3 &= \forall x \forall y : (A(x) \land B(x, y) \to U(y)), \\
        \Upsilon_4 &= \forall x \forall y : \neg A(x) \to (B(x, y) \leftrightarrow R(x, y)).
    \end{align*}
    Then for every domain size $n$, it holds that
    \begin{equation*}
        \symwfomc(\Phi \land \forall x : \left(A(x) \lor \exists^{=r,k} y : R(x, y)\right), n, \weight, \negweight) = \frac{1}{\binom{n}{r}} \symwfomc(\Phi \land \Upsilon_1 \land \Upsilon_2 \land \Upsilon_3 \land \Upsilon_4, n, \weight, \negweight).
    \end{equation*}
\end{lemma}

\begin{proof}
We prove the equivalence in both directions.

$(\Leftarrow)$ First, we show that every model of $\Phi \land \Upsilon_1 \land \Upsilon_2 \land \Upsilon_3 \land \Upsilon_4$ satisfies $\Phi \land \forall x : (A(x) \lor \exists^{=r,k} y : R(x, y))$ by contradiction.
Assume there exists a model $\mu$ 
such that $\mu \models \Upsilon_1 \land \Upsilon_2 \land \Upsilon_3 \land \Upsilon_4$ 
but $\mu \not\models \forall x : (A(x) \lor \exists^{=r,k} y : R(x, y))$.
This assumption implies that there must be an element $t \in \Delta$ for which $\mu \models \neg(A(t) \lor \exists^{=r,k} y : R(t, y))$, which is equivalent to $\mu \models \neg A(t)$ and $\mu \models \neg(\exists^{=r,k} y : R(t, y))$.
From the former and $\Upsilon_4$, we have that $\mu \models \forall y : (B(t, y) \leftrightarrow R(t, y))$.
This combined with $\Upsilon_1$ ,implies that $\mu \models \exists^{=r,k} y : R(t, y)$, leading to a contradiction.

$(\Rightarrow)$
To complete the proof, we show that for every model of $\Phi \land \forall x : (A(x) \lor \exists^{=r,k} y : R(x, y))$, there are exactly $\binom{n}{r}$ corresponding models of $\Phi \land \Upsilon_1 \land \Upsilon_2 \land \Upsilon_3 \land \Upsilon_4$.
Let $\mu$ be any model satisfying $\forall x : (A(x) \lor \exists^{=r,k} y : R(x, y))$. 
We can extend $\mu$ to a model $\mu'$ over the language augmented with the predicates $U$ and $B$ such that $\mu' \models \Phi \land \Upsilon_1 \land \Upsilon_2 \land \Upsilon_3 \land \Upsilon_4$ and its projection back to the original language is $\mu$.
The interpretation of the new predicates $U$ and $B$ in $\mu'$ is defined as follows.
First, we choose exactly $r$ elements from the domain $\domain$ to be in $U$.
Then for every pair of domain elements $(t_1, t_2)$, the ground atom $B(t_1, t_2)$ is true in $\mu'$ if and only if either $\mu \models A(t_1)$ and $t_2$ is in $U$, or $\mu \models \neg A(t_1)$ and $\mu \models R(t_1, t_2)$.
We show that $\mu'$ satisfies each of $\Upsilon_1$, $\Upsilon_2$, $\Upsilon_3$, and $\Upsilon_4$.
Clearly, $\mu'$ satisfies $\Upsilon_2$ by construction.
It is also easy to check that $\mu'$ satisfies $\Upsilon_3$ and $\Upsilon_4$.
We now verify that $\mu'$ satisfies $\Upsilon_1$.
Consider any element $t \in \Delta$. If $\mu \models A(t)$, then by our construction $B(t, y)$ is true for exactly the $r$ elements in $U$, and thus $\mu' \models \exists^{=r,k} y : B(t, y)$.
If, on the other hand, $\mu \models \neg A(t)$, then the interpretation of $B(t, y)$ in $\mu'$ coincides with that of $R(t, y)$ in $\mu$, which satisfies $\exists^{=r,k} y : R(t, y)$ ,meaning that $\exists^{=r,k} y : B(t, y)$ also holds.
Thus, $\mu'$ satisfies $\Upsilon_1$, and $\mu' \models \Phi \land \Upsilon_1 \land \Upsilon_2 \land \Upsilon_3 \land \Upsilon_4$.
Recall that there are $\binom{n}{r}$ ways to choose the interpretation of $U$.
For each such choice, we obtain a distinct model $\mu'$.
Therefore, the \wfomc{} of $\Phi \land \Upsilon_1 \land \Upsilon_2 \land \Upsilon_3 \land \Upsilon_4$ counts each model of $\Phi \land \forall x : (A(x) \lor \exists^{=r,k} y : R(x, y))$ exactly $\binom{n}{r}$ times, from which the factor $\frac{1}{\binom{n}{r}}$ in the equivalence follows.
\end{proof}



Applying \Cref{lemma:counting_reduction,lemma:counting_mod_reduction} consecutively to all formulas of the form $\forall x: \left(\tilde{P}_i(x) \lor \exists^{= k_i} y: R_i(x,y)\right)$ and $\forall x: \left(\tilde{P}'_i(x) \lor \exists^{= r'_i, k'_i} y: R'_i(x,y)\right)$ followed by Lagrange interpolation to eliminate the cardinality constraints and Skolemization to eliminate the existential quantifiers, we can obtain a \ctwomod{} sentence in the normal form given by \cref{eq:sctwo_mod}.

Finally, we note that if we do not transform $\exists^{\ge k}$ (or $\exists^{\ge r', k'}$) quantifiers into $\exists^{\le k}$ (or $\exists^{\le r', k'}$) quantifiers at the beginning of the reduction, the above reduction can still be applied with slightly modified transformations in \Cref{lemma:counting_reduction,lemma:counting_mod_reduction}.
Specifically, we have the following variant of these two lemmas, whose proofs are analogous to those of \Cref{lemma:counting_reduction,lemma:counting_mod_reduction} and are therefore omitted.

\begin{lemma}
    Let $\Phi$ be a first-order sentence, $A$ a unary predicate, $R$ a binary predicate, and $B$ a fresh binary predicate that does not occur in $\Phi$, and let $r$ and $k$ be non-negative integers with $0 \le r < k$.
    Define
    \begin{align*}
        \Upsilon_1 &= \forall x\exists^{\le k} y: B(x,y), \\
        \Upsilon_2 &= \forall x: \left(A(x) \to \left(\forall y: \neg B(x,y)\right)\right), \\
        \Upsilon_3 &= \forall x\forall y: \neg A(x) \rightarrow \left(B(x,y) \leftrightarrow R(x,y)\right).
    \end{align*}
    Then for every domain size $n$, it holds that
    \begin{equation*}
        \symwfomc(\Phi \land \forall x: \left(A(x) \lor \exists^{\le k} y: R(x,y)\right), n, \weight, \negweight) = \symwfomc(\Phi \land \Upsilon_1 \land \Upsilon_2 \land \Upsilon_3, n, \weight, \negweight).
    \end{equation*}
    For modulo counting quantifiers, define $\Upsilon_1' = \forall x \exists^{\le r, k} y : B(x, y)$.Then, for every domain size $n$, it holds that
    \begin{equation*}
        \symwfomc(\Phi \land \forall x : (A(x) \lor \exists^{\le r,k} y : R(x, y)), n, \weight, \negweight) = \symwfomc(\Phi \land \Upsilon_1' \land \Upsilon_2 \land \Upsilon_3, n, \weight, \negweight).
    \end{equation*}
\end{lemma}

\section{Additional Experimental Results}

We provide additional experimental results in this section, including correctness validation for our algorithm and peak memory usage results.

\subsection{Correctness Validation}\label{app:correctness}

We validate \ouralgo{} on three benchmark families 
chosen to cover the main logical settings studied in the paper: 
$k$-regular graphs for counting quantifiers, 
$r$-mod-$k$-regular graphs for modulo counting quantifiers, 
and $m$-odd-degree graphs for modulo counting combined with global cardinality constraints. 
Since independent verification is feasible only on small instances, 
we restrict these experiments to small domains and compare against propositional model counters.

The results are shown in \Cref{fig:check_kregular,fig:check_modregular,fig:check_odd_degree}.
For all tested instances, the counts returned by \ouralgo{} (\Ours{})
agree with the exact counts returned by \Ganak{} and fall within the $95\%$ confidence intervals returned by \ApproxMC{}.

\begin{figure}[tbp]
    \centering
    \begin{subfigure}[b]{0.32\textwidth}
        \includegraphics[width=\linewidth]{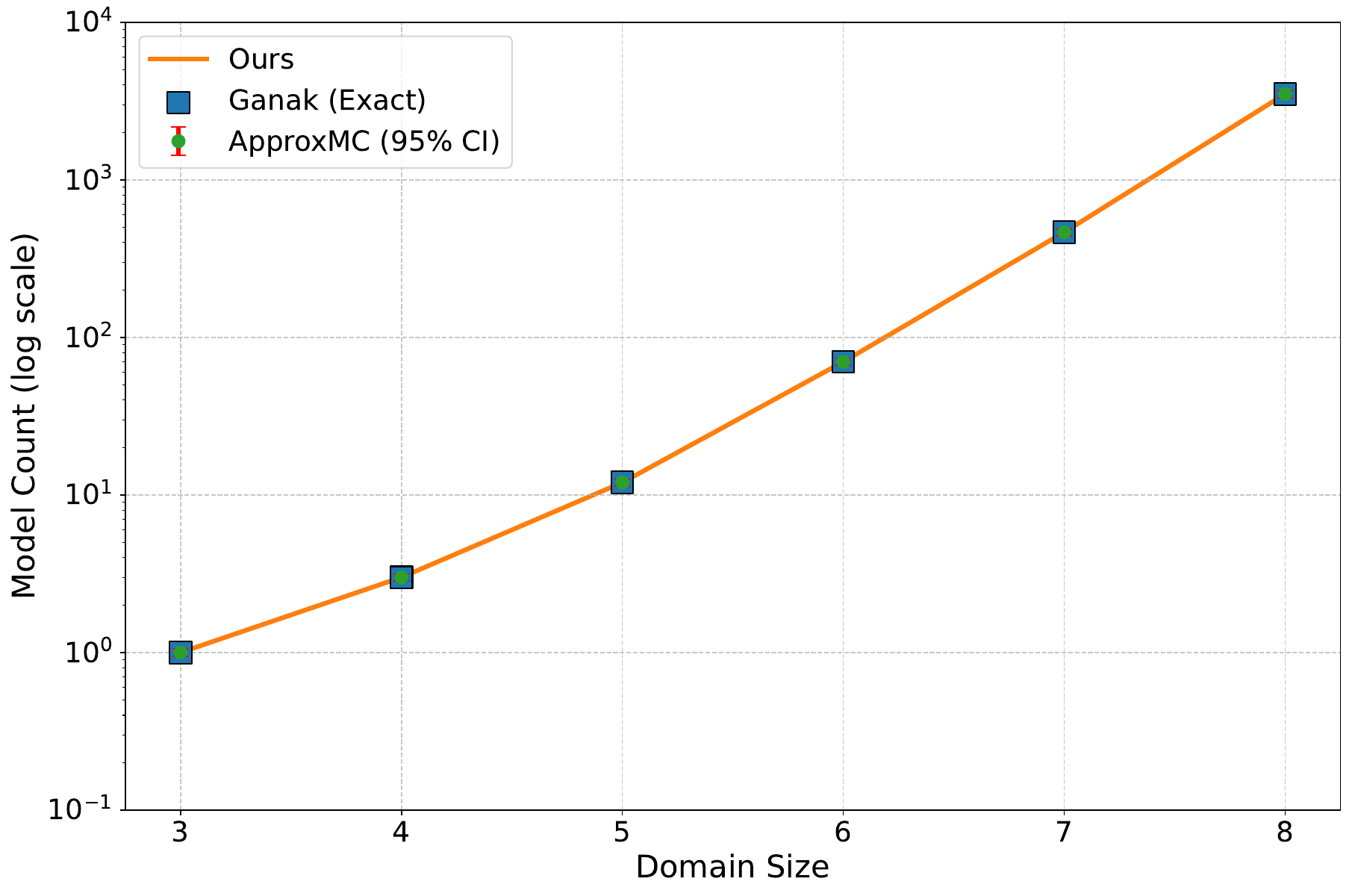}
        \caption{2-regular graphs}
    \end{subfigure}
    \begin{subfigure}[b]{0.32\textwidth}
        \includegraphics[width=\linewidth]{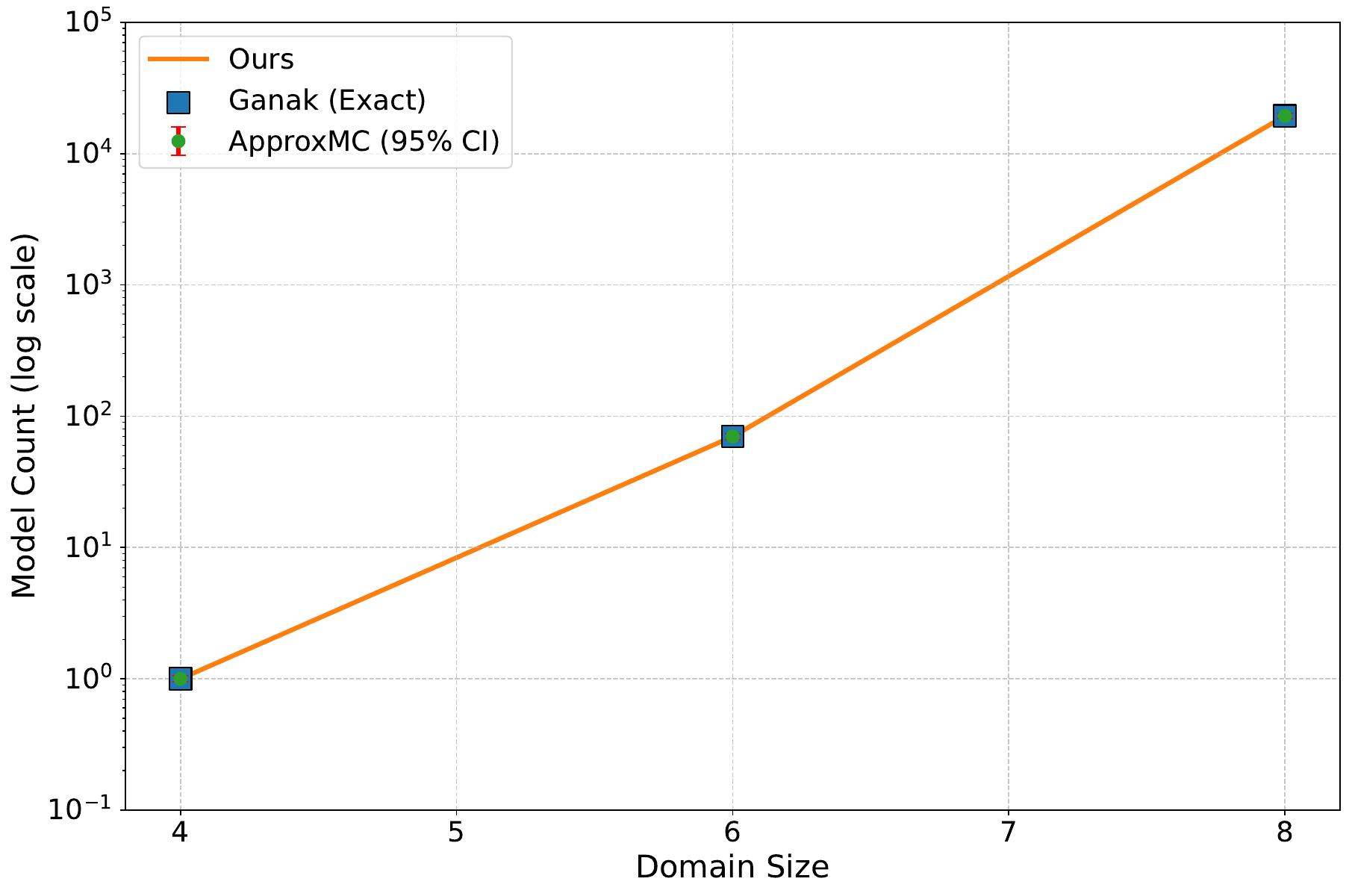}
        \caption{3-regular graphs}
    \end{subfigure}
    \begin{subfigure}[b]{0.32\textwidth}
        \includegraphics[width=\linewidth]{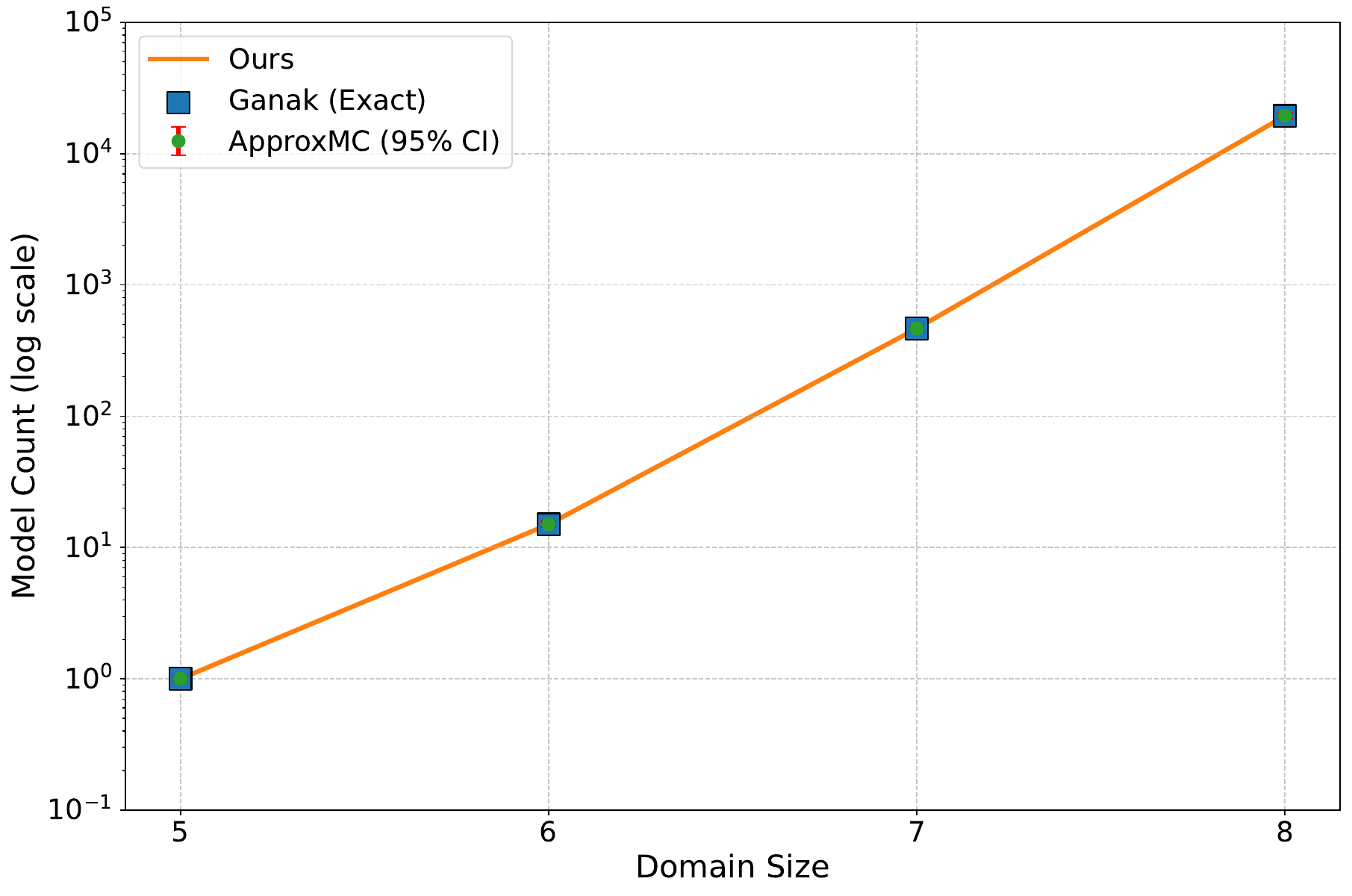}
        \caption{4-regular graphs}
    \end{subfigure}
    \caption{
    Correctness validation for counting (a) 2-regular, (b) 3-regular, and (c) 4-regular graphs. 
    }\label{fig:check_kregular}
\end{figure}

\begin{figure}[tbp]
    \centering
    \begin{subfigure}[b]{0.32\textwidth}
        \includegraphics[width=\linewidth]{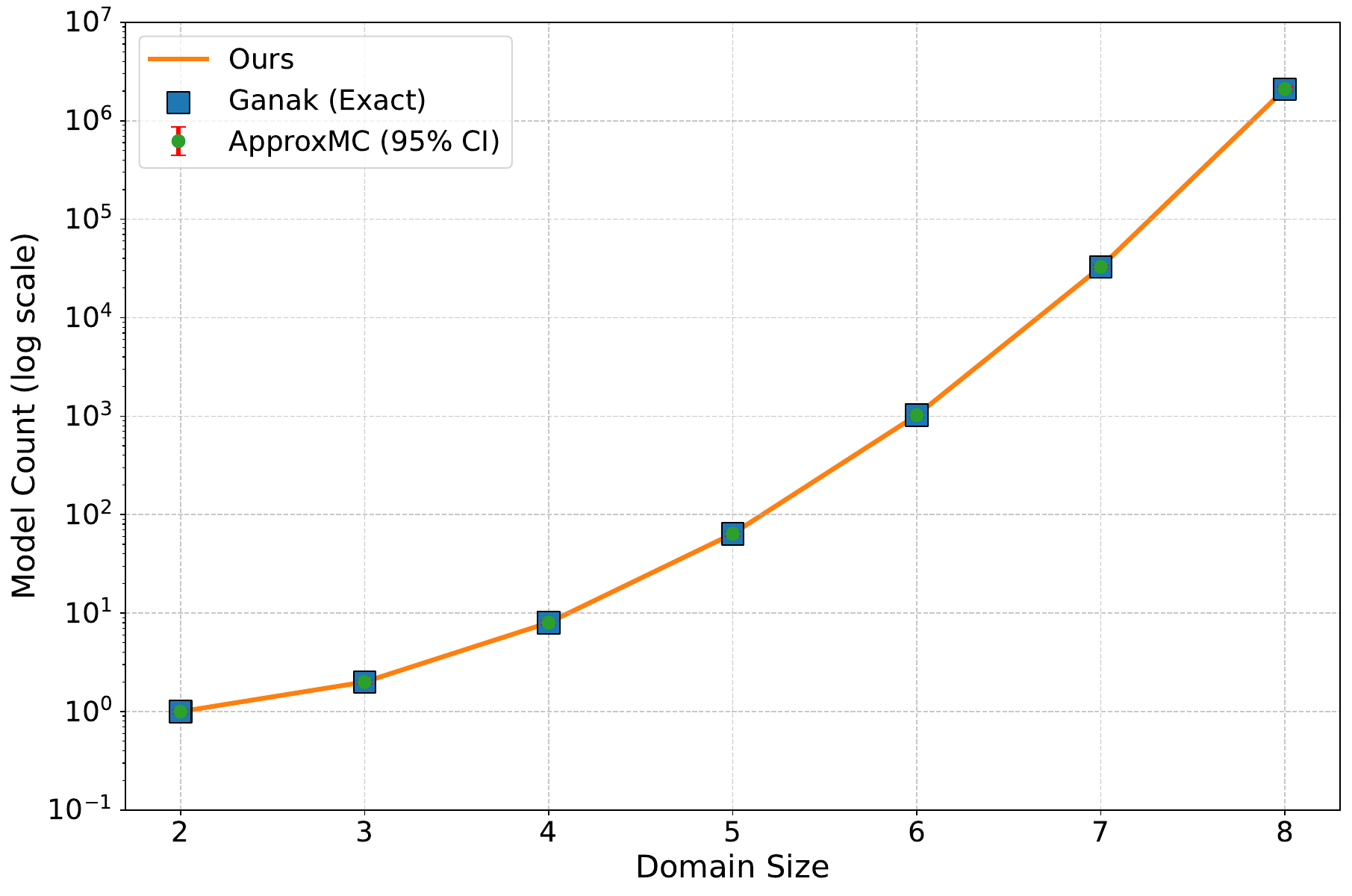}
        \caption{0-mod-2-regular graphs}
    \end{subfigure}
    \begin{subfigure}[b]{0.32\textwidth}
        \includegraphics[width=\linewidth]{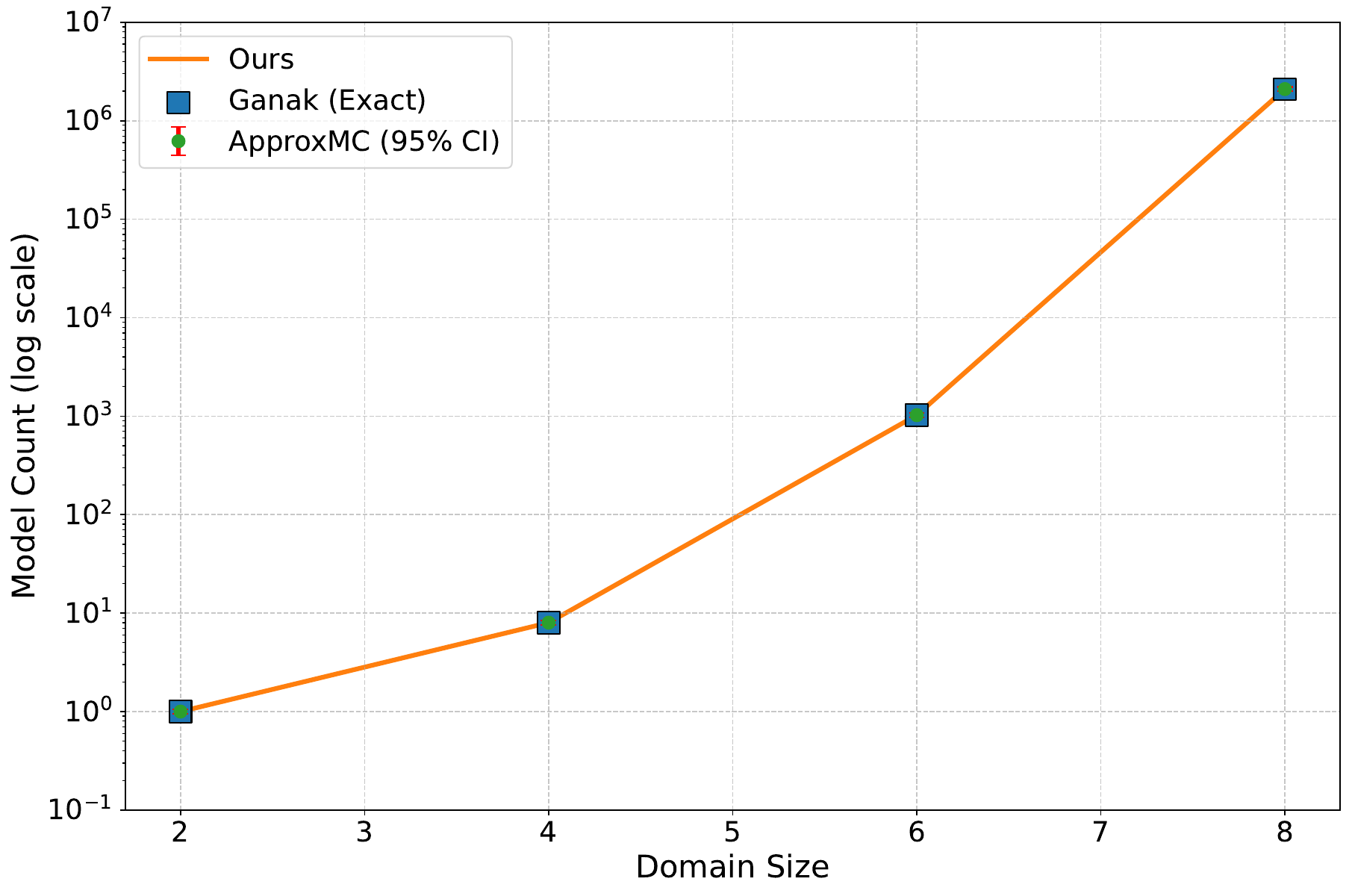}
        \caption{1-mod-2-regular graphs}
    \end{subfigure}
    \begin{subfigure}[b]{0.32\textwidth}
        \includegraphics[width=\linewidth]{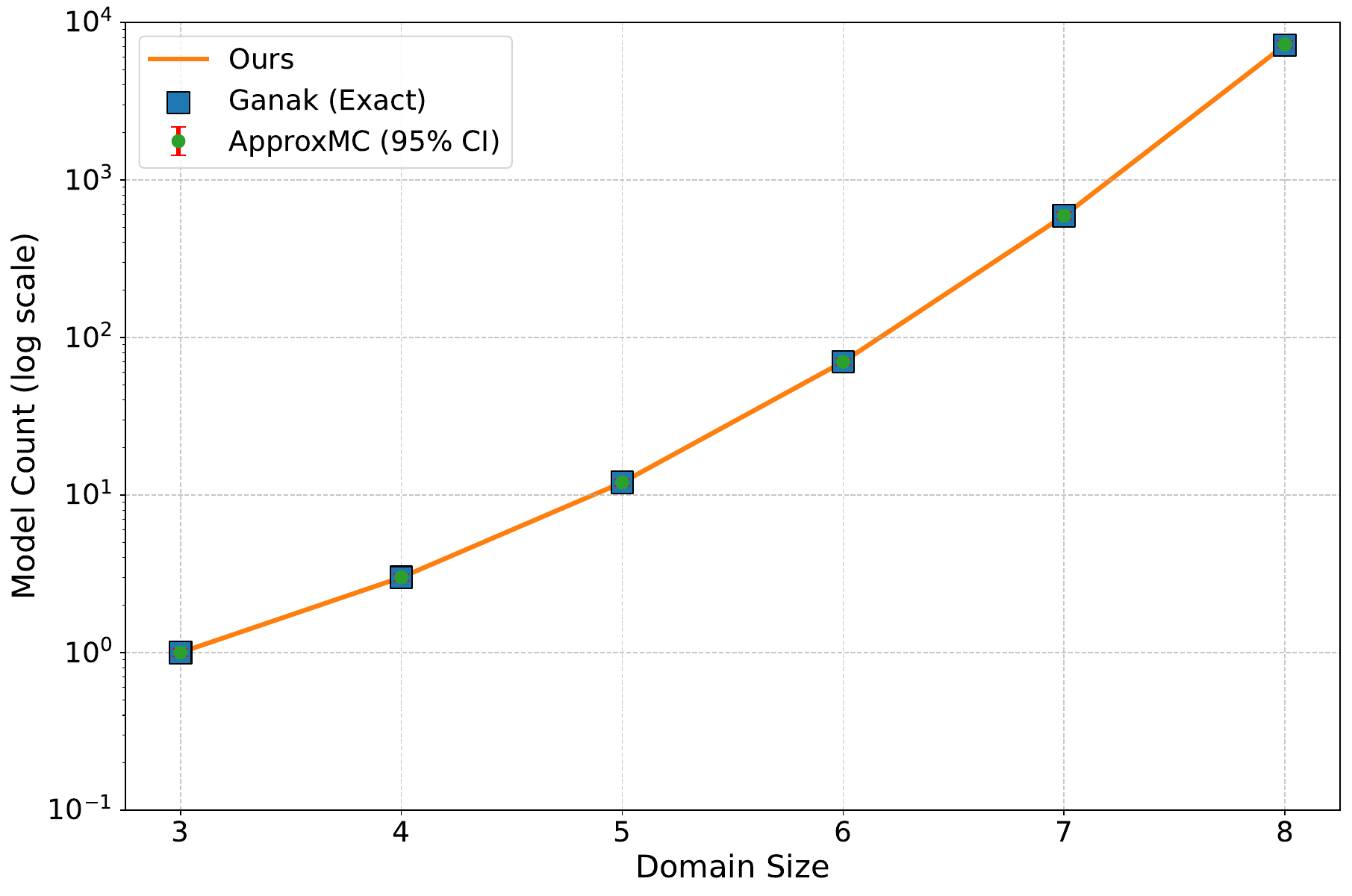}
        \caption{2-mod-4-regular graphs}
    \end{subfigure}
    \caption{
    Correctness validation for counting (a) 0mod2-regular, (b) 1mod2-regular, and (c) 2mod4-regular graphs. 
    }\label{fig:check_modregular}
\end{figure}

\begin{figure}[tbp]
    \centering
    \captionsetup[subfigure]{justification=centering}
    \begin{subfigure}[b]{0.32\textwidth}
        \includegraphics[width=\linewidth]{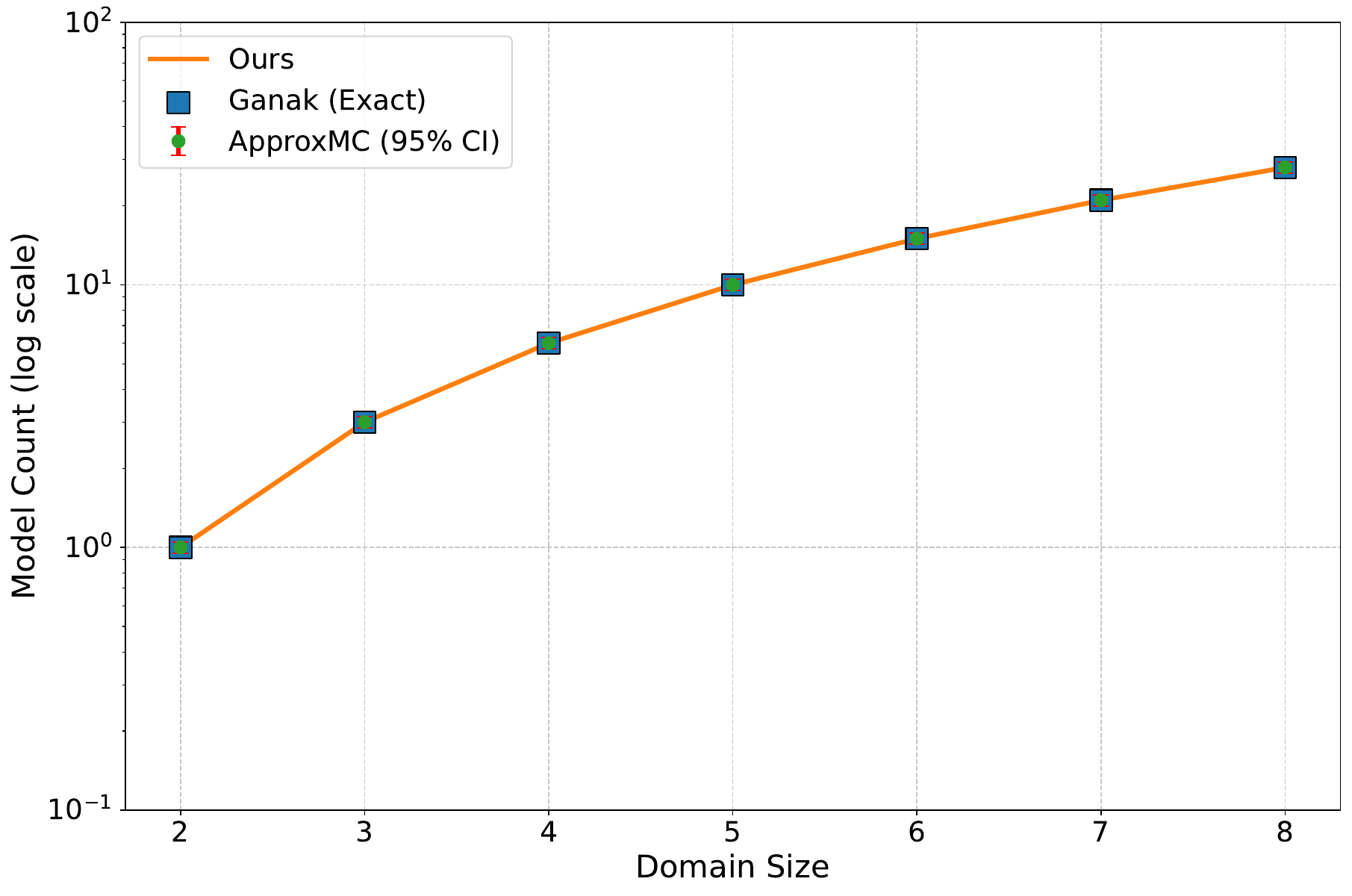}
        \caption{Varying domain size $n$ \\($m=2, k=1$)}
    \end{subfigure}
    \begin{subfigure}[b]{0.32\textwidth}
        \includegraphics[width=\linewidth]{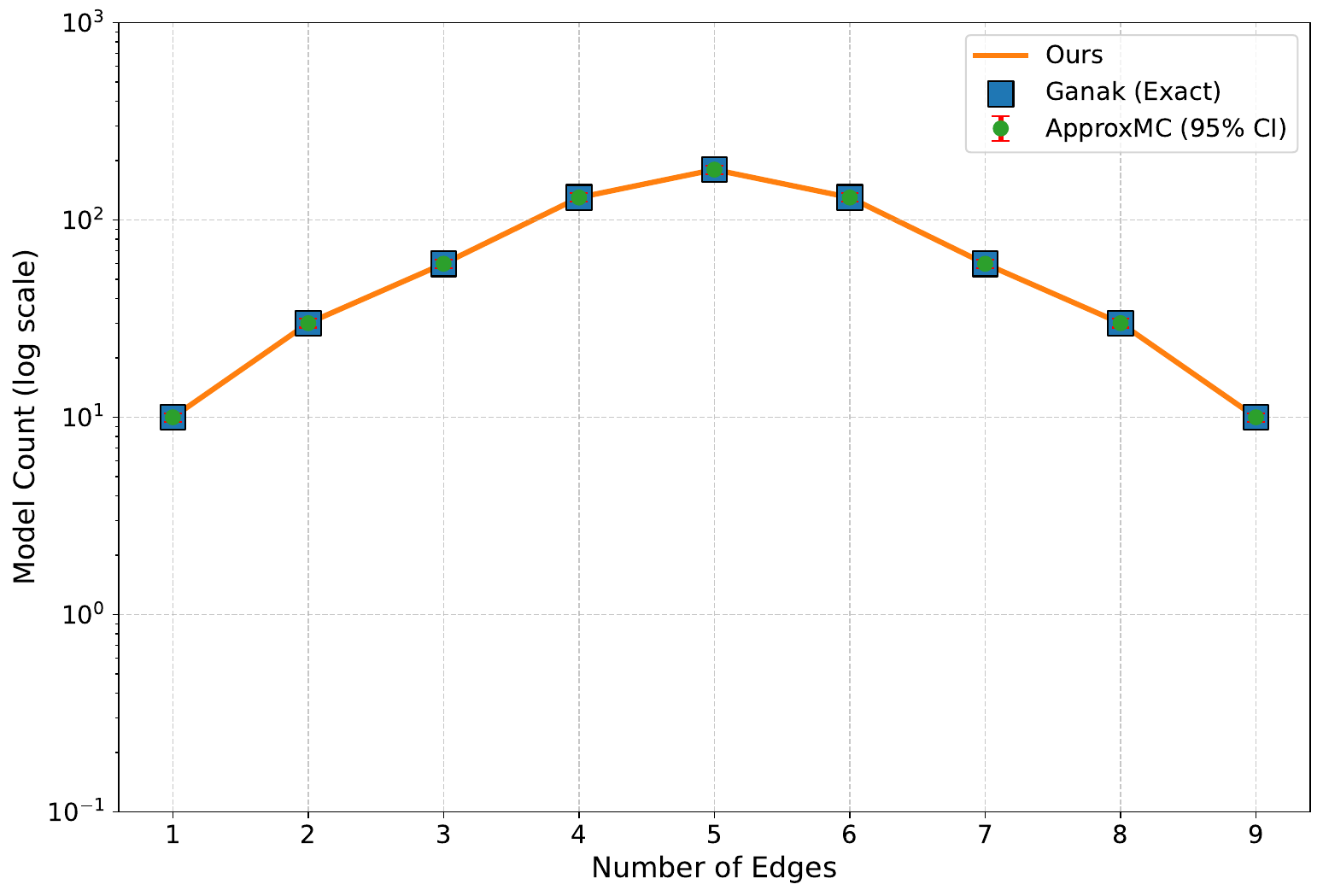}
        \caption{Varying number of edges $k$ \\ ($n=5, m=2$)}
    \end{subfigure}
    \begin{subfigure}[b]{0.32\textwidth}
        \includegraphics[width=\linewidth]{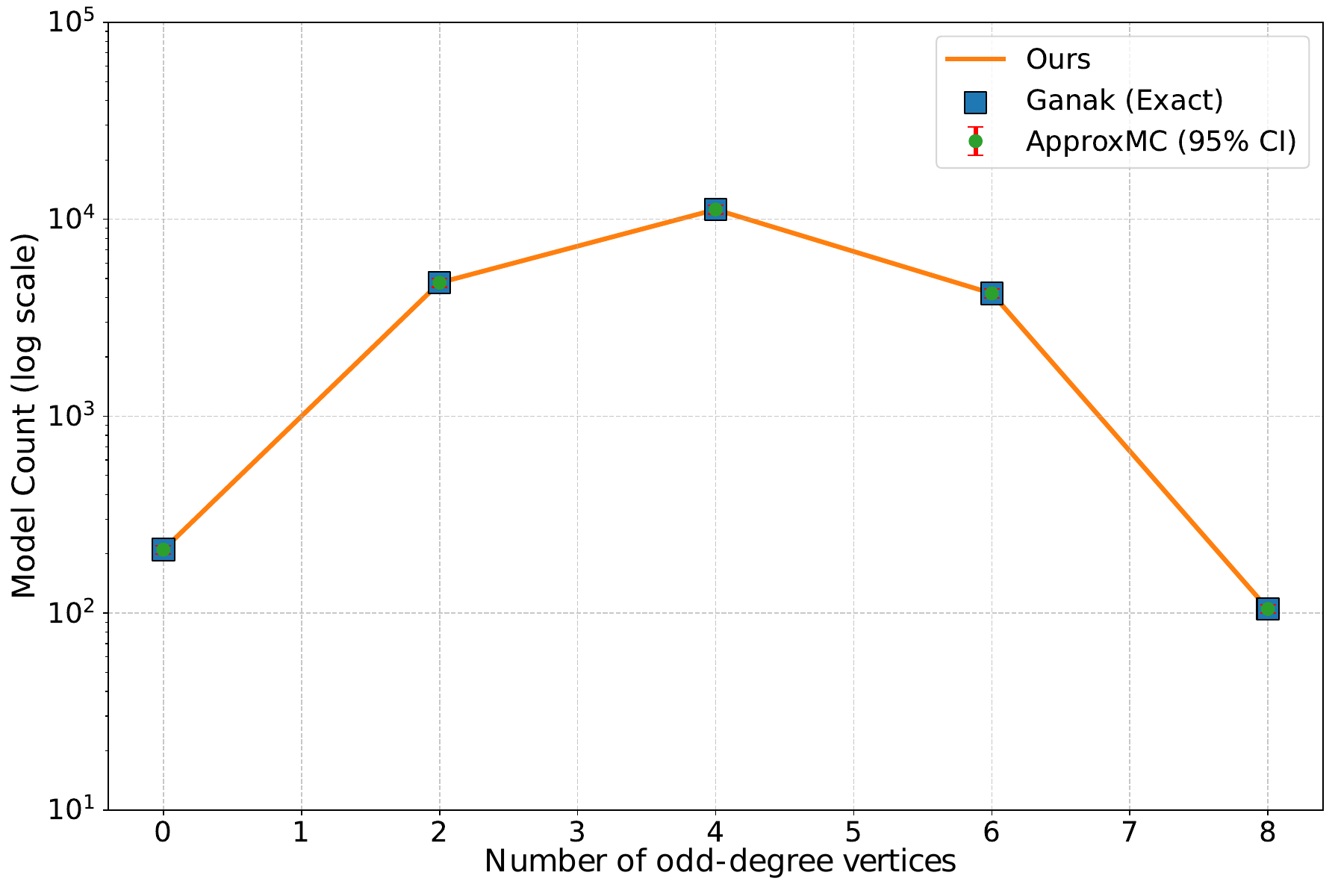}
        \caption{Varying number of odd-degree \\ vertices  $m$ ($n=8, k=3$)}
    \end{subfigure}
    \caption{
        Correctness validation for the $m$-odd-degree graph counting task from \Cref{sec:ctwomod_performance}. 
    }\label{fig:check_odd_degree}
\end{figure}

\subsection{Peak Memory Results}\label{app:memory_results}

Since \ouralgo{} is a typical dynamic programming algorithm, one would wonder about its memory usage as the domain size grows.
In this section, we report the peak memory usage of \ouralgo{} and compare it against baselines.

We measure peak memory usage using \texttt{psutil} in a dedicated monitoring thread, 
which periodically queried the target worker process via \texttt{memory\_full\_info().uss} and 
fell back to \texttt{rss} when \texttt{uss} was unavailable. 
The monitor tracked the worker process together with all of its child processes and
recorded the maximum aggregate memory usage observed over time. 
Thus, the reported peak memory excludes the parent controller process itself, 
but includes any subprocesses or solver processes spawned by the worker process. 
Since parsing, domain construction, grounding, preprocessing, 
and WFOMC execution were all carried out within the same worker process, 
the reported peak memory also includes these stages. Memory usage was sampled every 0.05 seconds throughout execution.

We report memory results on three representative \ctwo{} benchmark families from \Cref{sec:c2_benchmarks}: 
$k$-regular graphs, 
$k$-regular $l$-colored graphs, and 
$k$-regular digraphs. 
\Cref{fig:memory_kregular,fig:memory_k_regular_l_colored,fig:memory_k_regular_digraphs} show the results for these benchmarks, respectively.

Across all benchmark families, the same qualitative pattern emerges:
\ouralgo{} consistently maintains a relatively low peak memory usage as the domain size grows, while \textsc{RECURSIVE} shows a much steeper growth in memory usage, and \Fast{} falls in between, with higher memory usage than \ouralgo{} but still significantly below \textsc{RECURSIVE}.
The reason for the higher memory usage of both \textsc{RECURSIVE} and \Fast{} is that they both use the reduction of eliminating counting quantifiers via cardinality constraints.
Cardinality constraints require either multiple calls to the \wfomc{} solver with different weights for interpolation, or a single \wfomc{} call with symbolic weights, both of which can lead to increased memory usage due to the need to store multiple intermediate results or larger symbolic representations.

\begin{figure}[tbp] 
    \centering
    \begin{subfigure}[b]{0.32\textwidth}
        \includegraphics[width=\linewidth]{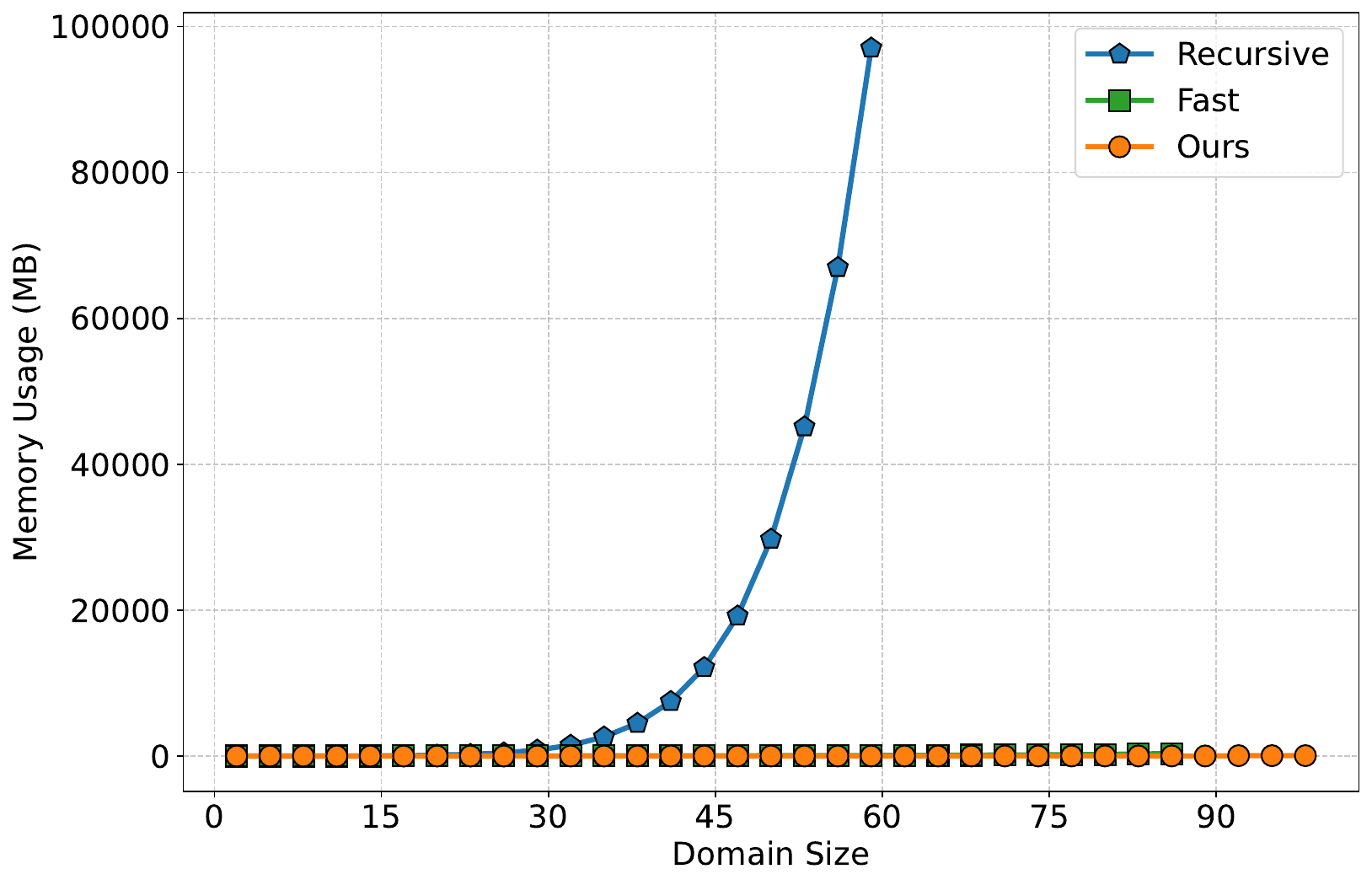}
        \caption{3-regular}
        \label{fig:memory_3_regular}
    \end{subfigure}
    \begin{subfigure}[b]{0.32\textwidth}
        \includegraphics[width=\linewidth]{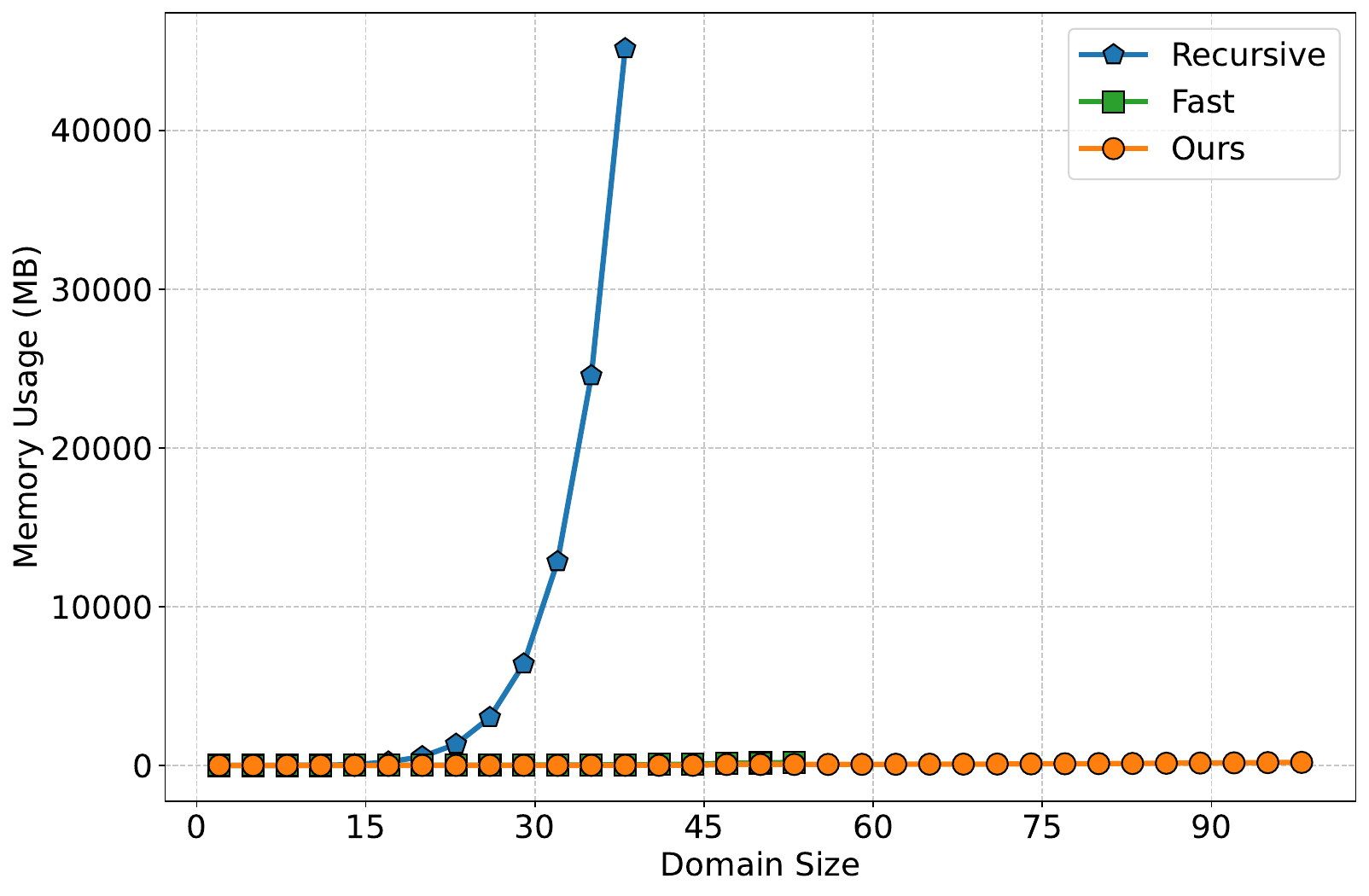}
        \caption{4-regular}
    \end{subfigure}
    \begin{subfigure}[b]{0.32\textwidth}
        \includegraphics[width=\linewidth]{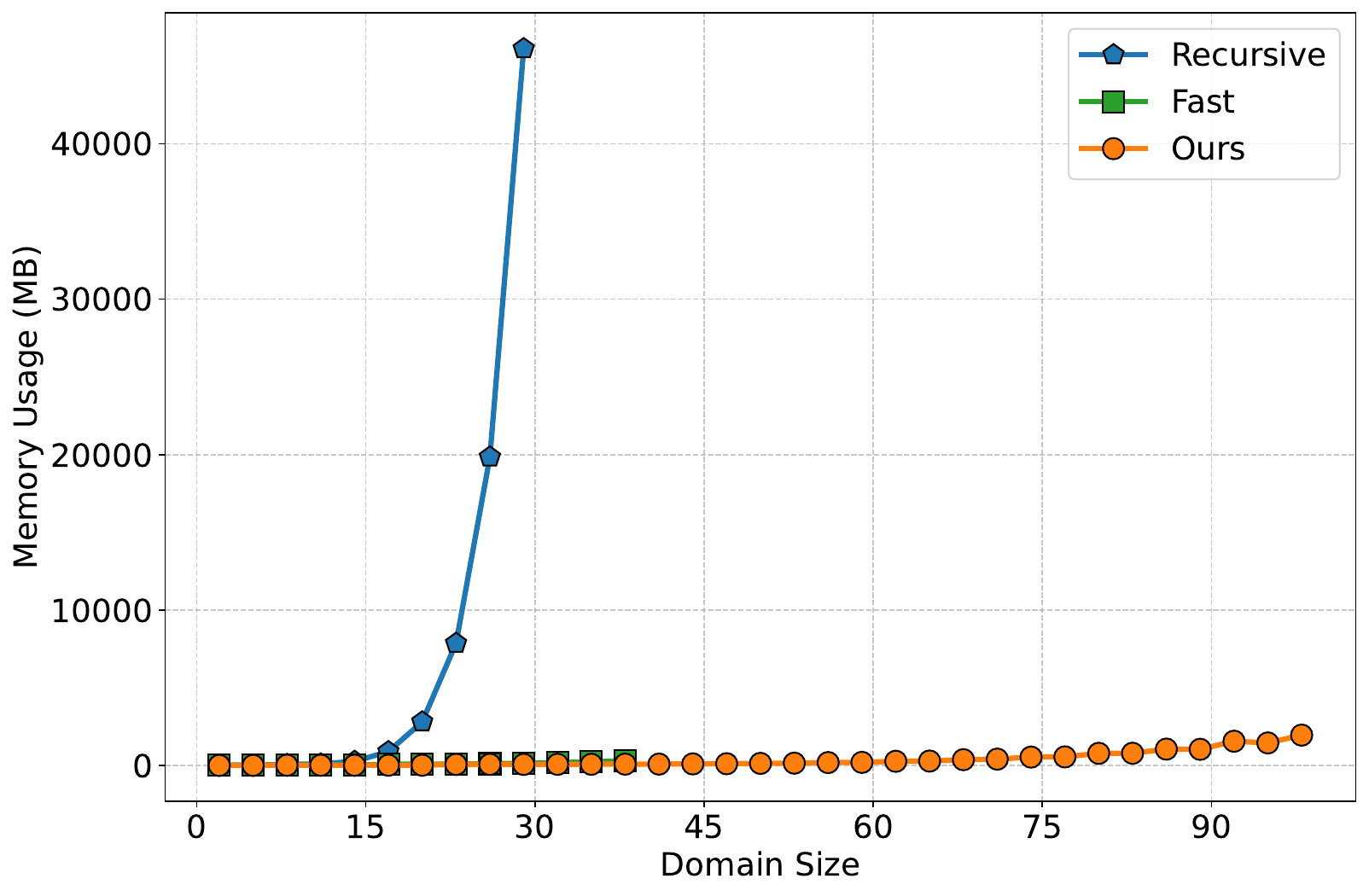}
        \caption{5-regular}
        \label{fig:memory_5_regular}
    \end{subfigure}
    \caption{
        Peak memory usage for counting (a) 3-regular, (b) 4-regular, and (c) 5-regular graphs. 
    }
    \label{fig:memory_kregular}
\end{figure}


\begin{figure}[tbp]
    \centering
    \begin{subfigure}[b]{0.32\textwidth}
        \includegraphics[width=\linewidth]{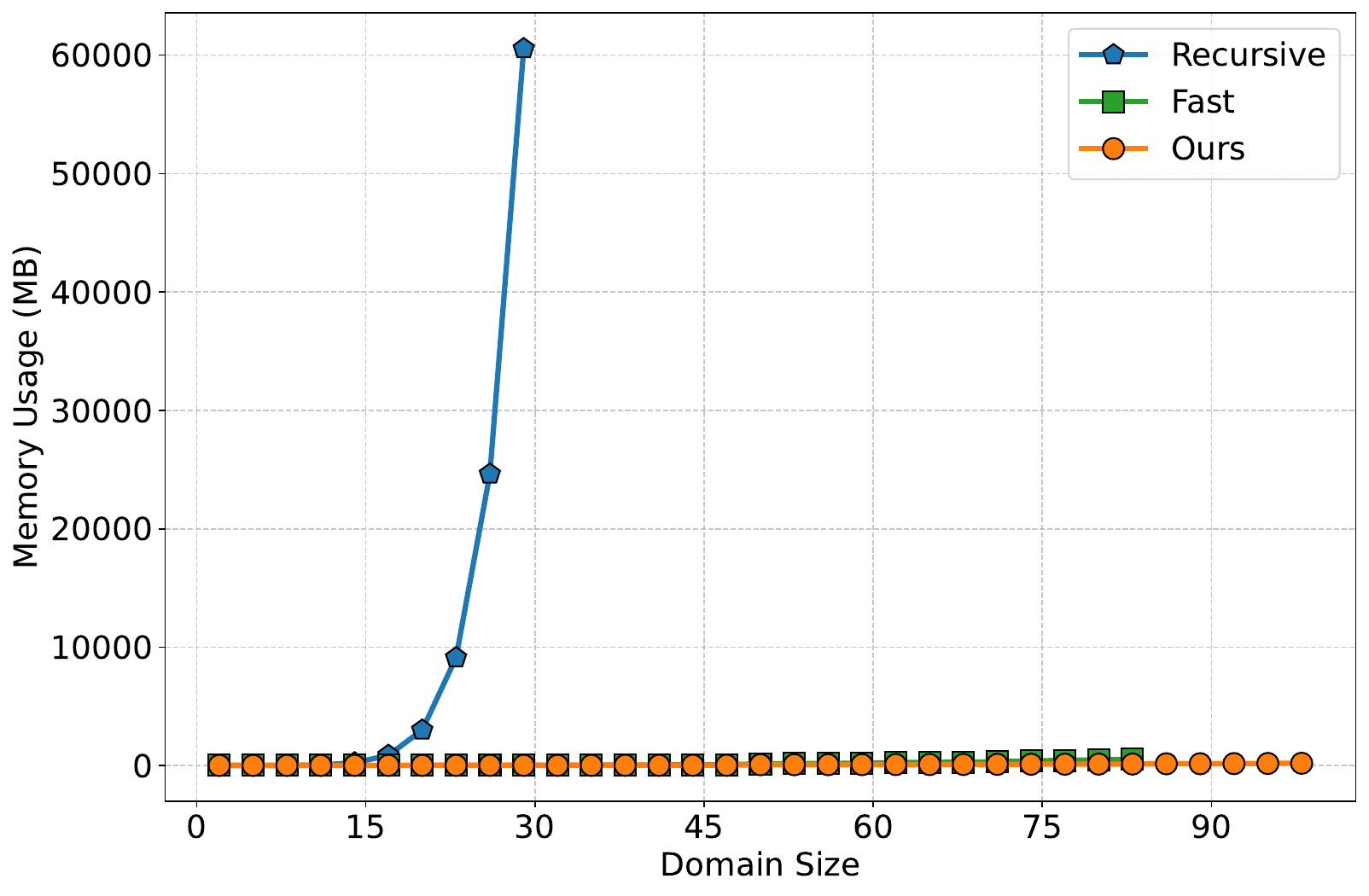}
        \caption{3-regular 2-colored}
    \end{subfigure}
    \begin{subfigure}[b]{0.32\textwidth}
        \includegraphics[width=\linewidth]{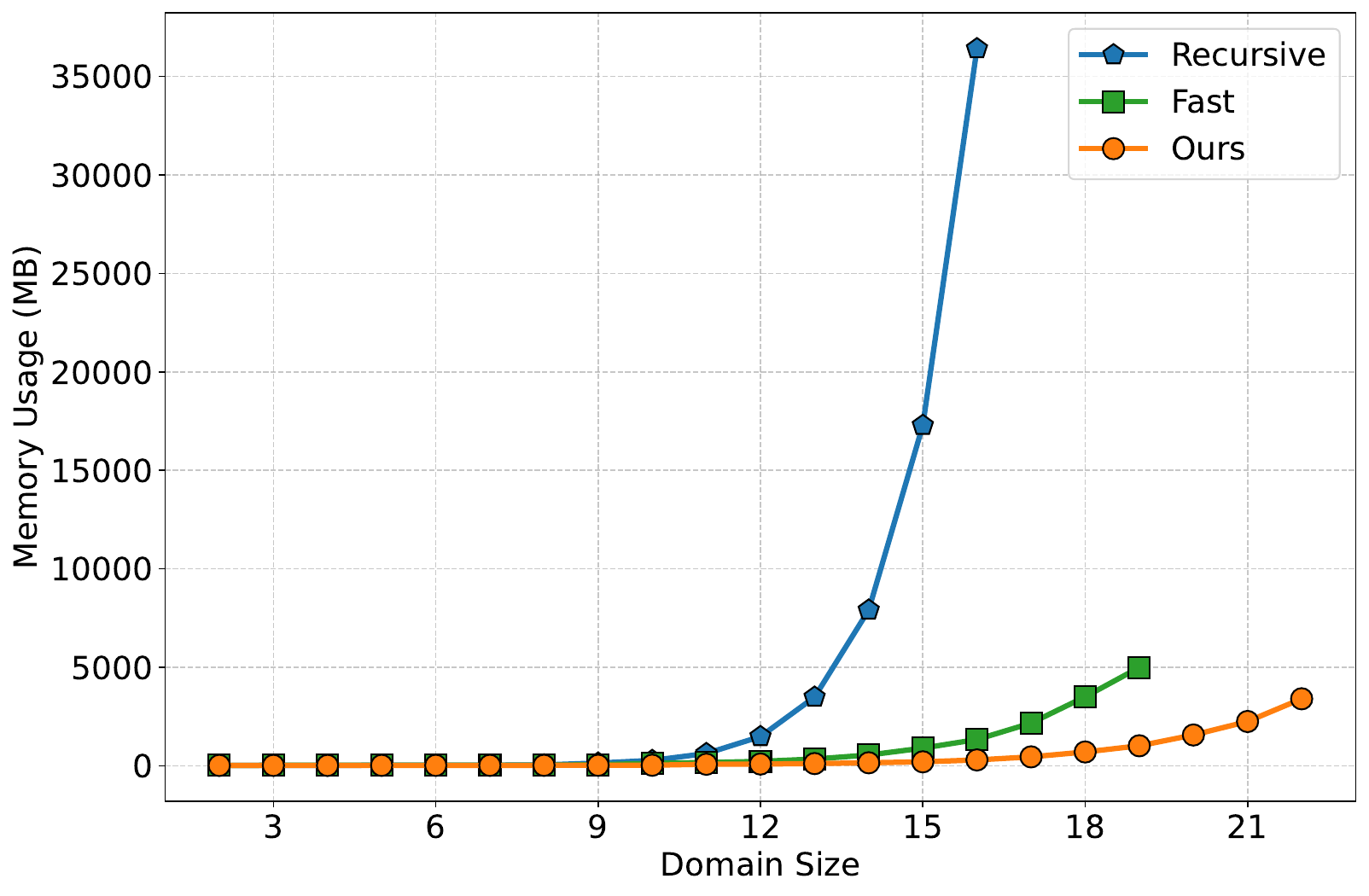}
        \caption{3-regular 4-colored}
    \end{subfigure}
    \begin{subfigure}[b]{0.32\textwidth}
        \includegraphics[width=\linewidth]{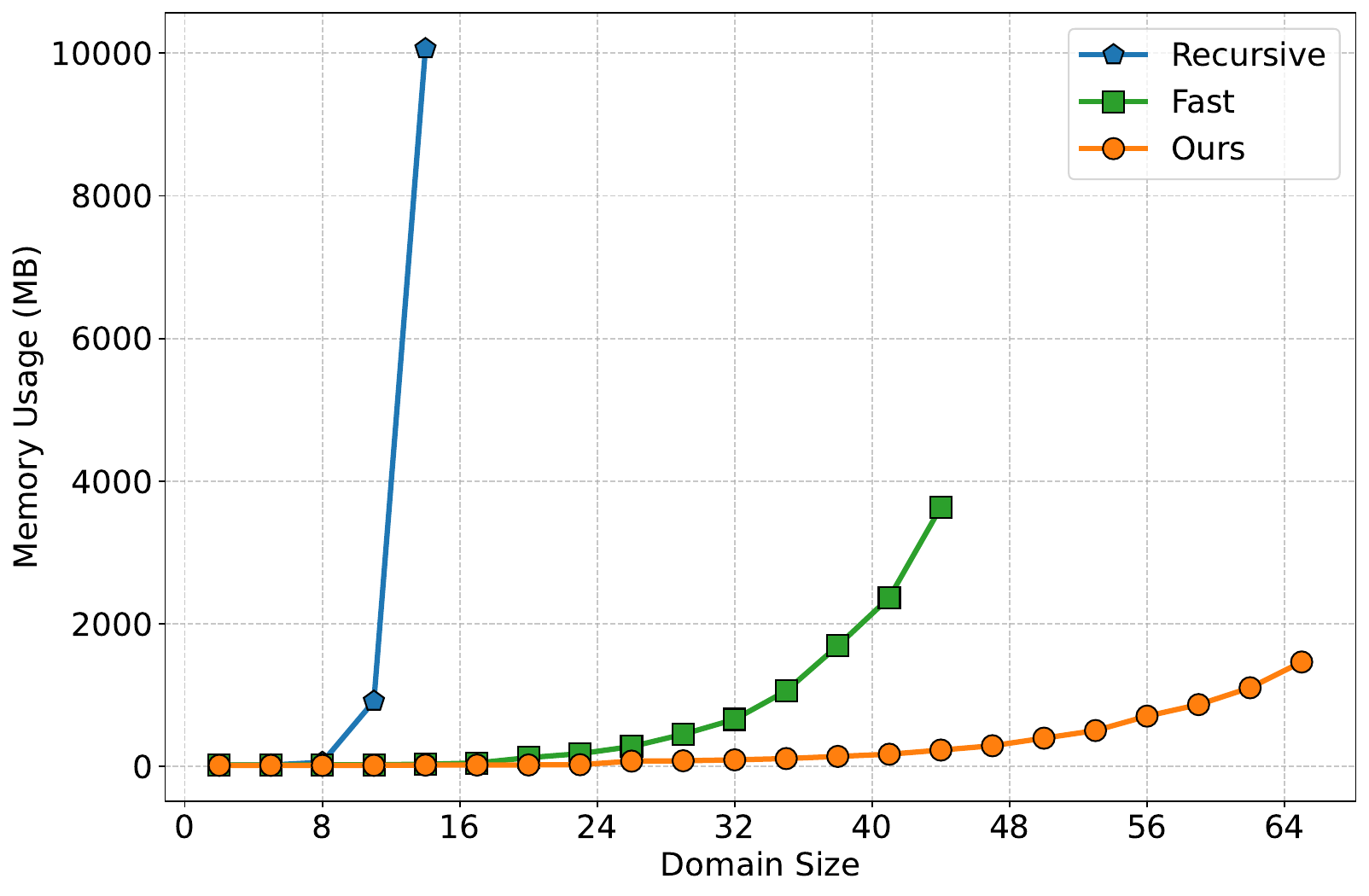}
        \caption{5-regular 2-colored}\label{fig:memory_5_regular_2_colored}
    \end{subfigure}
    \caption{
    Peak memory usage comparison for counting $k$-regular $l$-colored graphs. 
    }
    \label{fig:memory_k_regular_l_colored}
\end{figure}


\begin{figure}[tbp]
    \centering
    \begin{subfigure}[b]{0.32\textwidth}
        \includegraphics[width=\linewidth]{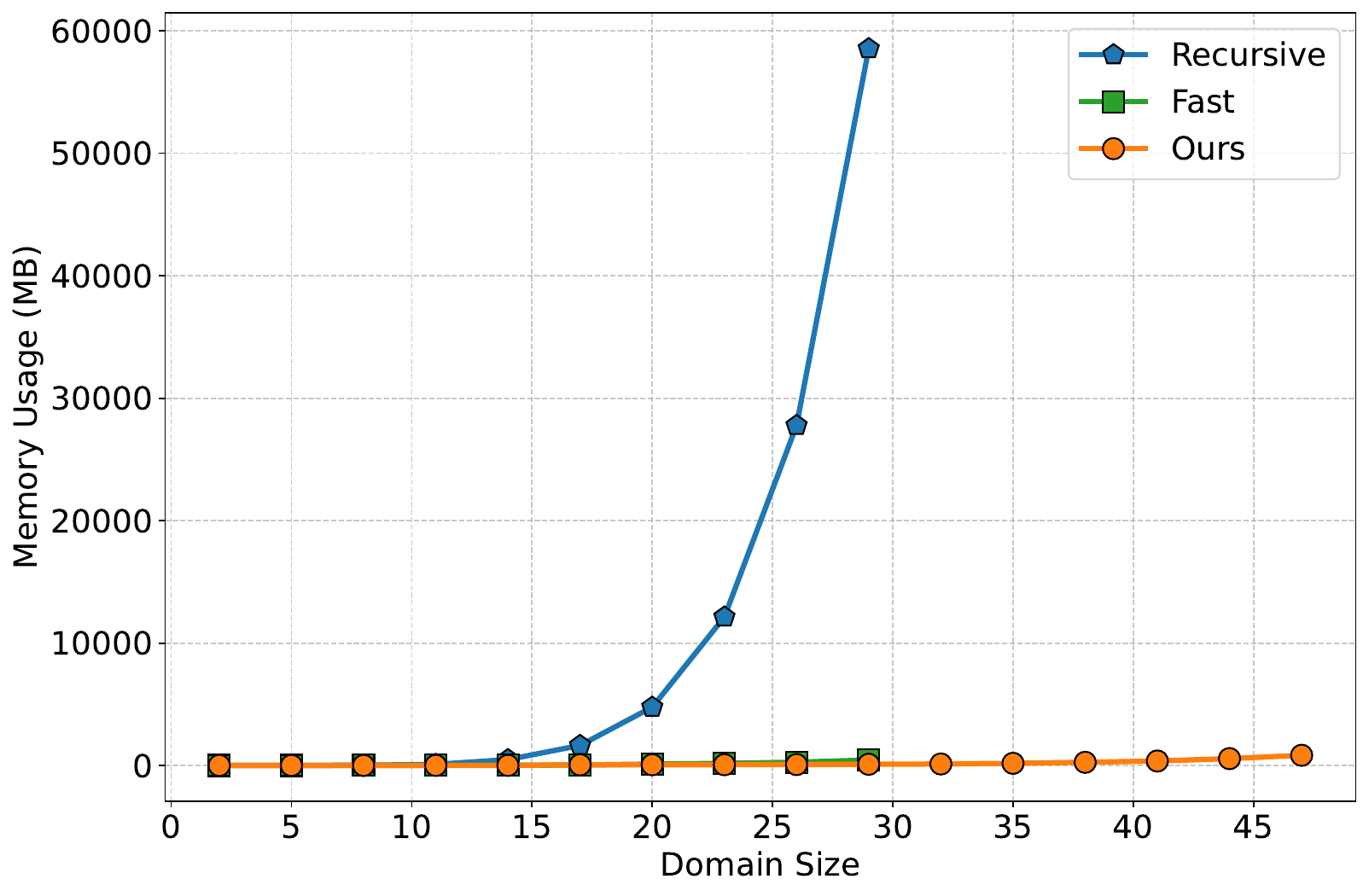}
        \caption{$2$-regular digraphs}\label{fig:memory_2_regular_digraph}
    \end{subfigure}
    \begin{subfigure}[b]{0.32\textwidth}
        \includegraphics[width=\linewidth]{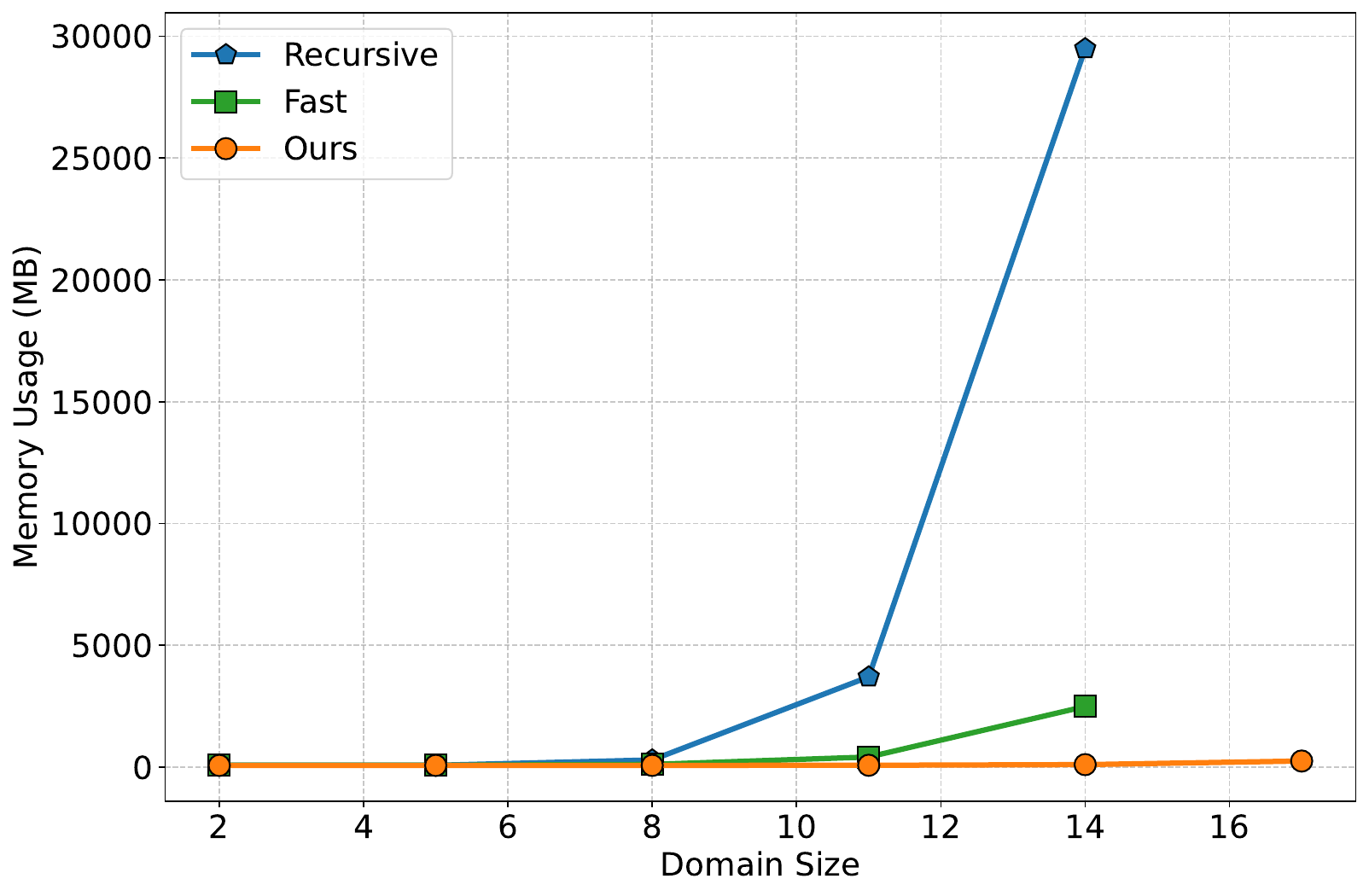}
        \caption{$3$-regular digraphs}\label{fig:memory_3_regular_digraph}
    \end{subfigure}
    \caption{
        Peak memory results for counting (a) $2$-regular and (b) $3$-regular digraphs.
    }\label{fig:memory_k_regular_digraphs}
\end{figure}


\section{Reduction of M-Odd-Degree Graphs}
\label{app:reduction_m_odd_degree}
The initial, more direct formulation of the problem is as follows:
\begin{align*}
    \sentence_{\text{$m$-odd-degree}} = 
    &(\forall x: \neg E(x,x)) \land (\forall x \forall y: E(x,y) \to E(y,x)) \land \\
    &(\forall x:Odd(x) \leftrightarrow \exists^{=1,2} y: E(x,y)) \land (\exists^{=m} x: Odd(x)) \land  (|E| = 2k).
\end{align*}

Then we used the transformations described in \Cref{app:modk_reduction}.
All predicates $A$, $B$, $C$, $U$, and $P$ appearing below are fresh auxiliary predicates introduced during the transformations. 
The transformed sentence is as follows:
\begin{align*}
    \sentence_{\text{$m$-odd-degree}} = 
    &\forall x : \neg E(x,x) \land \forall x \forall y: E(x,y) \to E(y,x) \land \exists^{=m} x: Odd(x) \land \\
    &\exists^{=1} x: U(x) \land \forall x: (P(x) \leftrightarrow (\neg Odd(x) \land A(x) \land C(x))) \land \\
    &\forall x \forall y: (P(x) \land B(x, y) \to U(y)) \land \\
    &\forall x \forall y: (\neg P(x) \to (B(x, y) \leftrightarrow E(x, y))) \land \forall x \exists^{=1,2}y: B(x, y) \land \\
    &\forall x: (Odd(x) \lor A(x))  \land \forall x : (A(x) \lor C(x)) \land |E| = 2k.
\end{align*}

\section{Integer Sequence Analysis of Model Counts of $\sentence_{\text{$m$-odd-degree}}$}
\label{app:sequences}

\begin{table}[htbp]
    \centering
    \includegraphics[width=\linewidth]{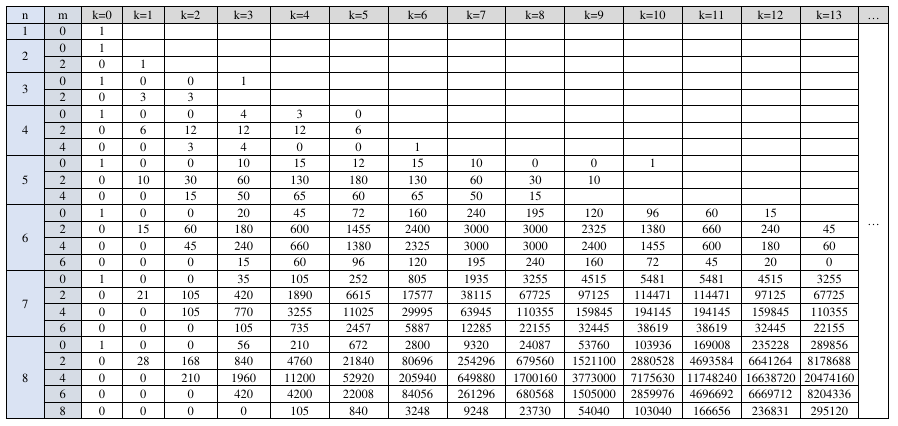}
    \caption{Model counts for $\sentence_{\text{$m$-odd-degree}}$. 
    The table shows the number of simple, undirected graphs for varying domain sizes $n$ and numbers of odd-degree vertices $m$. 
    Each column corresponds to a fixed number of edges $k$, and the value in each cell is the resulting model count for that specific $(n, m, k)$ configuration.
    Blank cells denote zero entries omitted for readability,
    whereas ellipses indicate that additional rows or columns are omitted from display.
    Most of these omitted zeros arise from structurally impossible configurations.
    }
    \label{tab:nmk}
\end{table}

\Cref{tab:nmk} reports the model counts of $\Gamma_{\text{$m$-odd-degree}}$ 
for simple undirected graphs as a function of the domain size $n$, the number $m$ of odd-degree vertices, 
and the number $k$ of edges. Let $T(n,m,k)$ denote the corresponding count. 
Below, we highlight several correspondences between selected slices of $T(n,m,k)$ and integer sequences in the OEIS.

The slice $T(n,2,1)$ counts graphs with exactly one edge. Therefore,
$T(n,2,1)=\binom{n}{2}$, which matches OEIS \href{https://oeis.org/A000217}{A000217}.
Similarly, $T(n,0,3)$ counts simple undirected graphs with three edges and no odd-degree vertices. 
Any such graph must be a triangle. Therefore,
$T(n,0,3)=\binom{n}{3}$,
which matches OEIS \href{https://oeis.org/A000292}{A000292}. More generally, 
the slice $T(n,0,k)$ counts simple graphs on $n$ labeled vertices 
with $k$ edges in which every vertex has even degree. 
This coincides with OEIS \href{https://oeis.org/A058878}{A058878}.
The slice $T(n,2,2)$ counts labeled copies of $P_3$, so
$T(n,2,2)=3\binom{n}{3}$,
which matches OEIS \href{https://oeis.org/A027480}{A027480}. 
Furthermore, $T(n,2k,k)$ counts matchings of size $k$ in the complete graph $K_n$, in agreement with OEIS \href{https://oeis.org/A100861}{A100861}.

Taken together, these examples show that 
the model counts of $\Gamma_{\text{$m$-odd-degree}}$ give rise to several classical graph-enumeration families.

\end{document}